\documentclass[10pt,a4paper]{article}

\usepackage{amsmath,amsgen,latexsym}
\usepackage{amstext,amssymb,amsfonts,latexsym}
\usepackage{theorem}
\usepackage{pifont}
\usepackage[dvips]{graphics,epsfig}
\usepackage[dvips]{graphicx}

 \newcommand{\bs}{\bigskip}
 
 \newcommand{\n}{\noindent}
 \newcommand{\s}{\smallskip}
 \newcommand{\hs}[1]{\hspace*{ #1 mm}}
 \newcommand{\vs}[1]{\vspace*{ #1 mm}}

 \newcommand{\setempty}{\mathrm{\O}}
 
 \newcommand{\nat}{\mathbb{N}}
 
 \newcommand{\integer}{\mathbb{Z}}

 \newcommand{\co}{\mathrm{co}\mbox{-}}

 \newcommand{\ie}{\textrm{i.e.},\hspace*{2mm}}
 \newcommand{\eg}{\textrm{e.g.},\hspace*{2mm}}
 
 \newcommand{\etalc}{\textrm{et al.}}

 \newcommand{\CC}{{\cal C}}

 \newcommand{\PP}{{\cal P}}
 
  \newcommand{\RR}{{\cal R}}

 \newcommand{\dl}{\mathrm{L}}
 
 \newcommand{\p}{\mathrm{P}}
 \newcommand{\np}{\mathrm{NP}}
 \newcommand{\bpp}{\mathrm{BPP}}
 \newcommand{\pp}{\mathrm{PP}}
 \newcommand{\pspace}{\mathrm{PSPACE}}

 \newcommand{\tally}{\mathrm{TALLY}}
 \newcommand{\sparse}{\mathrm{SPARSE}}

 \newcommand{\reg}{\mathrm{REG}}
 \newcommand{\cfl}{\mathrm{CFL}}
 
 \newcommand{\dcfl}{\mathrm{DCFL}}

\theoremstyle{plain}
\theoremheaderfont{\bfseries}
 \newtheorem{theorem}{Theorem}[section]
 \newtheorem{lemma}[theorem]{Lemma}
 \newtheorem{proposition}[theorem]{Proposition}
 \newtheorem{corollary}[theorem]{Corollary}

 {\theorembodyfont{\rmfamily}
  }
 {\theorembodyfont{\rmfamily} \newtheorem{example}[theorem]{Example}}
 {\theorembodyfont{\rmfamily} }

 \newtheorem{claim}{{\it Claim}}

 \newenvironment{proof}{\par \noindent
            {\bf Proof. \hs{2}}}{\hfill$\Box$ \vspace*{3mm}}

 \newenvironment{proofof}[1]{\vspace*{5mm} \par \noindent
         {\bf Proof of #1.\hs{1}}}{\hfill$\Box$ \vspace*{3mm}}

 \newcommand{\ceilings}[1]{\lceil #1 \rceil}
 \newcommand{\floors}[1]{\lfloor #1 \rfloor}

 \newcommand{\deltap}[1]{\Delta^{\mathrm{P}}_{#1}}
 \newcommand{\sigmap}[1]{\Sigma^{\mathrm{P}}_{#1}}
 \newcommand{\pip}[1]{\Pi^{\mathrm{P}}_{#1}}
 \newcommand{\thetap}[1]{\Theta^{\mathrm{P}}_{#1}}
 \newcommand{\ph}{\mathrm{PH}}

\newif\ifnotesw\noteswtrue
   {\ifnotesw\marginpar[\hfill\(\top\)]{\(\top\)}\fi}%
      {\ifnotesw\marginpar[\hfill\(\bot\)]{\(\bot\)}\fi}
      
\newcommand{\mnote}[1]%
   {\ifnotesw\marginpar%
	  [{\scriptsize\begin{minipage}[t]{\marginparwidth}
	  \raggedleft#1%
		  \end{minipage}}]%
	  {\scriptsize\begin{minipage}[t]{\marginparwidth}
	  \raggedright#1%
		  \end{minipage}}%
    \fi}

\newcommand{\ignore}[1]{}

\newcommand{\track}[2]{[\:\begin{subarray}{c} #1 \\%
      #2 \end{subarray} ]}

\newcommand{\cent}{|\!\! \mathrm{c}}
\newcommand{\dollar}{\$}

\newcommand{\bcfl}{\mathrm{BCFL}}

 \newcommand{\mlow}{\mathrm{low}_{m}}
 
 \newcommand{\bttlow}{\mathrm{low}_{btt}}
 \newcommand{\Tlow}{\mathrm{low}_{T}}

 \newcommand{\cflmvt}{\mathrm{CFLMV_t}}
 \newcommand{\cflmv}{\mathrm{CFLMV}}

 \newcommand{\sigmacfl}[1]{\Sigma^{\mathrm{CFL}}_{ #1 }}
 \newcommand{\picfl}[1]{\Pi^{\mathrm{CFL}}_{ #1 }}
 \newcommand{\deltacfl}[1]{\Delta^{\mathrm{CFL}}_{ #1 }}
 \newcommand{\relsigmacfl}[2]{\Sigma^{\mathrm{CFL}, #2 }_{ #1 }}
 \newcommand{\relpicfl}[2]{\Pi^{\mathrm{CFL}, #2 }_{ #1 }}

 \newcommand{\sigmacflm}[1]{\Sigma^{\mathrm{CFL}}_{m,  #1 }}
 \newcommand{\picflm}[1]{\Pi^{\mathrm{CFL}}_{m,  #1 }}

  \newcommand{\sigmacflt}[1]{\Sigma^{\mathrm{CFL}}_{ #1 }}
 \newcommand{\picflt}[1]{\Pi^{\mathrm{CFL}}_{ #1 }}
 \newcommand{\deltacflt}[1]{\Delta^{\mathrm{CFL}}_{ #1 }}
 \newcommand{\relsigmacflt}[2]{\Sigma^{\mathrm{CFL}, #2 }_{ #1 }}
 \newcommand{\relpicflt}[2]{\Pi^{\mathrm{CFL}, #2 }_{ #1 }}
 \newcommand{\reldeltacflt}[2]{\Delta^{\mathrm{CFL}, #2 }_{ #1 }}

 \newcommand{\sigmanfa}[1]{\Sigma^{\mathrm{NFA}}_{m,  #1 }}
 \newcommand{\pinfa}[1]{\Pi^{\mathrm{NFA}}_{m,  #1 }}
 
 \newcommand{\relsigmanfa}[2]{\Sigma^{\mathrm{NFA}, #2 }_{m,  #1 }}
 \newcommand{\relpinfa}[2]{\Pi^{\mathrm{NFA}, #2 }_{m,  #1 }}

 \newcommand{\nc}[1]{\mathrm{NC}^{ #1 }}
 \newcommand{\ac}[1]{\mathrm{AC}^{ #1 }}
 \newcommand{\sac}[1]{\mathrm{SAC}^{ #1 }}
 \newcommand{\tc}[1]{\mathrm{TC}^{ #1 }}
 
 \newcommand{\bhcfl}{\mathrm{BHCFL}}
 \newcommand{\cflh}{\mathrm{CFLH}}
 \newcommand{\bpcfl}{\mathrm{BPCFL}}
 \newcommand{\pcfl}{\mathrm{PCFL}}

 \newcommand{\paritynfa}{\oplus\mathrm{NFA}}

 \newcommand{\dfa}{\mathrm{DFA}}
 \newcommand{\nfa}{\mathrm{NFA}}

\begin{document}

\pagestyle{plain}
\setcounter{page}{1}

\begin{center}
{\Large {\bf Oracle Pushdown Automata, Nondeterministic \vs{-1}\s\\
Reducibilities,  and the CFL Hierarchy \s\\ Over the
Family of Context-Free Languages\footnote{A short extended abstract appeared under a slightly different title in the Proceedings of the 40th International Conference on Current Trends in Theory and Practice of Computer Science (SOFSEM 2014), High Tatras, Slovakia, January 25--30, 2014, Lecture Notes in Computer Science, Springer-Verlag, vol.8327, pp.514--525, 2014.}}} \bs\s\\

{\sc Tomoyuki Yamakami}\footnote{Present Affiliation: Department of Information Science, University of Fukui, 3-9-1 Bunkyo, Fukui 910-8507, Japan} \bs\\
\end{center}


\begin{abstract}
\n
To expand a fundamental theory of context-free languages, we equip nondeterministic one-way pushdown automata with additional oracle mechanisms, which naturally induce various nondeterministic reducibilities among formal languages.
As a natural restriction of NP-reducibility, we introduce a notion of many-one CFL reducibility and conduct a ground work to formulate a coherent framework for further  expositions. Two more powerful reducibilities---bounded truth-table and Turing CFL-reducibilities---are also discussed in comparison. The Turing CFL-reducibility, in particular, helps us introduce an exquisite hierarchy, called the CFL hierarchy, built over the family CFL of context-free languages. For each level of this new hierarchy, its basic structural properties are proven and three alternative characterizations are presented. The second level is not included in $\nc{2}$ unless NP equals $\nc{2}$. The first and second levels of the hierarchy are proven to be different. The rest of the hierarchy (more strongly, the Boolean hierarchy built over each level of the CFL hierarchy) is also infinite unless the polynomial (time) hierarchy over NP  collapses. This follows from a characterization of the Boolean hierarchy over the $k$th level of the polynomial hierarchy  in terms of the Boolean hierarchy over the $k+1$st level of the CFL hierarchy using logarithmic-space many-one reductions.
Similarly,  the $O(\log{n})$ query bounded complexity class $\thetap{k}$ is related to the closure of the $k$th level of the CFL hierarchy under logarithmic truth-table reductions.  We also argue that the CFL hierarchy coincides with a hierarchy over CFL built by application of many-one CFL-reductions.
We show that BPCFL---a bounded-error probabilistic version of CFL---is not included in CFL even in the presence of advice. Employing a known circuit lower bound and a switching lemma, we exhibit a relativized world where BPCFL is not located within the second level of the CFL hierarchy.

\s

\n{\bf Keywords:}
{regular language, context-free language, pushdown automaton, oracle, many-one reducibility, Turing reducibility, truth-table reducibility, CFL hierarchy, polynomial hierarchy, advice, Dyck language}
\end{abstract}

\sloppy
\section{Backgrounds and Main Themes}\label{sec:introduction}

For the development of a tidy theory of $\np$-completeness, a fundamental notion of {\em reducibility} has long played an essential role, where a  language is generally said to be ``reducible'' to another language when there is an appropriate algorithm that determines  any membership question to the former language by making certain membership queries to the latter  language. In the early 1970s, Cook \cite{Coo71} outlined the concept of  \emph{polynomial-time Turing reducibility} to demonstrate that the language SAT, composed of satisfiable propositional formulas, is
computationally hard.
His reducibility is founded on a model of \emph{oracle Turing machine} and it provides a useful tool to identify the most difficult languages in $\np$.
Karp \cite{Kar72}, in contrast, used \emph{polynomial-time  many-one reducibility} to present a number of $\np$-complete languages.
Since their works, various forms of polynomial-time reducibility have emerged \cite{BLS84,LLS75}.
Besides many-one and Turing reducibilities, for example, typical reducibilities in use today in computational complexity theory include conjunctive, disjunctive, truth-table, bounded truth-table, and query-bounded Turing reducibilities, which are obtained by imposing appropriate restrictions on the functionality of oracle mechanism of underlying Turing machines. Primarily, those reducibilities were defined by deterministic machines but they have been naturally expanded to nondeterministic reducibilities. Meyer and Stockmeyer \cite{MS72} concerned with nondeterministic polynomial-time reducibility and used it to build the so-called \emph{polynomial (time) hierarchy} over $\np$ (see also \cite{Sto77,Wra77}). A study on reducibilities have lead to promote the  understandings of the structure of $\p$, $\np$, and beyond, the polynomial hierarchy.

Various oracle mechanisms have provided us with a useful means to study {\em relativizations} of associated language families and such a relativization offers a ``relativized world'' in which  certain desirable properties are all met at once.
For issues not settled by the current knowledge of us, we often resort to a relativization, which
helps us discuss the existence of various relativized worlds in which a certain relationship among target language families either holds or fails.
Concerning the famous P=?NP problem, for instance, Baker, Gill, and Solovay \cite{BGS75}  constructed two conflicting relativized worlds where $\p=\np$ and $\p\neq\np$ indeed happen. These contradictory results suggest that a solution (that is, a proof) to the P=?NP problem
must be ``unrelativizable.'' Yao \cite{Yao85} (and later H{\aa}stad \cite{Has86}) presented a relativized world where the polynomial hierarchy is indeed an infinite hierarchy.

Away from standard complexity-theoretical subjects, we shift our attention to  a theory of formal languages and automata. By providing a solid foundation for various notions of reducibility and their associated relativizations, we intend to lay out a framework for a {\em structural complexity theory of formal languages and automata},  which enables us to conduct extensive studies on various structural complexity issues for formal languages.

Of many languages, we are particularly interested in {\em context-free languages}, which are characterized by context-free grammars or one-way nondeterministic pushdown automata (or npda's, hereafter).  The context-free languages are inherently nondeterministic. In light of the fact that the notion of nondeterminism appears naturally in real life, it has become a key to many fields of computer science.
The family $\cfl$ of context-free languages has proven to be a fascinating subject, simply because every language in $\cfl$ behaves quite differently from the corresponding nondeterministic polynomial-time class $\np$. Whereas $\np$ is closed under any Boolean operations except for complementation, $\cfl$ is not even closed under intersection. This non-closure property is caused by  the lack of flexibility in the use of  memory storage by an underlying model of npda. On the contrary,  a  restricted use of memory helps us prove a separation between the first and the second levels of the Boolean hierarchy $\{\cfl_{k}\mid k\geq1\}$ built over $\cfl$  by applying Boolean operations (intersection and union) alternatingly to $\cfl$ (whose variant was discussed in \cite{YK13}).
Moreover, we can prove that the family of languages
$\cfl(k)$ composed of intersections of $k$ context-free languages truly forms an infinite hierarchy \cite{LW73}.
Such an architectural restriction sometimes becomes a crucial issue in certain applications of, for example, one-way probabilistic pushdown automata (or ppda's). It is known in \cite{HS10} that bounded-error ppda's cannot, in general, amplify their success probabilities.

A most simple type of the aforementioned reducibilities is probably {\em many-one reducibility} and, by adopting the existing formulation of this reducibility, we intend to bring a notion of nondeterministic many-one reducibility into context-free languages under the name of {\em many-one CFL-reducibility}. Symbolically, we write $\cfl_{m}^{A}$ to express the family of languages that are many-one $\cfl$-reducible to a given oracle $A$. With a similar flavor,  Reinhardt \cite{Rei90}
considered  many-one reductions, which are induced by nondeterministic finite automata (or nfa's) with no memory device.
Notice that nondeterministic reducibility generally does not admit the  \emph{transitivity property}. (For this reason, such reducibility might have been called a ``quasi-reducibility'' if the transitive property is a prerequisite for a reducibility notion.)
Owing mostly to a unique architecture of npda's, our reducibility exhibits quite distinctive features; for instance,
the family $\cfl$ is not closed under the many-one $\cfl$-reducibility (that is, $\cfl_{m}^{\cfl}\neq \cfl$).
This non-closure property allures us to study the language family $\cfl_{m[k]}^{\cfl}$, whose elements are obtained by the $k$-fold application of many-one $\cfl$-reductions to context-free languages.  As is shown in Section \ref{sec:many-one-reduction}, the language family $\cfl_{m[k]}^{\cfl}$ coincides with $\cfl_{m}^{\cfl(k)}$.

We further discuss two more powerful reducibilities in popular use: \emph{bounded truth-table and Turing CFL-reducibilities}, which are based on appropriately defined oracle npda's.
In particular, the Turing $\cfl$-reducibility, which allows any underlying reduction npda to make adaptive queries, introduces a hierarchy $\{\deltacflt{k},\sigmacflt{k},\picflt{k}\mid k\geq1\}$, where the first level contains $\deltacfl{1}=\dcfl$, $\sigmacfl{1}=\cfl$, and $\picfl{1}=\co\cfl$, analogous to the polynomial hierarchy.
We succinctly call this hierarchy the {\em CFL hierarchy}, which turns out to be quite useful in classifying the computational complexity of a certain group of languages.  As a quick example, the languages $Dup_2 = \{xx\mid x\in\{0,1\}^*\}$ and $Dup_3 =\{xxx\mid x\in\{0,1\}^*\}$, which are known to be outside of $\cfl$, fall into the second level $\sigmacflt{2}$ of the CFL hierarchy. A simple matching language $Match=\{x\# w\mid \exists u,v\,[w=uxv]\,\}$ is also in $\sigmacflt{2}$.
Two more languages $Sq=\{0^n1^{n^2}\mid n\geq1\}$ and $Prim=\{0^n\mid \text{ $n$ is a prime number }\}$ respectively belong to $\sigmacflt{2}$ and $\picflt{2}$. A slightly more complex language $MulPrim=\{0^{mn}\mid \,\text{$m$ and $n$ are prime numbers}\,\}$ is a member of  $\sigmacflt{3}$.
The first and second levels of the $\cfl$ hierarchy are proven to be different; more strongly, we can prove that  $\sigmacflt{2}\nsubseteq\sigmacflt{1}/n$, where $\sigmacflt{1}/n$ is a non-uniform version of $\sigmacflt{1}$, defined in  \cite{TYL10} and further explored in \cite{Yam10}.
As is shown in Section \ref{sec:structure-hierarchy}, the whole hierarchy is included in $\mathrm{DSPACE}(O(n))$.
Regarding the aforementioned language families $\cfl(k)$ and $\cfl_{k}$, we can show in Section \ref{sec:Turing-reduction}  that the families $\cfl(\omega)= \bigcup_{k\geq1}\cfl(k)$ and $\bhcfl = \bigcup_{k\geq1}\cfl_{k}$ (called the Boolean hierarchy over $\cfl$) belong to the second level $\sigmacflt{2}\cap \picflt{2}$ of the CFL hierarchy, from a fact that $\cfl(\omega) \subseteq \bhcfl$. Notice that Wotschke \cite{Wot78} demonstrated the separation of $\cfl(\omega)\neq \bhcfl$. Moreover, we show that $\cfl_{m}^{\cfl(\omega)}$  is located within $\sigmacflt{3}$.
Despite obvious similarities between their definitions,
the $\cfl$ hierarchy and the polynomial hierarchy are quite different in nature.  Because of npda's architectural restrictions, ``standard'' techniques of simulating a two-way Turing machine, in general, do not apply; hence, we need to develop new simulation techniques for npda's.

Throughout this paper, we employ three simulation techniques to obtain some of the aforementioned results. The first technique is of guessing and verifying a {\em stack history} to eliminate a use of stack, where a stack history means a series of consecutive stack operations made by an underlying npda. The second technique is applied to the case of simulating two or more tape heads by a single tape head. To adjust the different head speeds, we intentionally insert extra dummy symbols to generate a single query word so that an oracle can ignore them when it accesses
the query word. The last technique is to generate a string that encodes a computation path generated by a nondeterministic machine. All the techniques are explained in details in Sections \ref{sec:many-one-reduction}--\ref{sec:tt-reduction}.  Those simulation techniques actually make it possible to obtain three alternative characterizations of the $\cfl$ hierarchy later in Section \ref{sec:structure-hierarchy}.

Reinhardt \cite{Rei90} related the aforementioned hierarchy of his to another hierarchy defined by alternating pushdown automata and he gave a characterization of the polynomial hierarchy in terms of this alternating hierarchy using logarithmic-space (or log-space) many-one reductions.
Using an argument similar to his, we can establish in Section \ref{sec:close-relation} an exact characterization of the $e$th level of the Boolean hierarchy over the $k$th level $\sigmap{k}$ of the polynomial hierarchy in terms of the corresponding $e$th level of the Boolean hierarchy over the $k+1$st level of the CFL hierarchy. Moreover, we give a new characterization of $\thetap{k}$ (\ie Wagner's \cite{Wag90} notation for $\p_{T}(\sigmap{k-1}[O(\log{n})])$, where ``$[O(\log{n})]$'' indicates at most $O(\log{n})$ adaptive oracle queries are allowed in an entire computation tree)
in terms of the $k$th level of the $\cfl$ hierarchy using log-space truth-table reductions. As an immediate consequence, all levels of the Boolean hierarchy over each level of the $\cfl$ hierarchy are different unless the polynomial hierarchy collapses. With respect to the circuit complexity class $\nc{2}$, we show that $\sigmacfl{2}\nsubseteq \nc{2}$ unless $\np=\nc{2}$.

Another relevant notion induced by reducibility is a {\em relativization} of language families. In Section \ref{sec:BPCFL},  we construct  a recursive oracle for which the family $\bpcfl$ of languages recognized by bounded-error ppds's is not included within the second level of the $\cfl$ hierarchy.
(Of course, there also exists an obvious oracle that makes this inclusion hold.) This separation result contrasts a well-known fact that $\bpp$ is included in $\sigmap{2}\cap\pip{2}$ in any relativized world. To deal with oracle-dependent languages in a relativized $\cfl$ hierarchy, we first characterize them using bounded-depth Boolean circuits of alternating ORs and ANDs.
Our proof relies on a special form of the well-known {\em switching lemma} \cite{Bea90}, in which a circuit of OR of ANDs can be transformed into another equivalent circuit of AND of ORs by assigning probabilistically $0$ and $1$ to input variables. In the unrelativized world, however, we prove that $\bpcfl\nsubseteq\cfl/n$. This separation extends a known result of \cite{HS10} that $\bpcfl\nsubseteq\cfl$.

\begin{figure}[t]
 \begin{center}
 \includegraphics[width=10.0cm]{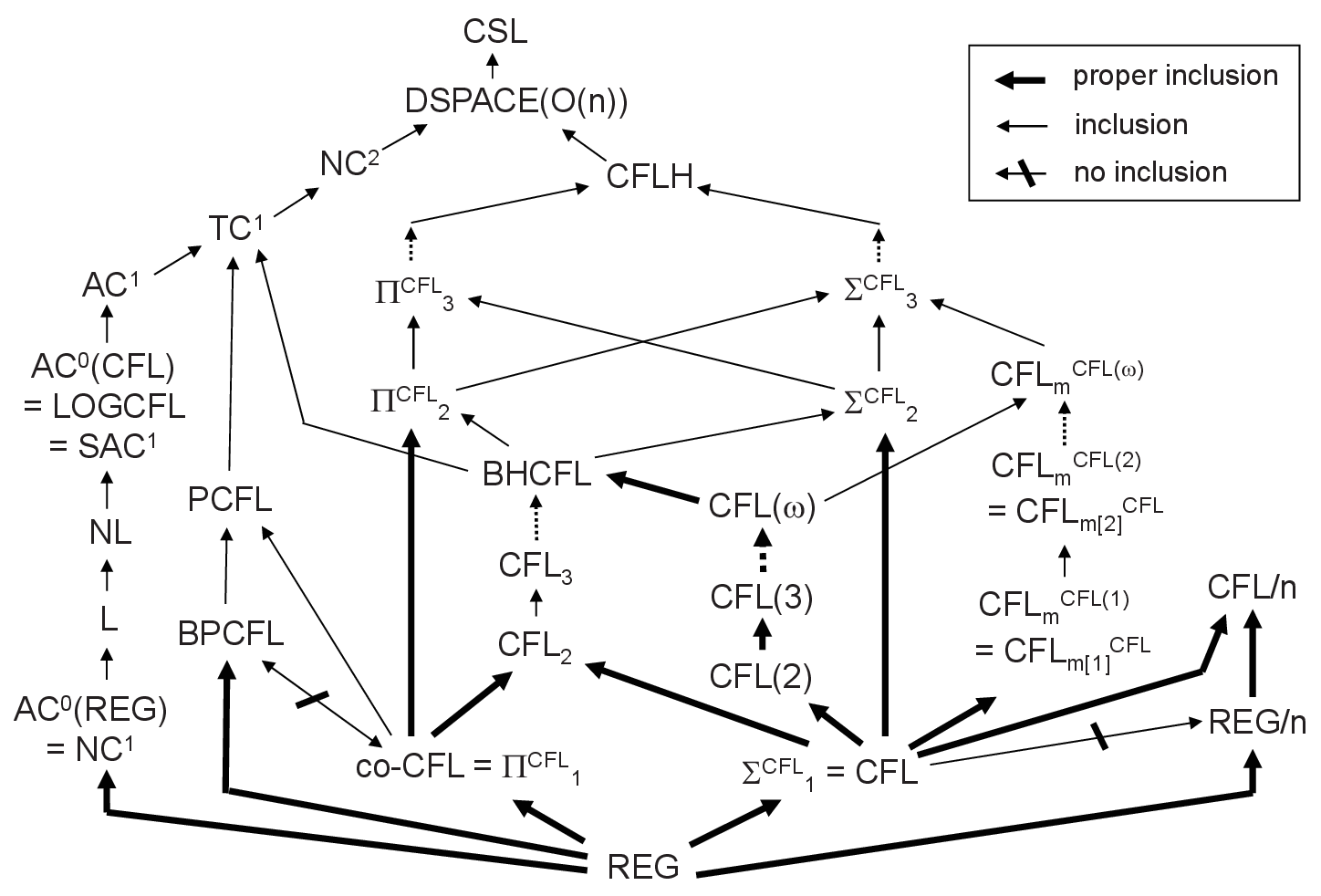}
 \end{center}
\caption{Hasse diagram of inclusion relations among language families}\label{fig:inclusion-map}
\end{figure}

A Hasse diagram in Fig.\ref{fig:inclusion-map} summarizes some of the inclusion relationships  among language families discussed so far. The notation $\cflh$ in the figure denotes the union $\bigcup_{k\geq1}(\sigmacfl{k}\cup\picfl{k})$.

Although most results in this paper are embryonic, we strongly believe that these results would pave a long but straight road to more exciting discoveries in a structural complexity theory of formal languages and automata. For the help to the avid reader, we provide in Section \ref{sec:open-problem} a short list of challenging problems that are left unsolved in its previous sections.

\section{A Preparation for Our Expositions}\label{sec:preliminaries}

We will briefly explain basic notions and notations that help the reader go  through the subsequent sections. Generally, we will follow the existing terminology in a field of formal languages and automata. However, the reader who is familiar with computational complexity theory needs extra attentions to  ceratin notations (for instance, $\cfl(k)$ and $\cfl_{k}$) that are used in quite different ways.

\subsection{Alphabets, Strings, and Languages}\label{sec:alphabet-string}

Given a finite set $A$, the notation $\parallel\! A\!\parallel$ expresses the number of elements in $A$. Let $\nat$ be the set of all {\em natural numbers} (\ie nonnegative integers) and set $\nat^{+}=\nat-\{0\}$. Given two integers $m$ and $n$ with $m\leq n$, $[m,n]_{\integer}$ denotes an integer interval $\{m,m+1,m+2,.\ldots,n\}$. In particular, we abbreviate $[1,n]_{\integer}$ for any number $n\in\nat^{+}$ as $[n]$.
The term ``polynomial'' always means a polynomial on $\nat$ with coefficients of non-negative integers.  In particular, a {\em linear polynomial} is of the form $ax+b$ with $a,b\in\nat$. The notation $A-B$ for two sets $A$ and $B$ indicates the {\em difference} $\{x\mid x\in A, x\not\in B\}$ and $\PP(A)$ denotes the {\em power set} of $A$; that is, the collection of all subsets of $A$.

An {\em alphabet} is a nonempty finite set $\Sigma$ and its elements are called  {\em symbols}. A {\em string}  $x$ over $\Sigma$ is a finite series of symbols chosen from $\Sigma$
and its {\em length}, denoted $|x|$, is the total  number of symbols in $x$. The {\em empty string} $\lambda$ is a special string whose length is zero.
Given a string $x=x_1x_2\cdots x_{n-1}x_{n}$ with $x_i\in\Sigma$, $x^R$ represents the {\em reverse} of $x$, defined by $x^R=x_{n}x_{n-1}\cdots x_2x_1$.
We set $\overline{0}=1$ and $\overline{1}=0$; moreover, for any string $x x=x_1x_2\cdots x_n$ with $x_i\in\Sigma$, $\overline{x}$ denotes $\overline{x_1}\,\overline{x_2}\cdots\overline{x_n}$.
To treat a pair of strings, we adopt a {\em track notation} $\track{x}{y}$ from \cite{TYL10}.  For two symbols $\sigma$ and $\tau$, the notation $\track{\sigma}{\tau}$ expresses a new symbol and, for two strings $x=x_1x_2\cdots x_n$ and $y=y_1y_2\cdots y_n$ of length $n$, $\track{x}{y}$ denotes a string $\track{x_1}{y_1}\track{x_2}{y_2}\cdots \track{x_n}{y_n}$ of length $n$. Since this notation can be seen as a column vector of dimension $2$, we can extend it to a {\em $k$-track notation}, denoted conveniently by $[x_1,x_2,\ldots,x_k]^{T}$, where ``$T$'' indicates a transposed vector.

A collection of strings over $\Sigma$ is  a {\em language} over $\Sigma$.
A set $\Sigma^k$, where $k\in\nat$, consists only of strings of length $k$. In particular, $\Sigma^0$  indicates  the set $\{\lambda\}$.  The {\em Kleene closure} $\Sigma^*$ of $\Sigma$ is the infinite union  $\bigcup_{k\in\nat}\Sigma^k$.  Similarly, the notation $\Sigma^{\leq k}$ is used to indicate the set $\bigcup_{i=1}^{k}\Sigma^{i}$. Given a language $A$ over $\Sigma$, its {\em complement} is $\Sigma^*- A$, which is also denoted by $\overline{A}$ as long as the underlying alphabet $\Sigma$ is clear from the context.
We use the following three class operations between two language families $\CC_1$ and $\CC_2$: $\CC_1\wedge \CC_2 = \{ A\cap B\mid A\in\CC_1, B\in \CC_2\}$, $\CC_1\vee \CC_2 = \{A\cup B\mid A\in\CC_1, B\in \CC_2\}$, and $\CC_1-\CC_2 = \{A-B\mid A\in \CC_1, B\in\CC_2\}$, where $A$ and $B$ must be defined  over the same alphabet.

\subsection{Nondeterministic Pushdown Automata}\label{sec:npda-and-TM}

As our basic computation models, we use the following types of finite-state machines: {\em one-way deterministic finite automaton} (or dfa, in short) with $\lambda$-moves,  {\em one-way nondeterministic pushdown automaton} (or npda) with $\lambda$-moves, and {\em one-way probabilistic pushdown automaton} (or ppda), where a {\em $\lambda$-move} (or a $\lambda$-transition) is a transition of the machine's configurations in which a target tape head stays still. Notice that, as remarked later,  allowing $\lambda$-moves in any computation of a one-way pushdown automaton is crucial when output tapes are particularly involved.

Formally, an npda $M$ is a tuple $(Q,\Sigma,\{\cent,\dollar\},\Gamma,\delta,q_0,Z_0, Q_{acc}, Q_{rej})$, where $Q$ is a finite set of \emph{inner states}, $\Sigma$ is an input alphabet, $\Gamma$ is a stack alphabet, $Z_0$ ($\in \Gamma$) is the bottom marker of a stack, $q_0$ ($\in Q$) is the initial state, $Q_{acc}$ ($\subseteq Q$) is a set of \emph{accepting states}, $Q_{rej}$ ($\subseteq Q$) is a set of \emph{rejecting states}, and $\delta$ is a transition function mapping $(Q-Q_{halt})\times (\check{\Sigma}\cup\{\lambda\}) \times\Gamma$ to $\PP(Q\times \Gamma^*)$ with  $\check{\Sigma}=\Sigma\cup\{\cent,\dollar\}$ and   $Q_{halt} = Q_{acc}\cup Q_{rej}$, where elements in $Q_{halt}$ are called \emph{halting states} and $\cent$ and $\dollar$ are two distinguished endmarkers.
Any step associated with an application of transition of the form $\delta(q,\lambda,a)$ is called a \emph{$\lambda$-move} (or a \emph{$\lambda$-transition}).
The machine $M$ is equipped with a read-only input tape and its tape head cannot move backward.  On such a read-only input tape, an input string is surrounded by the endmarkers as $\cent x \dollar$.
The machine $M$ follows the \emph{endmarker convention}: (i) $M$ begins with reading $\cent$ (namely, $\delta(q_0,\cent,Z_0)$ is the first move of $M$) and (ii) whenever $M$ reads $\dollar$, it must enter a halting state without modifying the stack content (namely,  for any $q\in Q-Q_{halt}$ and $a\in\Gamma$, $\delta_{M}(q,\dollar,a)\subseteq \{(p,a)\mid q\in Q_{halt}\}$). Remember that $M$ is allowed to enter a halting state at any step.
We express the content of the stack as $\tau_1\tau_2\cdots \tau_kZ_0$ ($\tau_i\in\Gamma$) from top to bottom, where $Z_0$ is located at the bottom of the stack and $\tau_1$ is at the top. Assume that $M$ is in inner state $p$, scanning symbol $\sigma$ on the input tape, with stack containing $au Z_0$. A transition $(q,w)\in\delta(p,\sigma,a)$ instructs $M$ to change $a$ to $w$ in the stack (hence the stack becomes $wuZ_0$), enter inner state $q$, and move its input tape head to the right unless
$\sigma=\lambda$.
In particular, when $w=\lambda$, $M$ is informally said to \emph{pop up} $a$ and then the stack content changes from $au Z_0$ to $uZ_0$. In contrast, if $w\neq \lambda$, we may say that $M$ \emph{pushes down} $w$
by replacing $a$.

The machine $M$ must halt instantly after entering a  halting state. More importantly, we may not be able to implement an internal clock inside an npda to measure its runtime.  Therefore, we need to demand that {\em all} computation paths of $M$ should terminate {\em eventually}; in other words, along any computation path, $M$ must enter an appropriate halting state to stop. An \emph{accepting} (resp., a \emph{rejecting}) \emph{computation path}  is a computation path that ends with an accepting (resp., a rejecting) state.
For any of the above machines $M$,  we write  $PATH_{M}(x)$ to express a collection of all computation paths produced by $M$ on input $x$ and we use $ACC_{M}(x)$ (resp., $REJ_{M}(x)$) to denote a set of all accepting (resp., rejecting)  computation paths of $M$ on input $x$. Similarly, let $ALL_{M}(x)$ denote a set of all (halting) computation paths of $M$ on $x$.
It is important to remember that even if we further require all computation paths of $M$ to terminate after $O(n)$ steps for any input of length $n$, the definition of context-free languages does not change. For this reason, we implicitly assume that \emph{all computation paths must terminate in linear time}.  In general, we say that $M$ \emph{recognizes} language $L$ over alphabet $\Sigma$ if, for every $x\in L$, $M$ \emph{accepts} $x$ (\ie $ACC_{M}(x)\neq\setempty$) and, for every $x\in\Sigma^*-L$, $M$ \emph{rejects} $x$ (\ie $ALL_{M}(x)=REJ_{M}(x)$).

In contrast, a dpda uses a transition function $\delta$ mapping $(Q-Q_{halt})\times (\check{\Sigma}\cup\{\cent,\dollar\})\times\Gamma^*$ to $(Q\times\Gamma^*)\cup\{\setempty\}$, which satisfies the following \emph{deterministic property}: for any $q\in Q$ and $a\in\Gamma$, if $\delta(q,\lambda,a)\neq\setempty$, then $\delta(q,\sigma,a)=\setempty$ for all $\sigma\in\check{\Sigma}$.

The notations $\reg$, $\cfl$, and $\dcfl$ stand for the families of all regular languages (recognized by dfa's), of all context-free languages (recognized by npda's), and of all deterministic context-free languages (recognized by deterministic pushdown automata), respectively.  An advised language family $\reg/n$ in \cite{TYL10} consists of languages $L$ such that there exist an advice alphabet $\Gamma$, a length-preserving (total) advice function $h:\nat\rightarrow\Gamma^*$, and a language $A\in\reg$ satisfying $L = \{x\mid \track{x}{h(|x|)}\in A\}$, where $h$ is {\em length preserving} if $|h(n)|=n$ for all numbers $n\in\nat$.  By replacing $\reg$ with $\cfl$ in $\reg/n$,
another advised family $\cfl/n$ in \cite{Yam08} is obtained from $\cfl$.
A language $L$ over $\Sigma$ is called {\em tally} if $L\subseteq\{a\}^*$ holds for a certain symbol $a\in\Sigma$, and the notation $\tally$ indicates the collection of all such tally languages.

To describe multi-valued partial functions, we need to make
an npda equipped with an output tape, which must be write-only.
Whenever we refer to a {\em write-only tape}, we always assume that (i) initially, all cells of the tape are blank, (ii) a tape head starts at the so-called {\em start cell}, (iii) the tape head steps forward whenever it writes down any non-blank symbol, and (iv) the tape head can stay still only in a blank cell.
Therefore, all cells through which the tape head passes during a computation must contain no blank symbols.
Such an npda $M$ is of the form $(Q,\Sigma,\{\cent,\dollar\}, \Theta,\Gamma,\delta,q_0,Z_0,Q_{acc},Q_{rej})$ with an extra output alphabet $\Theta$. A transition function $\delta$ thus maps $(Q-Q_{halt})\times (\check{\Sigma}\cup\{\lambda\})\times \Gamma$ to $\PP(Q\times \Gamma^* \times (\Theta\cup\{\lambda\}))$.
An {\em output} (outcome or output string) along a computation path is a string produced on the output tape after the  computation path is terminated. We call an output string {\em valid} (or \emph{legitimate}) if it is produced along a certain accepting computation path. When we refer to the machine's outputs, we normally disregard any strings left on the output tape on a rejecting computation path, and thus we consider only valid outcomes as ``outcomes.''

We say that $M$ is \emph{well-behaved at $\dollar$} if $M$ enters a halting state without writing any non-blank output symbol while reading $\dollar$ (\ie $\delta(q,\dollar,a)\subseteq \{(p,a,\lambda)\mid p\in Q_{halt}\}$) and that \emph{well-behaved at $\cent$} if $M$ writes no non-blank output symbol (\ie $\delta(q,\cent,Z_0) \subseteq \{(p,w,\lambda)\mid p\in Q,w\in\Gamma^*\}$). Throughout this paper, we demand implicitly that all machines having output tapes should be well-behaved at both $\cent$ and $\dollar$.

A {\em multi-valued partial function} $f$ is in $\cflmv$ if there exist a constant $c>0$ and an npda $M$ equipped with a one-way read-only input tape together with a write-only output tape such that, for every string $x$, (1) all computation paths of $M$ on $x$ terminate within $c|x|+c$ steps and (2)  $f(x)$ is a set composed of all outcomes of $N$ on the input $x$ along accepting computation paths \cite{Yam11}. Notice that, if $\Gamma$ is an output alphabet of $f$, the cardinality $\|f(x)\|$ is upper-bounded by $|\Gamma|^{cn+c}$. When we deal only with ``total'' functions in $\cflmv$,
we obtain $\cflmvt$ \cite{Yam11}.

For any dpda $M$ equipped with an output tape, we demand that (1) $M$ satisfies the deterministic property (described above), (2) before reading the right endmarker $\dollar$, $M$ must write a special symbol (called the \emph{termination symbol}\footnote{The use of this termination symbol is not necessary for npda's, because nondeterminism can eliminate it. The importance of this symbol will be clarified in the proof of Lemma \ref{basic-Turing}(2).}) $\dollar$ to mark the end of an output string, and (3) $M$ is well-behaved at both $\cent$ and $\dollar$.

A ppda is a variant of npda $M$, in which any transition of $M$ is dictated by a given probability distribution. For simplicity, we always assume that every next move of $M$ must be made with \emph{equal probability}. A language $L$ is in $\bpcfl$ if there exist a constant $\varepsilon\in[0,1/2)$ and a ppda $M$ such that, on any input $x$, if $x\in L$, then $M$ accepts $x$ with probability at least $1/2+\varepsilon$; otherwise, $M$ rejects $x$ with probability at least $1/2+\varepsilon$. In this case, we also say that $M$ makes \emph{bounded-error probability}. In contrast, $\pcfl$ is defined as the collection of languages $L$ that are recognized by ppda's $M$ with \emph{unbounded-error probability}; that is, if $x\in L$ then $M$ accepts with probability more than $1/2$; otherwise, $M$ rejects $x$ with probability at least $1/2$.

\subsection{Circuit Families and Higher Complexity Classes}

For higher complexity classes, we will review basic notions and notations. To handle time/spec-bounded computation, we use two models of
{\em two-way deterministic Turing machine} (or DTM) and {\em two-way nondeterministic Turing machine} (or NTM).  Each of those machines is conventionally equipped with a single read-only input tape as well as a single read/write work tape unless otherwise stated for clarity.
Let $\p$ (resp., $\np$) be composed of all languages recognized by DTMs (resp.,  NTMs) in polynomial time. Given each index $k\in\nat$, we define $\deltap{0}=\deltap{1}= \sigmap{0}=\pip{0}=\p$,  $\deltap{k+1} = \p^{\sigmap{k}}$, $\sigmap{k+1}=\np^{\sigmap{k}}$, and $\pip{k+1}=\co\sigmap{k+1}$, where $\p^{\CC} = \bigcup_{A\in\CC}\p^A$ (resp., $\np^{\CC} = \bigcup_{A\in\CC}\np^A$) and $\p^A$ (resp., $\np^A$) is the family of languages recognized by polynomial-time DTMs (resp., NTMs) with adaptive queries to a set $A$, which is given as an {\em oracle}. Those language families constitute the so-called {\em polynomial(-time) hierarchy} \cite{MS72,Sto77,Wra77}.  Let $\ph = \bigcup_{k\in\nat}\sigmap{k}$.  We denote by $\dl$  the family of all languages, each of which is recognized by a certain DTM with a two-way read-only input tape and a two-way read/write work tape using $O(\log{n})$ cells of the work tape.

Two probabilistic language families $\bpp$ and $\pp$ are defined using polynomial-time probabilistic Turing machines (or PTMs) allowing bounded-error and unbounded-error probabilities, respectively.

For each fixed constant $k\in\nat$,  $\nc{k}$ expresses a collection of languages recognized by log-space uniform Boolean circuits of polynomial-size and $O(\log^k{n})$-depth. It is known that $\nc{0}$ is properly included within $\nc{1}$; however, no other separations are known to date.  Similarly, $\ac{k}$ is defined except that all Boolean gates in a circuit may have unbounded fan-in.
Moreover, $\sac{1}$ demotes a class of languages recognized by log-space uniform families of polynomial-size Boolean circuits of $O(\log{n})$ depth and {\em semi-bounded} fan-in (that is, having $AND$ gates of bounded fan-in and $OR$ gates of unbounded fan-in), provided that the negations appear only at the input level. This class $\sac{1}$ is located between $\nc{1}$ and $\ac{1}$.  Venkateswaran \cite{Ven91} demonstrated that $\sac{1}$ coincides with the family $\mathrm{LOGCFL}$ of all languages that are log-space many-one reducible to context-free languages.
Moreover, $\tc{1}$ consists of all languages recognized by log-space uniform families of $O(\log{n})$-depth polynomial-size circuits whose gates compute {\em threshold functions}.

\section{Nondeterministic Reducibilities}
\label{sec:many-one-tt-reduction}

A typical tool in comparing the computational complexity of two formal languages is the form of {\em resource-bounded reducibility}. Such reducibility is also regarded as a {\em relativization} of its underlying language family. Hereafter, we intend to introduce an appropriate  notion of {\em nondeterministic many-one reducibility} to a theory of context-free languages using a specific computation model of  one-way nondeterministic pushdown automata (or npda's) described in Section \ref{sec:npda-and-TM}. This new reducibility catapults a basic architecture of a hierarchy built  in Section \ref{sec:CFL-hierarchy} over the family $\cfl$ of context-free languages.

\subsection{Many-One Reductions by Oracle Npda's}\label{sec:many-one-reduction}

Our exposition begins with an introduction of an appropriate form of nondeterministic many-one reducibility whose reductions are operated by  npda's. In the past literature, there were preceding ground works on many-one reducibilities within a framework of a theory of formal languages and automata.
Based on deterministic/nondeterministic finite automata (or dfa's/nfa's), for instance, Reinhardt \cite{Rei90} discussed two  many-one reducibilities between two languages. Tadaki, Yamakami, and Li \cite{TYL10} also studied the roles of various many-one reducibilities defined by one-tape linear-time Turing machines, which turn out to be closely related to finite automata. Notice that those computation models have no extra memory storage to use. In contrast, we attempt to use npda's as a basis of our reducibility.

Our reduction machine is essentially a restriction of ``pushdown transducer'' or ``algebraic transduction'' (see, \eg \cite{Ber79}); however,  we plan to  define such reduction machines in a style of ``oracle machines.'' An {\em $m$-reduction npda} is a variant of npda, which is additionally equipped with a {\em query tape} on which the machine writes a string surrounded by blank cells starting at the designated {\em start cell} for the purpose of making a query to a given {\em oracle}. This query tape is essentially a write-only  output tape, and thus the query-tape head must move to the next blank cell whenever it writes down a non-blank symbol. As noted for npda's in Section \ref{sec:npda-and-TM}, we also demand that the machine halts on \emph{all} computation paths within $O(n)$ steps, where $n$ is the input size.
In analogy with ``oracle Turing machine,''  we often use the term \emph{oracle npda} to mean an npda with an extra {\em query tape}.
When $M$ halts in an accepting state with a string $y$ written on its query tape, the string $y$ is automatically transmitted to a device called an \emph{oracle}. In this case, we informally say that $M$ \emph{makes a query} to oracle $A$ with query word $y$ (or $M$ \emph{queries} $y$ to $A$).

Formally, an $m$-reduction npda (or an oracle npda) is
a tuple $(Q,\Sigma,\{\cent,\dollar\},\Theta,\Gamma,\delta,q_0,Z_0,Q_{acc},Q_{rej})$, where $\Theta$ is a \emph{query alphabet} and $\delta$ is now of the form
\[
\delta:(Q-Q_{halt})\times (\check{\Sigma}\cup\{\lambda\})\times \Gamma \rightarrow \PP(Q \times \Gamma^* \times (\Theta\cup\{\lambda\})),
\]
where $Q_{halt}$ and $\check{\Sigma}$ are defined in Section \ref{sec:alphabet-string}.
There are two types of $\lambda$-moves to recognize. Assuming $(p,\tau,\xi)\in \delta(q,\sigma,\gamma)$, (i) when $\sigma=\lambda$, the input-tape head stays still (or makes a $\lambda$-move); on the contrary, (ii) when $\tau=\lambda$, the query-tape head stays still. Remember that \emph{all} tape heads must move only in one direction (from left to right) and \emph{all} computation paths halt in $O(n)$ time. Since the query tape is actually an output tape, we also use the terminology of ``well-behavedness''
given in Section \ref{sec:npda-and-TM} for an oracle npda $M$.

We say that a language $L$ over alphabet $\Sigma$ is  {\em many-one CFL-reducible} to another language $A$ over alphabet $\Theta$
if there exists an $m$-reduction npda $M$ using $\Sigma$ and $\Theta$ respectively as an input alphabet and a query alphabet such that, for every input $x\in\Sigma^*$, (1) $ACC_{M}(x)\neq\setempty$, (2) along each computation path $p\in ACC_{M}(x)$, $M$ produces a valid query  string $y_p\in\Theta^*$  on the query tape, and (3) $x\in L$ if and only if $y_p\in A$ for an appropriate computation path $p\in ACC_{M}(x)$.
In this case, we succinctly say that $M$ \emph{reduces} (or \emph{$m$-reduces}) $L$ to $A$. Occasionally, we also say that $M$ \emph{recognizes $L$ relative to} $A$. Given an oracle npda $M$ and an oracle $A$, the notation $L(M,A)$ (or $L(M^A)$) denotes the set of strings accepted by $M$ relative to $A$.
By substituting  a standard output tape for this query tape, $M$ can be seen as an npda outputting $y_p$ on its output tape; in other words, $M$ defines a multi-valued \emph{total} function in $\cflmvt$.
Based on this formulation, we can rephrase the above definition as follows: $L$ is many-one $\cfl$-reducible to $A$ if and only if there exists a function $f\in\cflmvt$ satisfying $L=\{x\mid f(x)\cap A\neq\setempty\}$. With the use of this $m$-reducibility, we make the notation $\cfl_{m}^{A}$ (or $\cfl_{m}(A)$)  denote the family of  all languages $L$ that are many-one $\cfl$-reducible to $A$.
Given a class of oracles, $\cfl_{m}^{\CC}$ (or $\cfl_{m}(\CC)$) denotes the union $\bigcup_{A\in\CC}\cfl_{m}^A$.

In a similar way, we can define \emph{many-one NFA-reducibility} using ``oracle nfa'' instead of ``oracle npda'';
an {\em $m$-reduction nfa} (or an \emph{oracle nfa})
$M$ is a tuple $(Q,\Sigma,\{\cent,\dollar\},\Theta,\delta,q_0,Q_{acc},Q_{rej})$, where $\delta$ is a map from $(Q-Q_{halt})\times (\check{\Sigma}\cup\{\lambda\})$ to $\PP(Q\times(\Theta\cup\{\lambda\}))$. Similarly to oracle npda's, we also impose an $O(n)$ time-bound on all computation paths of $M$.
Let $\nfa_{m}^A$ (or $\nfa_{m}(A)$) denote the family of all languages that are many-one NFA-reducible to $A$.
Obviously, $\nfa_{m}^{A}\subseteq \cfl_{m}^{A}$ holds for any oracle $A$. We use the notation $\nfa_{m}^{\CC}$ (or $\nfa_{m}(\CC)$) for the union $\bigcup_{A\in\CC}\nfa_{m}^A$.

Throughout this paper, we intend to use an informal term of ``guessing'' when we refer to a nondeterministic choice (or a series of nondeterministic choices). For example, when we say that an npda $M$ \emph{guesses} a string $z$, we actually mean that $M$ makes a series of nondeterministic choices that cause to produce
$z$.


Let us start with a quick example of languages that are many-one CFL-reducible to languages in $\cfl$.

\begin{example}\label{ex:DUP_2}
As the first concrete example, setting $\Sigma=\{0,1\}$, let us consider the language $Dup_2=\{xx\mid x\in\Sigma^*\}$. This language is known to be non-context-free; however, it can be many-one $\cfl$-reducible to $\cfl$ by the following $M$ and $A$. An $m$-reduction (or an oracle) npda $M$ nondeterministically produces a query word $x^R\natural y$ (with a special symbol $\natural$) from each input of the form $xy$ using a stack appropriately More formally, a transition function $\delta$ of this oracle npda $M$ is given as follows:
$\delta(q_0,\cent,Z_0) =\{(q_0,Z_0,\lambda)\}$,
$\delta(q_0,\dollar,Z_0) =\{(q_{acc},Z_0,\natural)\}$,
$\delta(q_0,\sigma,Z_0) =\{(q_1,\sigma Z_0,\lambda)\}$,
$\delta(q_1,\sigma,\tau) =\{(q_1,\sigma\tau,\lambda), (q_2,\sigma\tau,\lambda)\}$,
$\delta(q_2,\lambda,\tau) =\{(q_2,\lambda,\tau)\}$,
$\delta(q_2,\lambda,Z_0) =\{(q_3,Z_0,\natural)\}$,
$\delta(q_3,\sigma,Z_0) =\{(q_3,Z_0,\sigma)\}$, and
$\delta(q_3,\dollar,Z_0) =\{(q_{acc},Z_0,\lambda)\}$, where $\sigma,\tau\in\Sigma$.
Our $\cfl$-oracle $A$ is defined as $\{x^R\natural x\mid x\in\Sigma^*\}$; that is, the oracle $A$ checks whether $x=y$ from the input $x^R\natural y$ using its own stack. In other words, $Dup_{2}$ belongs to $\cfl_{m}^{A}$, which is included in $\cfl_{m}^{\cfl}$. Similarly, the non-context-free language $Dup_3=\{xxx\mid x\in\Sigma^*\}$ also falls into $\cfl_{m}^{\cfl}$. For this case, we design an $m$-reduction npda to produce $x^R\natural y \natural y^R \natural z$ from each input $xyz$ and make a CFL-oracle check whether $x=y=z$ by using its stack twice. Another language $Match=\{x\# w\mid \exists u,v\,[w =uxv]\,\}$, where $\#$ is a separator not in $x$ and $w$, also belongs to $\cfl_{m}^{\cfl}$.  These examples prove that $\cfl_{m}^{\cfl} \neq \cfl$.
\end{example}

A slightly more complicated example is given below.

\begin{example}\label{ex:CFL(k)}
The language $Sq =\{0^n1^{n^2}\mid n\geq1\}$ belongs to $\cfl_{m}^{\cfl}$. To see this fact, let us consider the following oracle npda $N$ and oracle $A$. Given any input $w$, $N$ checks if $w$ is of the form $0^i1^j$, it  nondeterministically selects $(j_1,j_2,\ldots,j_k)\in\nat^{k}$, and it produces on its query tape a string $w'$ of the form $0^i\natural 1^{j_1}\natural 1^{j_2}\natural \cdots \natural 1^{j_k}$. Simultaneously, $N$ checks if (i) $j=j_1+j_2+\cdots +j_k$,  (ii) $i=j_1$ by first pushing $0^i$ onto a stack, and (iii) $j_2=j_3$, $j_4=j_5$, $\ldots$ using the stack properly.
The desired oracle $A$ receives $w'$ and checks if the following two  conditions are all met: (i') $j_1=j_2$, $j_3=j_4$, $\ldots$ and (ii') $i=k$ by first pushing $0^i$ onto a stack and then counting the number of $\natural$s.
Clearly, $A$ belongs to $\cfl$. Therefore, $Sq$ is in $\cfl_{m}^{A}$, which is included in $\cfl_{m}^{\cfl}$. A similar idea proves that the language $Comp=\{0^n\mid \,\text{$n$ is a composite number}\,\}$  belongs to $\cfl_{m}^{\cfl}$. In symmetry, $Prim =\{0^n\mid \,\text{$n$ is a prime number}\,\}$ is a member of $\co(\cfl_{m}^{\cfl})$.
\end{example}


Notice that $\np$ is closed under many-one NP-reductions; nonetheless,
$\cfl$ cannot be closed under many-one $\cfl$-reductions. For this latter claim, we argue that, if
$\cfl$ is closed under this reducibility, then $\cfl_{m}^{\cfl}=\cfl$ follows; however, this contradicts what we have seen in Example \ref{ex:DUP_2}. This non-closure property  certainly marks a critical feature of the computational behaviors of context-free languages.
In what follows, we want to strengthen the separation between $\cfl_{m}^{\cfl}$ and $\cfl$ even in the presence of advice.

\begin{proposition}\label{CFL_m^CFL-separation}
$\cfl_{m}^{\cfl}  \nsubseteq \cfl/n$.
\end{proposition}

To show this separation, we will briefly review a notion of $k$-conjunctive closure over $\cfl$.
Given each number $k\in\nat^{+}$, the {\em $k$-conjunctive closure of CFL}, denoted $\cfl(k)$ in \cite{Yam11}, is defined recursively as follows: $\cfl(1) = \cfl$ and $\cfl(k+1) = \cfl(k)\wedge \cfl$. These language families truly form an infinite hierarchy \cite{LW73}.  For convenience, we set $\cfl(\omega) = \bigcup_{k\in\nat^{+}}\cfl(k)$.
For advised language families, it is known that $\cfl\nsubseteq\reg/n$ \cite{TYL10}, $\co\cfl\nsubseteq\cfl/n$ \cite{Yam08},  and $\cfl(2)\nsubseteq\cfl/n$ \cite{Yam09}.
In the following proof of Proposition \ref{CFL_m^CFL-separation}, we attempt to prove that $\cfl(2)\subseteq \cfl_{m}^{\cfl}$ and $\cfl(2)\nsubseteq\cfl/n$.

\vs{-2}
\begin{proofof}{Proposition \ref{CFL_m^CFL-separation}}
Toward a contradiction, we assume that  $\cfl_{m}^{\cfl}\subseteq\cfl/n$. In this proof, we need the following containment:  $\cfl(2)\subseteq\cfl_{m}^{\cfl}$. For the sake of a later reference, we prove a more general statement described below.

\begin{claim}\label{CFL(2)-bound}
For every index $k\geq1$, $\cfl(k+1)\subseteq \cfl_{m}^{\cfl(k)}$.
\end{claim}

\begin{proof}
Let $L$ be any language in $\cfl(k+1)$ and take two languages $L_1\in\cfl$ and $L_2\in\cfl(k)$ for which $L = L_1\cap L_2$. There exists an npda $M_1 = (Q_1,\Sigma,\{\cent,\dollar\},\delta_1,q_0,Q_{acc,1},Q_{rej,1})$ that recognizes $L_1$. Without loss of generality, we assume that $M_1$ enters a halting state  when it scans the right endmarker $\dollar$; namely, $(q,w)\in\delta_{1}(p,\dollar,a)$ implies $q\in Q_{halt}$. A new oracle npda $N$ (with a transition function $\delta$) is defined to behave as follows. On input $x$, $N$ starts  simulating $M_1$ on $x$.  While reading each symbol $\sigma$ from $x$, $N$ also copies it down to a write-only query tape; namely, if $(q,w)\in\delta_{1}(p,\sigma,a)$, then $(q,w,\sigma)\in\delta(p,\sigma,a)$.
When $M_1$ enters a halting state, $N$ enters the same halting state; namely, $(q,w)\in\delta_{1}(p,\dollar,a)$ implies $(q,w,\lambda)\in\delta(p,\dollar,a)$.
It follows that, given any input $x$, $x$ is  in $L$ if and only if $N$ on the input $x$ produces the query string $x$ in an accepting state and $x$ is actually in $L_2$. This equivalence relation implies that $L$ belongs to $\cfl_{m}^{L_2}$, which is obviously a subclass of $\cfl_{m}^{\cfl(k)}$.
\end{proof}

Since $\cfl(2)\subseteq \cfl_{m}^{\cfl}$ by Claim \ref{CFL(2)-bound},  our assumption implies  that $\cfl(2)\subseteq \cfl/n$. This contradicts the class separation $\cfl(2)\nsubseteq\cfl/n$, proven in \cite{Yam09}. Therefore, the proposition must hold.
\end{proofof}


A {\em Dyck language} $L$ over alphabet $\Sigma=\{\sigma_1,\sigma_2,\ldots,\sigma_d\}\cup \{\sigma'_1,\sigma'_2,\ldots,\sigma'_d\}$ is a language generated by a deterministic context-free grammar whose production set is $\{S\rightarrow \lambda | SS | \sigma_i S \sigma'_i: i\in[d]\}$, where $S$ is a start symbol. For convenience, denote by $DYCK$ the family of all Dyck languages.

\begin{lemma}\label{non-closure-many-one}
$\cfl_{m}^{\cfl} = \cfl_{m}^{\dcfl} = \cfl_{m}^{DYCK}$.
\end{lemma}

In the following proof, we will employ a simple but useful technique of guessing and verifying a correct {\em stack history} (that is, a series of pushed and popped symbols). Whenever an npda tries to either push down  symbols into its stack or pop up a symbol from the stack, instead of actually using the stack, we guess symbols and write a series of those guessed symbols (as a stack history) down on a query tape and ask an oracle to verify that it is indeed a correct stack history. This technique will be frequently used in other sections.

\vs{-2}
\begin{proofof}{Lemma \ref{non-closure-many-one}}
Since $DYCK \subseteq \dcfl \subseteq \cfl$, we obtain $\cfl_{m}^{DYCK}\subseteq \cfl_{m}^{\dcfl}\subseteq \cfl_{m}^{\cfl}$. Thus,  we are  hereafter focused on proving that $\cfl_{m}^{\cfl} \subseteq \cfl_{m}^{DYCK}$.  As the first step toward this goal, we will prove the following characterization of $\cfl$ in terms of Dyck languages.

\begin{claim}\label{NFA-to-DYCK}
$\cfl = \nfa_{m}^{DYCK}$.
\end{claim}

This claim can be seen as a different form of the Chomsky-Sch{\"u}tzenberger theorem (see, \eg \cite{Ber79}).
Notice that Reinhardt \cite{Rei90} proved a similar statement using his special language $L_{pp}$.
Our intended proof relies on the behaviors of underlying npda's and aims at constructing many-one reductions.

\vs{-2}
\begin{proofof}{Claim \ref{NFA-to-DYCK}}
($\subseteq$)
For any language $L$ in $\cfl$, consider an npda $M = (Q,\Sigma,\{\cent,\dollar\},\Gamma,\delta,q_0,Z_0,Q_{acc},Q_{rej})$ that recognizes $L$. Without loss of generality, we impose on $M$ the following acceptance criterion: $M$ enters an accepting state
exactly when its stack becomes ``empty'' by erasing the bottom marker $Z_0$. Let $\Gamma =\{\sigma_1,\sigma_2,\ldots,\sigma_d\}$ be a stack alphabet of $M$,  including  $Z_0$. Corresponding to each symbol $\sigma_i\in\Gamma$, we introduce another fresh symbol $\sigma'_i$ and then we set $\Gamma'=\{\sigma'_1,\sigma'_2,\ldots,\sigma'_d\}$.
Without loss of generality, it is possible to assume that $M$ makes no $\lambda$-move until its tape head scans the right endmarker $\dollar$ (see, \eg \cite{HMU01} for the proof). For convenience, we further assume that, when the tape head reaches $\dollar$, $M$ must make a series of $\lambda$-moves to empty the stack before entering a certain halting state. Even with such acceptance criteria, our npda can be assumed to halt on ``all'' computation paths in linear time.

Let us construct a new oracle nfa $N$.
In the following description of $N$, we will intentionally identify all symbols in $\Gamma$ with  ``pushed'' symbols, and all symbols in $\Gamma'$ with ``popped'' symbols. Given any input string $x$, $N$ basically  simulates each step of $M$'s computation made on the input $x$ with using no stack.
At the first step, if $M$ pushes down a string $wZ_0\in\Gamma^*$, then $N$ writes down $Z_0w^R$ on a query tape.
At any later step, $M$ in state $p$ scanning $\tau$ pushes down a string $w=\sigma_{i_1}\sigma_{i_2}\cdots \sigma_{i_k}\in\Gamma^*$ in place of the top symbol of the stack, $N$ first guesses this top symbol, say, $\sigma_j$ and, if $\sigma_j\neq Z_0$, then $N$ writes down $\sigma'_jw^R$ on its query tape.  This is because $\sigma_j$ is on the top of the stack, and thus it must be first popped up before pushing $w$ down. When $M$ pops up a symbol, $N$ instead guesses it, say, $\sigma_i$ and writes down $\sigma'_i$ (not $\sigma_i$) on the query tape.
For example, if the stack content changes as $Z_0\stackrel{push}{\rightarrow} \sigma_1\sigma_2Z_0\stackrel{push}{\rightarrow} \sigma_3\sigma_1\sigma_2Z_0\stackrel{pop}{\rightarrow} \sigma_1\sigma_2Z_0\stackrel{pop}{\rightarrow} \sigma_2Z_0\stackrel{pop}{\rightarrow}  Z_0\stackrel{pop}{\rightarrow} \lambda$, then the corresponding string that $N$ has produced is $Z_0\sigma_2\sigma_1\sigma'_1\sigma_1 \sigma_3\sigma'_3\sigma'_1\sigma'_2Z'_0$. Note that $N$ halts in $O(n)$ steps because $M$ always moves its input tape head until reaching $\dollar$.
More formally, let $N= (Q,\Sigma,\{\cent,\dollar\},\Theta,\delta_N,q_0,Q_{acc},Q_{rej})$ with $\Theta=\Gamma\cup\Gamma'$. The transition function $\delta_N$ is given as follows: $\delta_N(q_0,\cent)= \{(p,Z_0w^R)\mid p\in\Theta^*, w\in\Gamma^*, (p,wZ_0)\in\delta(q_0,\cent,Z_0)\}$ and  $\delta_N(q,\tau) = \{(p,\sigma'w^R)\mid p\in Q, \sigma\in\Gamma, w\in\Gamma^*, (p,w)\in\delta(q,\tau,\sigma)\}$ for any $q\in Q-Q_{halt}$ and $\tau\in\check{\Sigma}\cup\{\lambda\}$.

Finally, $B$ is chosen to be a Dyck language over the alphabet $\Gamma\cup\Gamma'$.
Note that, if a query word $u$ encodes a stack history (with the stack becoming empty when the machine halts), then $u$ obviously belongs to $B$. Therefore, $L$ must be in $\cfl_{m}^{B}$, which is a subclass of $\cfl_{m}^{DYCK}$.

($\supseteq$) Assume that $L\in\nfa_{m}^{B}$ for an appropriate oracle $B$ in $DYCK$. Let us take an $m$-reduction nfa $M_1= (Q_1,\Sigma,\{\cent,\dollar\},\Theta,\delta_1,q_0,Z_0,Q_{1,acc},Q_{1,rej})$ that reduces $L$ to $B$. Since $B$ is in $\dcfl$, take a dpda $M_2= (Q_2,\Theta,\{\cent,\dollar\},\Gamma_2,\delta_2,q_0,Z_0,Q_{2,acc},Q_{2,rej})$ that recognizes $B$. For our convenience, let us assume that \emph{$M_2$ makes no $\lambda$-move}. We will simulate both $M_1$ and $M_2$ on a special npda $N = (Q,\Sigma,\{\cent,\dollar\},\Gamma,\delta_N,q_0,Z_0,Q_{acc},Q_{rej})$ in the following fashion. Given any input $x$, $N$ starts the simulation of $M_1$ on $x$ without using any stack. If $M_1$ tries to write a symbol, say, $\tau$ on a query tape, then $N$ instead simulates one step of $M_2$'s computation that corresponds to the scanning of $\tau$  together with a certain stack symbol, say, $a$ on the top of the stack.

To be more formal, let $a\in\Gamma_2$, $w\in\Gamma_2^*$, $p\in Q_1$, $p\in Q_1-Q_{1,halt}$, $q\in Q_1$, $\tau\in\Theta$, $p'\in Q_2-Q_{2,halt}$, $q'\in Q_2$ and $\sigma\in\check{\Sigma}\cup\{\lambda\}$. Define $Q_{acc} = Q_{1,acc}\times Q_{2,acc}$ and $Q_{rej} = (Q_{1,rej}\times Q_2)\cup (Q_1\times Q_{2,rej})$. Moreover, the desired transition function $\delta_N$ of $N$ is defined as follows. For simplicity, we assume that, at the time when $M_1$ is entering a halting state, it never write any non-blank output symbol. Let $((q,q'),w)\in\delta_N((q_0,q_0),\cent,Z_0)$ if $(q,\lambda)\in\delta_1(q_0,\cent)$ and $(q',w)\in\delta_2(q_0,\cent,Z_0)$.
When $(q,\lambda)\in\delta_1(p,\sigma)$ with $q\notin Q_{1,halt}$, we set  $((q,p'),a)\in\delta_{N}((p,p'),\sigma,a)$. When $(q,\tau)\in \delta_1(p,\sigma)$ and $(q',w)\in\delta_2(p',\tau,a)$ with $q\notin Q_{1,halt}$ and $q'\in Q_{2,halt}$, we define   $((q,q'),w)\in\delta_{N}((p,p'),\sigma,a)$. When $(q,\lambda)\in\delta_1(p,\sigma)$ and $(q',w)\in\delta_2(p',\dollar,a)$ with $q\in Q_{1,halt}$ and $q'\in Q_{2,halt}$, we set $((q,q'),w)\in\delta_{N}((p,p'),\sigma,a)$. For all other cases, let $N$ enter appropriate rejecting states.

It is not difficult to show that $N$ accepts $x$ if and only if $x$ is in $L$. Therefore, it follows that $L\in\cfl$.
\end{proofof}

We note that another way to show ($\supseteq$) in the above proof is to use Greibach's  \cite{Gre73} hard language $L_0$, which satisfies the useful property that every context-free language is an inverse  homomorphic image of $L_0$ or $L_0-\{\lambda\}$.

Next, we claim the following equality.

\begin{claim}\label{CFL(NFA)-equal-CFL}
$\cfl_{m}^A = \cfl_{m}(\nfa_{m}^A)$ for any oracle $A$.
\end{claim}

\begin{proof}
($\subseteq$)  This is rather trivial, because $A\in\nfa_{m}^A$ holds by choosing an oracle nfa that takes input $x$ and then queries $x$ itself to an oracle.

($\supseteq$)  Assume that $L\in\cfl_{m}^{B}$ for a certain language $B$ in $\nfa_{m}^{A}$. Let us take an oracle npda $M_1 = (Q_1,\Sigma,\{\cent,\dollar\},\Theta,\Gamma, \delta_1,q_0,Z_0,Q_{1,acc},Q_{1,rej})$ recognizing $L$ relative to $B$ and also an oracle nfa $M_2 = (Q_2,\Theta,\{\cent,\dollar\},\Theta_2, \delta_2,q_0,Q_{2,acc},Q_{2,rej})$ recognizing $B$ relative to $A$.  A new machine $N$ is defined to behave as follows. On input $x$, $N$ simulates $M_1$ on $x$. Whenever $M_1$ tries to write a symbol, say, $\tau$, on its own query tape,  since $M_2$ has no stack usage, $N$ can simulate one step of $M$ while reading $\tau$ and a certain number of $\lambda$-moves made by $M_2$.

To be more precise, let $N = (Q,\Sigma,\{\cent,\dollar\},\Theta_2,\Gamma, \delta_N,(q_0,q_0),Z_0,Q_{acc},Q_{rej})$ with $Q=[2]\times Q'_1\times Q_2$, where   $Q'_1=\{q_0\}\cup \{\track{p}{c}\mid p\in Q_1,c\in C\}$ with $C=\{1,2,\cent,\dollar,\lambda\}$. For simplicity, assume that $Q_{a,acc}=Q_{2,acc}=\{q_{acc}\}$ and $Q_{1,rej}=Q_{2,rej}=\{q_{rej}\}$.
Let $Q_{acc} = \{(b,\track{q_{acc}}{c},q_{acc})\mid b\in[2],c\in C \}$ and $Q_{rej} = \{(b,\track{q_{rej}}{c},p_2),(b,\track{p_1}{c},q_{rej})\mid b\in[2], c\in C, p_1\in Q_{1}, p_2\in Q_{2}\}$.
Recall that $M_1$ is well-behaved at $\cent$ (\ie $(p,w,\tau)\in\delta_1(q_0,\cent,Z_0)$ implies $\tau=\lambda$).
To simplify the following description, we assume, without loss of generality,  that  $(p,w,\tau)\in\delta_1(q,\sigma,a)$ with $p\in \{q_{acc},q_{rej}\}$ implies $\tau=\lambda$.
Let $\delta_N((1,q_0,q_0),\cent,Z_0)$ contain $((2,\track{p_1}{\cent},q_0),w,\lambda)$ if $(p_1,w,\lambda)\in\delta_1(q_0,\cent,Z_0)$. For any $\sigma\in\Sigma\cup\{\lambda,\dollar\}$, let $\delta_N((1,\track{q_1}{\lambda},q_2),\sigma,a)$ contain $((1,\track{p_1}{\lambda},q_2),w,\lambda)$ if $(p_1,w,\lambda)\in\delta_1(q_1,\sigma,a)$ and $p_1\notin Q_{1,acc}$, and  $((2,\track{p_1}{\tau},q_2),w,\lambda)$ if $(p_1,w,\tau)\in\delta_1(q_1,\sigma,a)$, $\tau\in\Theta$, and $p_1\notin Q_{1,acc}$. In these cases, if $p_1\in Q_{1,acc}$, then we set $((2,\track{p_1}{\dollar},q_2),w,\lambda) \in \delta_N((1,\track{q_1}{\lambda},q_2),\sigma,a)$. Moreover, let $\delta_{N}((2,\track{q_1}{\tau},q_2),\lambda,a)$ contain $((b,\track{q_1}{\lambda},p_2),a,\xi)$ for any $b\in[2]$ if $(p_2,\xi)\in\delta_2(q_2,\tau)$, where  $\tau\in\Theta\cup\{\cent,\dollar,\lambda\}$.

In the end of its computation, $N$ produces query words exactly as $M_2$ does. Moreover, all computation paths of $M_1$ as well as $M_2$ terminate in linear time, $N$ also halts in linear time. Therefore, it is possible to verify that $N$ is indeed an $m$-reduction npda reducing $L$ to $A$. Thus, we obtain the desired membership $L\in\cfl_{m}^{A}$.
\end{proof}

By combining Claims \ref{NFA-to-DYCK} and \ref{CFL(NFA)-equal-CFL}, it follows that $\cfl_{m}^{\cfl} = \cfl_{m}(\nfa_{m}^{DYCK}) \subseteq \cfl_{m}^{DYCK}$.
\end{proofof}

We will introduce another technique of  simulating  two or more tape heads moving at (possibly) different speeds by a single tape head. Let us consider an npda $M$ with a write-only output tape. Since the tape heads of $M$ may make $\lambda$-moves and stay still at any moments on both input and output tapes, it seems difficult to synchronize the moves of those two tape heads for the purpose of splitting the output tape into two tracks and produce a string $\track{x}{y}$ from input string $x$ and output string $y$ of $M$. The best we can do is to insert a fresh symbol, say, $\natural$ between input symbols as well as output symbols to adjust the speeds of two tape heads.
To implement this idea, it is useful to introduce a terminology to describe strings obtained by inserting $\natural$. Assuming that $\natural\not\in\Sigma$, a {\em $\natural$-extension} of a given string $x$ over $\Sigma$ is a string $\tilde{x}$ over $\Sigma\cup\{\natural\}$ for which $x$ is obtained directly from $\tilde{x}$ simply by removing all occurrences of $\natural$ in $\tilde{x}$.  For instance, if $x=01101$, then $\tilde{x}$ may be $01\natural 1\natural01$, $011\natural\natural 01\natural$, or $\natural0\natural11\natural00\natural1\natural$.

Associated with such $\natural$-extensions, we can extend Dyck languages by adding the symbol $\natural$ as a part of its underlying alphabet, assuming that $\natural$ does not appear in the original alphabet, and also by considering $d$-tuples of strings over this extended alphabet. To be more formally, we first generalize our track notation $\track{x_1}{x_2}$ to $[x_1,x_2,\ldots,x_k]^{T}$, where subscript ``$T$'' refers to ``transposed'' as in the case of matrices. In particular, $[x_1,x_2]^{T}$ coincides with $\track{x_1}{x_2}$.
Given each index $d\in\nat^{+}$, $DYCK^{ext}_d$ consists of  all languages $L$  such that there exist $d$ \emph{extended Dyck languages} $A_1,A_2,\ldots,A_d$ for which $L$ consists of elements of the form $[x_1,x_2,\ldots,x_d]^{T}$  satisfying the following:  for every index $i\in[d]$, $x_i$ belongs to $A_i$. In particular, when $d=2$, any language $L$ in $DYCK^{ext}_2$ has the form $\{\track{x}{y}\mid x\in A,y\in B\}$ for certain extended Dyck languages $A$ and $B$.
It is worth noting that $DYCK^{ext}_d$ is  a subclass of $\dcfl(d)$ ($=\bigwedge_{i\in[d]}\dcfl$).
The following corollary generalizes Claim \ref{NFA-to-DYCK}.

\begin{corollary}\label{NFA-Dyck-equal-CFL}
For each fixed index $d\in\nat^{+}$, $\cfl(d) = \nfa_{m}(DYCK^{ext}_d)$.
\end{corollary}

In the proof of Corollary \ref{NFA-Dyck-equal-CFL} that follows shortly, to simplify the description of simulations of given oracle npda's, we need to introduce a special terminology.
Let $M=(Q,\Sigma,\{\cent,\dollar\},\Theta,\Gamma,\delta,q_0,Z_0,Q_{acc},Q_{rej})$ denote any oracle npda and $A$ be any oracle.
We say that a string $w$ of the form $\track{\tilde{x}}{\tilde{y}}$ {\em encodes input $x$ and query word $y$ along a computation path of $M$} if (i) along a ``{certain}'' \emph{accepting} computation path $\gamma$ of $M$ on $x$, $M$ starts with the input $x$ and produces this string $y$ on its query tape, (ii) $\tilde{x}$ and $\tilde{y}$ are respectively $\natural$-extensions of those $x$ and $y$, and (iii) there is another oracle npda $N$ that takes $w$ and, by scanning each symbol in $w$, $N$ traverses the computation path $\gamma$ as follows. Assume that $N$ scans a symbol of the form $\track{\sigma}{\tau}$. If $\sigma\neq\natural$, then $N$ simulates exactly one step of $M$ made by scanning $\sigma$ on its input tape and writing $\tau$ on its query tape.
On the contrary, if $\sigma=\natural$, then $N$ simulates exactly one $\lambda$-move of $M$ (without moving its input-tape head).
To be more precise, for any $\sigma\in\Sigma\cup\{\natural\}$ and $\tau\in\Theta\cup\{\natural\}$,
let $\delta_{N}(q,\track{\sigma}{\tau},a) = \{(p,w)\mid (p,w,t(\tau))\in\delta_{M}(q,t(\sigma),a)\}$, where $t(\xi)=\xi$ if $\xi\neq\natural$ and $t(\xi)=\lambda$ otherwise.  Let $\delta_{N}(q_0,\cent,Z_0) = \{(p,\hat{\tau}w)\mid p\in Q, w\in\Gamma^*, \tau\in\Theta\cup\{\lambda\}, (p,w,\tau)\in\delta_{M}(q_0,\cent,Z_0)\}$,
$\delta_{N}(q,\track{\natural}{\tau},\hat{\tau}) = \{(q,\lambda)\}$, and $\delta_{N}(q,\dollar,a) = \{(p,w)\mid p\in Q, w\in\Gamma^*, (p,w,\lambda)\in\delta_{M}(q,\dollar,a)\}$.

In a similar fashion, we can define the concept of ``$\track{\tilde{y}}{\tilde{z}}$ encodes stack history $y$ and query word $z$ along a computation path.''

\vs{-2}
\begin{proofof}{Corollary \ref{NFA-Dyck-equal-CFL}}
Let $L$ be any language defined as $L = \bigcap_{i\in[d]} A_i$ for  $k$ context-free languages $A_1,A_2,\ldots,A_k$. For each index $i\in[d]$, an npda $M_i$ is assumed to recognize $A_i$.  Consider an oracle nfa $N$ that, on input $x$, simulates several steps of $M_1,M_2,\ldots,M_d$ {\em in parallel} while they scan each input symbol.  During this simulation, instead of using a stack, $N$ writes a stack history of $M_i$ onto the $i$th track of its query tape. However, to adjust the speeds of $d$ tape heads, we appropriately insert the  symbol $\natural$. If  a query word $w$ correctly represents $d$ stack histories of $d$ machines, then $w$ must be in $DYCK^{ext}_{d}$ as discussed in the proof of Claim \ref{NFA-to-DYCK}.
\end{proofof}


For later use, we will generalize an argument used in the proof of Claim \ref{NFA-to-DYCK}.   We say that a language family $\CC$ is {\em  $\natural$-extendible} if, for every language $A$ in $\CC$, two special languages $A^{\natural}_1  =\{\track{\tilde{y}}{\tilde{z}}\mid z\in A\}$ and  $A^{\natural}_2  =\{\track{\tilde{y}}{\tilde{z}}\mid  y\in A\}$ also belong to $\CC$, where $\tilde{y}$ and $\tilde{z}$ are arbitrary  $\natural$-extensions of $y$ and $z$, respectively.

\begin{lemma}\label{extendible-inclusion}
Let $\CC$ be any nonempty language family. If $\CC$ is $\natural$-extendible, then $\cfl_{m}^{\CC}\subseteq \nfa_{m}^{\dcfl\wedge \CC}$ holds.
\end{lemma}

\begin{proof}
Take any oracle $A$ from $\CC$ and consider any language $L$ in $\cfl_{m}^{A}$. Moreover, take any $m$-reduction npda $M$ that reduces $L$ to $A$. With a similar construction as in the proof of Corollary \ref{NFA-Dyck-equal-CFL}, we intend to construct an oracle nfa $N$, which  guarantees the lemma. On input $x$, let $N$ simulate $M$ on $x$ and produce,  on its own query tape using no stack, strings of the form $\track{\tilde{y}}{\tilde{z}}$ that encode stack history $y$ and query word $z$ along a computation path of  $M$. Choose a Dyck language $D$ that correctly represents all stack histories of $M$ and then define $B$ as the set $ \{\track{\tilde{y}}{\tilde{z}}\mid y\in D, z\in A\}$. It is clear by its definition that $N$ $m$-reduces $L$ to $B$ and that $B$ belongs to $\dcfl\wedge \CC$. Therefore, $L$ is in $\nfa_{m}^{\dcfl\wedge\CC}$.
\end{proof}


Hereafter, let us explore basic properties of many-one $\cfl$-reducibility.
Unlike many-one NP-reducibility, the lack of the \emph{transitivity property} of many-one $\cfl$-reducibility necessitates an introduction of a helpful abbreviation of a {\em $k$-fold application of the reductions}. For any given oracle $A$,  we recursively set $\cfl_{m[1]}^{A} = \cfl_{m}^{A}$ and $\cfl_{m[k+1]}^{A} = \cfl_{m}(\cfl_{m[k]}^{A})$ for each index  $k\in\nat^{+}$.  Given each  language family $\CC$, the notation $\cfl_{m[k]}^{\CC}$ denotes the union $\bigcup_{A\in\CC}\cfl_{m[k]}^{A}$.
A close relationship between $\cfl(k)$'s and $\cfl_{m[k]}^{\cfl}$'s is exemplified below.

\begin{theorem}\label{CFL(k)-equal-m[k]}
For every index $k\in\nat^{+}$, $\cfl_{m}^{\cfl(k)} = \cfl_{m[k]}^{\cfl}$.
\end{theorem}

As an immediate consequence of Theorem \ref{CFL(k)-equal-m[k]}, the infinite union $\bigcup_{k\in\nat^{+}}\cfl_{m[k]}^{\cfl}$ has a succinct expression of $\cfl_{m}^{\cfl(\omega)}$.

\begin{corollary}
$\cfl_{m}^{\cfl(\omega)} = \bigcup_{k\in\nat^{+}}\cfl_{m[k]}^{\cfl} $.
\end{corollary}

\begin{proof}
Since $\cfl_{m[k]}^{\cfl} = \cfl_{m}^{\cfl(k)}$ by Theorem \ref{CFL(k)-equal-m[k]}, it follows that  $\bigcup_{k\in\nat^{+}}\cfl_{m[k]}^{\cfl} = \bigcup_{k\in\nat^{+}}\cfl_{m}^{\cfl(k)}$. It thus suffices to show that
$\bigcup_{k\in\nat^{+}}\cfl_{m}^{\cfl(k)} = \cfl_{m}^{\cfl(\omega)}$.  Since $\cfl(k)\subseteq \cfl(\omega)$ for any $k\in\nat^{+}$, we obtain $\bigcup_{k\in\nat^{+}}\cfl_{m}^{\cfl(k)} \subseteq \bigcup_{k\in\nat^{+}}\cfl_{m}^{\cfl(\omega)}$.  The last term coincides with $\cfl_{m}^{\cfl(\omega)}$ since $\cfl_{m}^{\cfl(\omega)}$ is independent of the value of $k$.
Conversely, let $L$ be any language in $\cfl_{m}^{\cfl(\omega)}$ and take an appropriate oracle $B\in \cfl(\omega)$ for which $L\in \cfl_{m}^{B}$. Since $\cfl(\omega) = \bigcup_{k\in\nat^{+}}\cfl(k)$,  $B$ must belong to  $\cfl(k)$ for a certain index $k\in\nat^{+}$.  This fact implies that $L\in\cfl_{m}^{B}\subseteq \cfl_{m}^{\cfl(k)}$. Therefore, we conclude
that $\cfl_{m}^{\cfl(\omega)}\subseteq \bigcup_{k\in\nat^{+}}\cfl_{m}^{\cfl(k)}$. This completes the proof of the corollary.
\end{proof}

Subsequently, we are focused on the proof of Theorem \ref{CFL(k)-equal-m[k]}. Notice that, when $k=1$, $\cfl_{m[1]}^{\cfl} = \cfl_{m}^{\cfl(1)} = \cfl_{m}^{\cfl}$.  The proof of Theorem \ref{CFL(k)-equal-m[k]} for $k\geq2$ is made up of two lemmas, Lemmas \ref{CFL(k)-in-many-one} and \ref{CFL(k)-included-CFL-m[k]}.

\begin{lemma}\label{CFL(k)-in-many-one}
For every number $k\geq2$, $\cfl_{m[k]}^{\cfl}\subseteq \cfl_{m}^{\cfl(k)}$ holds.
\end{lemma}

\begin{proof}
Let us consider the base case of $k=2$. Here, we claim the following inclusion relationship.

\begin{claim}\label{CFL_m[2]-property}
For every index $r\in\nat^{+}$, $\cfl_{m[2]}^{\cfl(r)} \subseteq \cfl_{m}^{\cfl(r+1)}$.
\end{claim}

\begin{proof}
With a certain language $A\in\cfl(r)$, let us assume that $L\in\cfl_{m[2]}^{A}$.  Furthermore, choose an appropriate set $B$ and let  two $m$-reduction npda's $M_1 = (Q_1,\Sigma,\{\cent,\dollar\},\Theta,\Gamma_1, \delta_1,q_0,Z_0,Q_{1,acc},Q_{2,rej})$ and $M_2 =( Q_2,\Theta,\{\cent,\dollar\},\Theta_2,\Gamma_2, \delta_2,q_0,Z_0,Q_{2,acc},Q_{2,rej})$ respectively witness the membership relations  $L\in\cfl_{m}^{B}$ and $B\in\cfl_{m}^{A}$.  We will define a new oracle npda $N$  in part using  a stack-history technique shown in the proof of Claim \ref{NFA-to-DYCK}.  Corresponding to $\Gamma_1$, we prepare an associated alphabet $\Gamma'_1=\{\tau'\mid \tau\in \Gamma_1\}$ and we set  $\Gamma_1\cup\Gamma'_1$  to be a stack alphabet $\Gamma$ for $N$. On input $x$, $N$ simulates $M_1$ on $x$ using $M_2$ as a subroutine in the following way.  Whenever $M_1$ tries to write a symbol, say, $b$ on a query tape, $N$ instead simulates,  without using any actual stack, one or more steps  (including a certain number of $\lambda$-moves) of $M_2$ that can be  made after scanning $b$.  When $M_2$ tries to push down a string $w$ after removing the top symbol of its stack content, $N$ guesses this top symbol, say, $\tau$ and then produces $\tau' w^R$ on the {\em upper track} of its query tape. When $M_2$ pops up a symbol, $N$ guesses this  popped symbol, say, $\tau$ and writes $\tau'$ on the {\em upper track} of the query tape. At the same time during this  simulation, $N$ produces $M_2$'s query word on the {\em lower track} of the query tape. To fill the idling time of certain tape heads, we need to insert an appropriate number of symbols $\natural$ so that $\track{\tilde{y}}{\tilde{z}}$ encodes stack history $y$ and query word $z$ along a computation path of $M_2$.

Here, we give a formal definition of $N$ only for the case of  $r=1$.
Choose a constant $k\in\nat^{+}$ for which $(p,w,c)\in\delta_2(q,b,\tau)$ implies $|w|\leq k$ for all possible tuples $(p,w,c,a,b,\tau)$.
Let $N = (Q,\Sigma,\{\cent,\dollar\},\Theta,\Gamma, \delta_N,\overline{q}_0,Z_0,Q_{acc},Q_{rej})$ with $Q=[2]\times Q'_1\times Q'_2$, where $Q'_1=\{q_0\}\cup \{\track{q}{b}\mid q\in Q,b\in\check{\Theta}\cup\{\lambda\}$ and $Q'_2=\{q_0\}\cup\{\track{q}{w}\mid q\in Q_2,w\in\Gamma_2^{\leq k}\}$.

For simplicity, assume that $Q_{1,acc}= Q_{2,acc}=\{q_{acc}\}$. Since $M$ is well-behaved at $\cent$,  $(p,s,b)\in\delta_1(q_0,\cent,Z_0)$ implies $b=\lambda$. Let $\overline{q}_0 = (1,q_0,q_0)$ and  $\delta_N(\overline{q}_0,\cent,Z_0) = \{((2,\track{p_1}{\cent},q_0),s,\lambda) \mid p_1\in Q_1, s\in\Gamma_1^*, (p_1,s,\lambda)\in\delta_1(q_0,\cent,Z_0)\}$. For any $\sigma\in\Sigma\cup\{\lambda,\dollar\}$, let $\delta_N((1,\track{q_1}{\lambda},\track{q_2}{\lambda}),\sigma,a)$ contain $((1,\track{p_1}{\lambda},\track{q_2}{\lambda}),s,\lambda)$ if $(p_1,s,\lambda)\in\delta_1(q_1,\sigma,a)$, and $((2,\track{p_1}{b},\track{q_2}{\lambda}),s,\lambda)$ if $(p_1,s,b)\in\delta_1(q_1,\sigma,a)$ with $b\in\Theta$.
We set $((2,\track{q_{1}}{\dollar},\track{q_2}{\lambda}),a,\lambda)\in \delta_N((1,\track{q_{1}}{\lambda},\track{q_2}{\lambda}),\sigma,a)$ if $q_1\notin Q_{1,halt}$.
Moreover, for $b\in\check{\Theta}$ and $q_1\in Q_1$ with
$\sigma\neq\dollar$,
let $\delta_N((2,\track{q_1}{b},\track{q_2}{\lambda}),\sigma,a)$ contain
$((2,\track{q_1}{\lambda},\track{p_2}{w}),a,\track{\tau'}{t(c)})$ if $(p_2,w,c)\in\delta_2(q_2,b,\tau)$ for all $\tau\in\Gamma_2$, where $t(\lambda)=\natural$ and $t(c)=c$ if $c\neq\lambda$. In this case, if $w=w_1w_2\cdots w_m\in\Gamma^{\leq k}-\{\lambda\}$, then we set  $((2,\track{q_1}{\lambda},\track{q_2}{w_1\cdots w_{i-1}}),a,\track{w_i}{\natural}) \in \delta_N((2,\track{q_1}{\lambda},\track{q_2}{w_1\cdots w_i}),\lambda,a)$ for each $i\in[|w|]$. Let  $\delta_N((2,\track{q_1}{\lambda},\track{q_2}{\lambda}),\lambda,a)$ contain $(1,\track{q_1}{\lambda},\track{q_2}{\lambda}),a,\lambda)$ if $q_1\notin Q_{1,halt}$, and $(2,\track{q_1}{\lambda},\track{p_2}{w}),a,\track{\tau'}{t(c)})$ if $(p_2,w,c)\in\delta_2(q_2,\lambda,\tau)$ for all $\tau\in\Gamma_2$.

If $(p_1,s,\lambda)\in\delta_1(q_1,\dollar,a)$ and $(p_2,\tau,\lambda)\in \delta_2(q_2,\dollar,\tau)$ for all $\tau\in\Gamma_2$, then $((2,\track{p_1}{\lambda},\track{p_2}{\lambda}),a,\lambda)\in \delta_N((1,\track{q_1}{\lambda},\track{q_2}{\lambda}),\dollar,a)$, since $M$ is well-behaved at $\dollar$.
Let all other transitions go to appropriate rejecting states.
Let $Q_{acc} = \{(2,\track{q_{acc}}{\lambda},\track{q_{acc}}{\lambda})\}$ and $Q_{rej} = \{(2,\track{q_{rej}}{\lambda},\track{q_2}{\lambda}), (2,\track{q_1}{\lambda},\track{q_{rej}}{\lambda})  \mid q_1\in Q_{1}, q_2\in Q_2\}$.

Finally, we define oracle $C$ as a collection of strings of the form $\track{\tilde{y}}{\tilde{z}}$ for which $\tilde{y}$ (resp., $\tilde{z}$) is a $\natural$-extension of a correct stack history $y$ (resp., of a valid query string $z$ in $A$).
Since the above simulation of $M_1$ and $M_2$ ensures the correctness of $y$ and $z$,  $C$  belongs to $\cfl\wedge \cfl(r) = \cfl(r+1)$. Moreover, for every string $x$, $x$ is in $L$ if and only if there exists a computation path in $ACC_{N}(x)$, along which $N$ produces $\track{\tilde{y}}{\tilde{z}}$ in $C$. This equivalence relation implies that $L$ belongs to $\cfl_{m}^{C}$, which is a subclass of $\cfl_{m}^{\cfl(r+1)}$.
\end{proof}

For the case of $k\geq3$, it holds that $\cfl_{m[k]}^{\cfl} = \cfl_{m}(\cfl_{m[k-1]}^{\cfl}) \subseteq \cfl_{m}(\cfl_{m}^{\cfl(k-1)})$, where the last inclusion comes from our induction hypothesis: $\cfl_{m[k-1]}^{\cfl}\subseteq \cfl_{m}^{\cfl(k-1)}$. Since  $\cfl_{m}(\cfl_{m}^{\cfl(k-1)}) = \cfl_{m[2]}^{\cfl(k-1)}$, we obtain $\cfl_{m[k]}^{\cfl} \subseteq \cfl_{m[2]}^{\cfl(k-1)}$.
Claim \ref{CFL_m[2]-property} then yields the containment  $\cfl_{m[2]}^{\cfl(k-1)} \subseteq  \cfl_{m}^{\cfl(k)}$, from which the claim follows.
\end{proof}

\begin{lemma}\label{CFL(k)-included-CFL-m[k]}
For every index $k\geq1$, $\cfl_{m}^{\cfl(k)}\subseteq \cfl_{m[k]}^{\cfl}$.
\end{lemma}

\begin{proof}
Using induction on $k\geq1$, we intend to prove the lemma.  Since  the lemma is trivially true for $k=1$,  let  us assume that $k\geq2$. Since $\cfl(k)\subseteq \cfl_{m}^{\cfl(k-1)}$ by Claim \ref{CFL(2)-bound}, it instantly follows that $\cfl_{m}^{\cfl(k)}\subseteq \cfl_{m}(\cfl_{m}^{\cfl(k-1)})$. Moreover, because $\cfl_{m}^{\cfl(k-1)}\subseteq \cfl_{m[k-1]}^{\cfl}$ holds by our induction hypothesis,  we conclude that  $\cfl_{m}^{\cfl(k)} \subseteq \cfl_{m}(\cfl_{m[k-1]}^{\cfl})=\cfl_{m[k]}^{\cfl}$.
\end{proof}


Toward the end of this section, we will make a brief discussion on a relationship between  two  language families $\cfl/n$ and $\cfl_{m}^{\tally}$. In comparison, it is known that $\p/poly = \p^{\tally}$ (see, \eg \cite{DK00}). Given a language $A$ over alphabet $\Sigma$, the notation $dense(A)(n)$ indicates  the cardinality $\|A\cap\Sigma^n\|$ for every length $n\in\nat$.  Let $DENSE(f(n))$ be the collection of all languages $A$ such that  $dense(A)(n)\leq f(n)$ holds for all lengths $n\in\nat$. Note that $\tally \subseteq DENSE(O(1)) \subseteq \sparse$, where $\sparse = DENSE(n^{O(1)})$.
For the proof of Proposition \ref{TALLY-advice-impossible}, recall the notation of $[x_1,x_2,\ldots,x_k]^T$.

\begin{proposition}\label{TALLY-advice-impossible}
\begin{enumerate}
  \setlength{\topsep}{-2mm}%
  \setlength{\itemsep}{1mm}
  \setlength{\parskip}{0cm}

\item $\cfl/n\subseteq \cfl_{m}^{DENSE(O(1))}$.
\item $\cfl_{m}^{\tally}\subseteq \cfl_{m}^{\cfl(2)}/n$.
\end{enumerate}
\end{proposition}

\begin{proof}
(1)  This is rather obvious by taking any length-preserving advice function $h$ and define $A=\{h(n)\mid n\in\nat\}$, which belongs to $DENSE(1)$ by the definition.

(2) Let $L\in\cfl_{m}^{A}$ for a certain $A$ in $\tally$ over alphabet $\Sigma$. Let $M$ be an $m$-reduction npda that reduces $L$ to $A$. Without loss of generality, we assume that $1\in\Sigma$ and $A\subseteq \{1\}^*$.  To simplify the following description,  we impose on $M$ a restriction that $\lambda\not\in A$. Next, let us define a language $B$ as follows. Let $p(n)=an$ be a linear polynomial that bounds the running time of $M$ on inputs of length $n\geq1$, where $a\in\nat^{+}$.
We define our advice function $h$ as $h(n) = h_1h_2\cdots h_n$  for any number $n\in\nat^{+}$, where each symbol $h_i$ equals $[\chi^A(1^{(i-1)a+1}),\chi^A(1^{(i-1)a+2})\ldots,\chi^A(1^{ia})]^T$.  Note that $|h(n)|=n$ for any $n\in\nat^{+}$.
The desired language $B$ must satisfy $L = \{x\mid \track{x}{h(|x|)}\in B\}$.

Next, we aim at showing that $B\in\cfl_{m}(\cfl_{m}^{\cfl}) = \cfl_{m[2]}^{\cfl}$ by constructing three appropriate machines $N_1$, $N_2$, and $N_3$. On input $\track{x}{s}$ with $|s|=|x|$, $N_1$ simulates $M$ on $x$ and generates a query string $\track{\tilde{y}}{\tilde{s}}$ if $M$ makes a query $y$, where $\tilde{y}$ and $\tilde{s}$ are respectively $\natural$-extensions of $y$ and $s$. Moreover, in this simulation process, if $y\not\in\{1\}^*$, then $N_1$ immediately enters a rejecting state.  Another oracle npda $N_2$ works as follows. On input of the form  $\track{\tilde{y}}{\tilde{s}}$, using a stack appropriately, $N_2$ removes all $\natural$s and produces $y\# s$ on its query tape, where $\#$ is a fresh symbol. The third machine $N_3$,  taking  $y\# s$ as input, finds the $i$th block $h_i = [b_1,b_2,\ldots,b_a]^T$ of $s$, where $i=\ceilings{|y|/a}$,  and reads a symbol $b_j$, where $j=|y|-(i-1)a$.  If $b_j=1$, then $N_3$ enters an accepting state and, otherwise, it enters a rejecting state. This whole process puts $B$ to $\cfl_{m}(\cfl_{m}^{\cfl})$, which is $\cfl_{m[2]}^{\cfl}$. Theorem \ref{CFL(k)-equal-m[k]} implies that $B$ belongs to $\cfl_{m}^{\cfl(2)}$. It is not difficult to show by the definition of $h$ and $B$ that, for any string $x$, $x\in L$ if and only if $\track{x}{h(|x|)}\in B$. Therefore, $L$ belongs to $\cfl_{m}^{\cfl(2)}/n$.
\end{proof}

\subsection{Turing and Truth-Table Reducibilities by Oracle Npda's}\label{sec:tt-reduction}

In the previous sections, our primary interest has been rested on the many-one $\cfl$-reducibility. It is also possible to introduce two more powerful reducibilities, known as truth-table reducibility and Turing reducibility, into a theory of context-free languages.

Firstly, we define a notion of {\em Turing CFL-reducibility} using a variant of npda, equipped with a write-only query tape and three extra inner states $q_{query}$,  $q_{no}$, and $q_{yes}$ that represent a query signal and two different oracle answers, respectively.  More specifically, when such a machine enters $q_{query}$, it triggers a \emph{query}, by which (1) a query word produced on a query tape is automatically transmitted to a given oracle, (2) the query tape instantly becomes blank, and (3) its query tape head also  returns to the start cell. The machine then waits for a reply from the oracle.
We informally say that the oracle \emph{returns} (or \emph{replies}) \emph{an answer}, either $0$ (no) or $1$ (yes), if the oracle automatically resets the machine's inner state to $q_{no}$ or $q_{yes}$, accordingly. After the oracle replies, the machine resumes its computation starting with inner state $q\in\{q_{yes},q_{no}\}$ reset by the oracle.  This machine is also called an  {\em oracle npda} as before, and it
is used to reduce a target language to another language.
To be more precise, an oracle npda $M$ is
a tuple $(Q,\Sigma,\{\cent,\dollar\},\Theta,\Gamma,\delta,q_0,Z_0,Q_{oracle}, Q_{acc},Q_{rej})$, where $Q_{oracle} = \{q_{query},q_{yes},q_{no}\}$, $\Theta$ is a query alphabet, and $\delta$ has the form:
\[
\delta:(Q-Q_{halt}\cup\{q_{query}\})\times (\check{\Sigma}\cup\{\lambda\})\times \Gamma \rightarrow \PP((Q-\{q_{yes},q_{no}\})\times \Gamma^* \times (\Theta\cup\{\lambda\})).
\]

Unlike $m$-reduction npda's, each computation of the above oracle npda's  depends on a series of replies (called \emph{oracle answers}) from the given oracle. Since such oracle npda's, in general, cannot physically implement an internal clock to control their running time, certain oracle answers may lead to an extremely long computation, and thus the machine may result in recognizing even ``infeasible'' languages. To avoid such a pitfall, we need to demand that, {\em no matter what oracle is provided}, its underlying oracle npda $M$ must halt on {\em all} computation paths within $O(n)$ time, where $n$ refers to input size. As noted in Section \ref{sec:npda-and-TM}, we always demand that $M$ should be well-behaved at
both $\cent$ and $\dollar$.
We say that a language $L$ is \emph{Turing CFL-reducible to} $A$ (or an oracle npda $M$ \emph{T-reduces} $L$ to $A$) if, for every string $x$, $x\in L$ iff $M$ accepts $x$ using $A$ as its oracle (or $M$ accepts $x$ relative to $A$).

Similarly to $\cfl_{m}^{A}$ and $\cfl_{m}^{\CC}$, we introduce two new notations $\cfl_{T}^{A}$ and $\cfl_{T}^{\CC}$ as the families of all languages that are Turing $\cfl$-reducible to $A$ and languages in $\CC$, respectively. Notice that the requirement of $O(n)$ time-bound for our oracle npda's naturally implies that $\cfl_{T}^{A}\subseteq\np^{A}$ for any oracle $A$.

An associated $T$-reduction dpda is defined using its transition function that maps $(Q-Q_{halt}\cup\{q_{query}\})\times (\check{\Sigma}\cup\{\lambda\})\times \Gamma$ to $(Q-\{q_{yes},q_{no}\})\times \Gamma^* \times (\Theta\cup\{\lambda\})$. It is important to note that, unlike $m$-reduction dpda's, this machine $M$ does not need to write the termination symbol on a query tape, because $M$ is well-behaved at $\dollar$ and therefore it must enter a query state $q_{query}$ before reading the right endmarker $\dollar$. A languages family defined by $T$-reduction dpda's relative to oracle $A$ is denoted  $\dcfl_{T}^A$ (or $\dcfl_{T}(A)$). For any class $\CC$ of oracles, we set $\dcfl_{T}^{\CC} = \dcfl_{T}(\CC) = \bigcup_{A\in\CC}\dcfl_{T}^{A}$.

\begin{lemma}\label{basic-Turing}
For any oracle $A$, (1) $\cfl_{m}^{A}\subseteq \cfl_{T}^{A} = \cfl_{T}^{\overline{A}}$ and (2)  $\dcfl_{m}^{A}\subseteq \dcfl_{T}^{A} = \dcfl_{T}^{\overline{A}}$ .
\end{lemma}

\begin{proof}
(1) The first containment is rather obvious because  Turing $\cfl$-reducibility can naturally ``simulate''  many-one $\cfl$-reducibility as follows.
Assume that an $m$-reduction npda $M = (Q_M,\Sigma,\{\cent,\dollar\},\theta,\Gamma, \delta_M,q_0,Z_0,Q_{acc},Q_{rej})$  writes a query string on its query tape.  Another $T$-reduction npda $N = (Q_N,\Sigma,\{\cent,\dollar\},\Theta,\Gamma, \delta_N,q_0,Z_0,Q_{oracle},Q_{acc},Q_{rej})$ rejects an input if $q$ is a rejecting state, where $Q_N= Q_M\cup\{\overline{p}_0\}$. Otherwise, $N$ makes a query $w$ and decides to accept and reject the input if its oracle answer is ``yes'' and ``no,'' respectively.
When $(p,w,\tau)\in\delta_M(q,\sigma,a)$, let $(p,w,\tau)\in\delta_N(q,\sigma,a)$ for $\sigma\in\Sigma\cup\{\cent,\lambda\}$. When  $(p,a,\lambda)\in\delta_M(q,\dollar,a)$, let $(\overline{p}_0,a,\lambda)\in\delta_N(q,\lambda,a)$ and $(q_{query},a,\lambda)\in\delta_N(\overline{p}_0,\lambda,a)$. Moreover, let $(q_{acc},a,\lambda)\in\delta_N(q_{yes},\dollar,a)$ and $(q_{rej},a,\lambda)\in\delta_N(q_{no},\dollar,a)$. For all other transitions, let $N$ enter rejecting states.

To see the last equality, let $L$ be any language in $\cfl_{T}^{A}$, witnessed by a $T$-reduction npda $M$.  Let us show that $L\in \cfl_{T}^{\overline{A}}$ via another $T$-reduction npda $N$. This  machine $N$ is defined to behave as follows. Given any input $x$, $N$ simulates $M$ on $x$ and, whenever $M$ receives an oracle answer, say, $b\in\{0,1\}$, $N$ treats it as if $\overline{b}$ ($=1-b$) and continues the simulation of $M$. This definition clearly implies that $N$ accepts $x$ relative to $\overline{A}$ if and only if $M$ accepts $x$ relative to $A$. Thus, $L$ is in $\cfl_{T}^{\overline{A}}$.  We therefore conclude that $\cfl_{T}^{A}\subseteq \cfl_{T}^{\overline{A}}$. By symmetry, we also
obtain $\cfl_{T}^{\overline{A}}\subseteq\cfl_{T}^{A}$, implying that $\cfl_{T}^{A} = \cfl_{T}^{\overline{A}}$.

(2) A basic proof idea is the same as (1). Nonetheless, we briefly comment on $\dcfl_{m}^{A}\subseteq \dcfl_{T}^{A}$ for any oracle $A$. The difference from (1) is exemplified below. Given an $m$-reduction dpda $M$, here we want to construct a new $T$-reduction dpda that simulates $M$. By the definition of $m$-reduction dpda's, $M$ must write the termination symbol, $\dollar$, before reading the right endmarker $\dollar$. Without this particular requirement, we do not know whether $\dcfl_{m}^A\subseteq \dcfl_T^A$. Let $\delta_M$ and $\delta_{N}$ denote transition functions of $M$ and of $N$, respectively. If $\delta_{M}(q,\sigma,a) = (p,w,\dollar)$, then we set $\delta_{N}(q,\sigma,a) = (\overline{p}_0,w,\lambda)$ and $\delta_{N}(\overline{p}_0,\lambda,a) = (q_{query},a,\lambda)$. Moreover, let $\delta_N(q_{yes},\sigma,a)= (q_{acc},a,\lambda)$ and $\delta_N(q_{no},\sigma,a) = (q_{rej},a,\lambda)$. It is not difficult to see that $N$ simulates $M$.
\end{proof}

A simple relationship between Turing and many-one $\cfl$-reducibilities is exemplified in Proposition \ref{first-level-equal}. To describe this  proposition, we need a notion of the {\em Boolean hierarchy over CFL}. Earlier, Yamakami and Kato \cite{YK13} introduced the Boolean hierarchy over the family of bounded context-free languages. In a similar fashion, we set  $\cfl_{1}=\cfl$, $\cfl_{2k}=\cfl_{2k-1}\wedge \co\cfl$, and $\cfl_{2k+1} = \cfl_{2k}\vee\cfl$.
Finally, we denote by $\bhcfl$ the infinite union $\bigcup_{k\in\nat^+}\cfl_{k}$.
In particular, $\cfl_{2}=\cfl\wedge \co\cfl$ holds, and thus $\cfl\neq \cfl_{2}$ follows because $\co\cfl\subseteq\cfl_{2}$ and $\co\cfl\nsubseteq\cfl$.
\begin{proposition}\label{first-level-equal}
$\cfl_{T}^{\cfl} = \cfl_{m}^{\cfl_{2}} = \nfa_{m}^{\cfl_{2}}$.
\end{proposition}

For the proof of this proposition, the following notation is required. If $M$ is an (oracle) npda, then $\overline{M}$ denotes an (oracle) npda obtained from $M$ simply by exchanging between accepting states and rejecting states.

\vs{-2}
\begin{proofof}{Proposition \ref{first-level-equal}}
In this proof, we will demonstrate that (1)  $\cfl_{T}^{\cfl}\subseteq \cfl_{m}^{\cfl_2}$, (2)  $\cfl_{m}^{\cfl_2}\subseteq \nfa_{m}^{\cfl_2}$ , and (3) $\nfa_{m}^{\cfl_2}\subseteq \cfl_{T}^{\cfl}$.  If all are proven, then the proposition immediately follows.

(1)  We start with an arbitrary language $L$ in $\cfl_{T}^{A}$ relative to a certain language $A$ in $\cfl$. Take an oracle npda $M = (Q_M,\Sigma,\{\cent,\dollar\},\Theta,\Gamma_M, \delta_M,q_0,Z_0,Q_{oracle},Q_{M,acc},Q_{M,rej})$ that $T$-reduces $L$ to $A$, and let $M_A =(Q_A,\Theta,\{cent,\dollar\},\Gamma_A, \delta_A,q_0,Z_0,\{q_{acc}\},\{q_{rej}\})$ be an npda recognizing $A$.  Hereafter, we will build three new machines $N_1$, $N_2$, and $N_3$ to show that $L\in\cfl_{m}^{\cfl_2}$.
The first machine $N_1$ is an $m$-reduction npda trying to  simulate $M$ on input $x$ by running the following procedure. Along each computation path in $ACC_{M}(x)$, before $M$ begins producing the $i$th query word on a query tape, $N_1$ guesses its oracle answer $b_i$   (either $0$ or $1$)   and writes it down onto its query tape. While $M$ writes the $i$th query word $y_i$, $N_1$  does the same and then appends $\natural$ to $y_i$.
Recall that $M$ halts in linear time, no matter what answers it receives from oracles; thus, $N_1$ also halts in linear time. When $M$ halts, $N_1$ produces query words $w$ of the form $b_1y_1\natural b_2y_2\natural\cdots \natural b_ky_k\natural$ for certain numbers $k\in\nat$.

More formally, let $N_1 = (Q_1,\Sigma,\{\cent,\dollar\},\Theta_1,\Gamma_1,\delta_1,q_0,Z_0)$, where $Q_1 = \{q_0\}\cup (Q_{M}\times \{0,1\}\times (\Theta\cup\{\lambda,*\}))$.  Let $((p,b,\tau),w,b)\in\delta_1(q_0,\cent,Z_0)$ for all $b\in\{0,1\}$ if $(p,w,\tau)\in\delta_{M}(q_0,\cent,Z_0)$ with $\tau\in\Theta$.
Similarly, let $((p,b,\lambda),w,\lambda)\in\delta_1(q_0,\cent,Z_0)$ if $(p,w,\lambda)\in\delta_{M}(q_0,\cent,Z_0)$. Moreover, let $((p,b,\tau),w,b)\in\delta_1((q,b,\lambda),\sigma,a)$ if $(p,w,\tau)\in\delta_{M}(q,\sigma,a)$ with $\tau\in\Theta\cup\{\lambda\}$ and $\sigma\in\Sigma$. Let $((p,b,*),a,\tau)\in\delta_1((q,b,\tau),\lambda,a)$. Let $((p,b,*),w,\tau)\in\delta_1((q,b,*),\sigma,a)$ if $(p,w,\tau)\in\delta_{M}(q,\sigma,a)$.
When $(p,w,\lambda)\in\delta_M(q,\dollar,a)$ with $p\in Q_{M,halt}$, let $((p,b,*),w,\natural)\in \delta_1((q,b,*),\dollar,a)$ and $((p,b,*),w,\lambda)\in \delta_1((q,b,\lambda),\dollar,a)$.
Assume that $(p,w,\tau)\in\delta_{M}(q_{e},\sigma,a)$ with
$e\in\{yes,no\}$.
Let $\delta_{1}((q_{query},0,\lambda),\lambda,a) = \{ ((q_{no},b,\lambda),a,\natural)\mid b\in\{0,1\}\}$ and $\delta_{1}((q_{query},1,\lambda),\lambda,a)= \{((q_{yes},b,\lambda),a,\natural)\mid b\in\{0,1\}\}$.
For all other cases, $N_1$ enters its own rejecting states. Let $Q_{acc} = \{(q,b,*)\mid q\in Q_{M,acc},b\in\{0,1\}\}$ and $Q_{rej} = \{(q,b,*)\mid q\in Q_{M,rej},b\in\{0,1\}\}$.

The second machine $N_2$ is an npda that works as follows. On input $w$ of the form $b_1y_1\natural b_2y_2\natural\cdots \natural b_ky_k\natural$, $N_2$ executes the following procedure. Choosing each index $i\in[k]$ sequentially, $N_2$ simulates $M_A$ on $y_i$ if $b_i=1$, and $N_2$ skips string $y_i$  otherwise, and move to the next index. During this simulation process, whenever $M_A$ enters a rejecting state, $N_2$ also enters a rejecting state and halts.  The third machine $N_3$ takes input $w$ and, for each $i\in[k]$, $N_3$ simulates $\overline{M_A}$ on $y_i$ if $b_i=0$, and $N_3$ skips $y_i$ otherwise. If $\overline{M}_{A}$ halts in a rejecting state, then $N_3$ also rejects $w$.
It is not difficult to verify that $N_1$ $m$-reduces $L$ to $L(N_2)\wedge \overline{L(N_3)}$. This leads to a conclusion that $L$ is included in $\cfl_{m}(\cfl\wedge \co\cfl) = \cfl_{m}^{\cfl_{2}}$.

(2) Note that $\cfl_2$ is $\natural$-extendible. Proposition \ref{extendible-inclusion} implies that $\cfl_{m}^{\cfl_2}\subseteq \nfa_{m}^{\dcfl\wedge \cfl_2}$. Since $\dcfl\subseteq \co\cfl$, it follows that $\dcfl\wedge \cfl_{2} \subseteq  \co\cfl\wedge (\cfl\wedge \co\cfl) = (\co\cfl\wedge\co\cfl)\wedge \cfl$. The last term clearly equals $\co\cfl\wedge \cfl = \cfl_2$, and  thus  we conclude that $\cfl_{m}^{\cfl_2}\subseteq \nfa_{m}^{\cfl_2}$.

(3) Choose an oracle $A$ in $\cfl_{2}$ and consider any language $L$ in $\nfa_{m}^{A}$.  Because of $A\in \cfl_{2}$,  we can take two languages $A_1,A_2\in\cfl$ over the same alphabet, say, $\theta$ for which $A=A_1\cap \overline{A}_2$. Let $M =(Q,\Sigma,\{\cent,\dollar\},\Theta,\delta_M,q_0,Q_{M,acc},Q_{M,rej})$ be an oracle nfa that $m$-reduces $L$ to $A$. Remember that $M$ has no stack.
Let us define a new $T$-reduction npda $N$, which makes two queries to an oracle. On input $x$, $N$  first marks $0$ on its query tape and starts simulating $M$ on $x$. Whenever $M$ tries to write a symbol $\tau$ on its query tape, $N$  writes it down on a query tape and simultaneously copies it into a stack. After $M$ halts with a query word, say, $w$, $N$ makes the first query with the query word $0w$.  If its oracle answer
is $0$ (no), then $N$ rejects the input. Otherwise, $N$ writes $1$ on the query tape (after the tape automatically becomes blank), pops up the stored string $w$ from the stack, and copies it (in a reverse form) onto the query tape. After making the second query with $1w^R$,  if its oracle answer equals $1$ (yes),  then $N$ rejects the input. When $N$ has not entered any rejecting state, then $N$ must accept the input.

Formally, let $N = (Q_N,\Sigma,\{\cent,\dollar\},\Theta_N,\Gamma, \delta_N,\overline{q}_0,Z_0,Q_{oracle},Q_{N,acc},Q_{N,rej})$.
For simplicity, we assume that $M$ does not query the empty word $\lambda$.
Let $Q_N= Q\cup\{\overline{q}_0,\overline{p}_0,\overline{p}_1\}$ and let $\Theta_N= \Theta\cup\{0,1\}$. Consider $T_M=\{q\in Q_M\mid \exists p\in Q_{M,acc}\,[(p,\lambda)\in \delta_M(q,\dollar)]\}$.
Let $\delta_{N}(q_0,\cent,Z_0) = \{(\overline{q}_0,Z_0,0)\}$ and $\delta_{N}(\overline{q}_0,\lambda,Z_0) =\{(p,\tau Z_0,\tau)\mid \tau\in\Theta\cup\{\lambda\},  (p,\tau)\in\delta_{M}(q_0,\cent)\}$.  Next, assume that  $\sigma\in\Sigma\cup\{\lambda\}$,  $\tau\in\Theta\cup\{\lambda\}$, and $a\in\Sigma$.
Let $\delta_{N}(q,\sigma,a) = \{(p,\tau a,\tau)\mid  (p,\tau)\in\delta_{M}(q,\sigma)\}$. For any $q\in T_M$, let $(q_{query},a,\lambda)\in\delta_{N}(q,\lambda,a)$. Let $\delta_{N}(q_{no},\lambda,a) =\{(q_{rej},a,\lambda)\}$. Moreover, for $a\neq Z_0$, let  $\delta_{N}(q_{yes},\lambda,a)=\{(\overline{p}_0,a,1)\}$, $\delta_{N}(\overline{p}_0,\lambda,a) = \{(\overline{p}_0,\lambda,a)\}$, $\delta_{N}(\overline{p}_0,\lambda,Z_0) =\{(q_{query},Z_0,\lambda)\}$, $\delta_{N}(q_{yes},\lambda,Z_0) = \{(q_{rej},Z_0,\lambda)\}$, and $\delta_{N}(q_{no},\lambda,Z_0) =\{(\overline{p}_1,Z_0,\lambda)\}$,  $(q_{acc},Z_0,\lambda)\in \delta_N(\overline{p}_1,\dollar,Z_0)$, and $(q_{rej},a,\lambda)\in \delta_N(\overline{p}_1,\sigma,a)$ with $\sigma\neq\dollar$, where $\overline{p}_0\notin Q$. Let $Q_{N,acc} = \{q_{acc}\}$ and $Q_{N,rej} = \{q_{rej}\}$.

The corresponding oracle $B$ is defined as $\{0w\mid w\in A_1\}\cup \{1w^R\mid w\in A_2\}$.   It is easy to see that $x\in L$ if and only if $N$ accepts $x$ relative to $B$. Note that the language  $\{0w\mid w\in A_1\}$ is context-free. Since $\cfl$ is closed under \emph{reversal} (see, \eg \cite{HMU01}), the language  $\{1w^R\mid w\in A_2\}$ is also context-free; therefore,  $B$ belongs to  $\cfl$. We thus conclude that $L\in\cfl_{T}^{B}\subseteq \cfl_{T}^{\cfl}$.
\end{proofof}


As another useful reducibility, we are focused on {\em truth-table CFL-reducibility}. Notice that an introduction of nondeterministic truth-table reducibility to context-free languages does not seem to be as obvious as that of the aforementioned Turing $\cfl$-reducibility.  Ladner, Lynch, and Selman \cite{LLS75} first offered a notion of polynomial-time nondeterministic truth-table  reducibility. Another definition, which is apparently  weaker than that of Ladner \etalc, was later proposed by Book, Long, and Selman \cite{BLS84}  as well as  Book and Ko \cite{BK88}. The next definition follows a spirit of Ladner \etalc~\cite{LLS75} with a slight twist for its evaluator in order to accommodate our oracle npda's.

Letting $k\in\nat^{+}$, a language $L$ is in $\cfl_{ktt}^{A}$ (or $\cfl_{ktt}(A)$) if  there are a regular language $B$ and an npda $N = (Q,\Sigma,\{\cent,\dollar\},\Theta,\Gamma,\delta,q_0,Z_0,Q_{acc},Q_{rej})$  having  $k$ write-only output tapes besides a single read-only input tape such that, for any input string $x$,  (1) $ACC_{M}(x)\neq\setempty$,  (2) along every computation path $p\in ACC_{N}(x)$, for each index $i\in[k]$, $N$ produces a string  $y^{(i)}_{p}\in\Theta^*$ on the $i$th write-only query tape,   (3)  to a vector  $(y^{(1)}_p,y^{(2)}_p,\ldots,y^{(k)}_p)$ of such query words,  we assign a $k$-bit string $z_p = \chi_k^{A}(y^{(1)}_p,y^{(2)}_p,\ldots,y^{(k)}_p)\in\{0,1\}^k$, and (4) $x$ is in $L$ if and only if
$\track{x}{z_p}$ is in $B$ for an appropriate computation path $p\in ACC_{N}(x)$. Remember that, by our convention for the track notation, $\track{x}{z_p}$ is a shorthand for $\track{x}{z_p\#^{m}}$ if $|x|\geq k$ and for $\track{x\#^m}{z_p}$ if $|x|<k$, where with $m=||x|-k|$.
A transition function $\delta$ of $N$ has the form:
\[
\delta: (Q-Q_{halt})\times (\check{\Sigma}\cup\{\lambda\})\times \Gamma \to \PP( Q\times \Gamma^*\times (\Theta\cup\{\lambda\})^k ).
\]
The set $B$ is called a {\em truth table} for $A$. For the sake of convenience, we often treat $B$ as a characteristic function defined as $B(x,z)=1$ if $\track{x}{z}\in B$ and $B(x,z)=0$ otherwise.
For readability, we often write $B(x,z_p)$ instead of $B(x,z_p\#^n)$.
The machine $N$ is in general called a \emph{bounded-truth-table (btt) CFL-reduction} from $L$ to $A$.
Since $B\in\reg$, we also treat as a truth table a dfa that recognizes (or computes) $B$, instead of $B$ itself.
In the end,  we set $\cfl_{btt}^{A}$ (or $\cfl_{btt}(A)$) to be the union $\bigcup_{k\in\nat^{+}}\cfl_{ktt}^{A}$. Similarly, we can define $\nfa_{btt}^{A}$.

It is not difficult to show $\cfl_{m}^A\cup\cfl_{m}^{\overline{A}}\subseteq \cfl_{1tt}^A$.  By flipping the outcome (\ie an accepting or a rejecting state) of a computation generated by each dfa that computes its truth-table $B$, we obtain $\cfl_{ktt}^{A}\subseteq \cfl_{ktt}^{\overline{A}}$.  In symmetry, $\cfl_{ktt}^{\overline{A}}\subseteq \cfl_{ktt}^{A}$ also holds.  Therefore, the statement given below follows immediately.

\begin{lemma}\label{m-vs-ktt-compare}
For every language $A$ and index $k\geq1$, $\cfl_{m}^{A}\cup \cfl_{m}^{\overline{A}} \subseteq \cfl_{ktt}^{A} = \cfl_{ktt}^{\overline{A}}$.
\end{lemma}

Unlike $\np$, we do not know whether  $\cfl_{btt}^{\cfl}\subseteq \cfl_{T}^{\cfl}$ holds. This is mainly because of a restriction on the usage of npda's memory device.
It may be counterintuitive that Turing reducibility cannot be  powerful enough to simulate truth-table reducibility.
Nonetheless, the next theorem characterizes $\cfl_{btt}^{\cfl}$.

\begin{theorem}\label{manyone-btt-equiv}
$\cfl_{m}^{\bhcfl} = \cfl_{btt}^{\cfl} = \nfa_{btt}^{\cfl}$.
\end{theorem}

Before proving this theorem, we  will present a new characterization of $\bhcfl$ in terms of btt $\cfl$-reductions. For this purpose, we need to introduce a $btt$-relativization of $\dfa$. Given a language $A$ over alphabet $\theta$, a language $L$ over alphabet $\Sigma$ is in $\dfa_{ktt}^{A}$ if there exists an oracle dfa $M= (Q,\Sigma,\{\cent,\dollar\},\Theta,\delta,q_0,Q_{acc},Q_{rej})$ with $\delta: (Q-Q_{halt})\times (\check{\Sigma}\cup\{\lambda\})\to Q\times (\Theta\cup\{\lambda,\dollar\})^k$ such that $M$ produces $k$ query words $y_1\dollar,y_2\dollar,\ldots,y_k\dollar$ (where $\dollar$ is the termination symbol) with $y_i\in\Theta^*$ for all $i\in[k]$ on $k$ tracks of a single query tape satisfying the following: for every string $x$, $x\in L^A$ if and only if $B(x,\chi^A_k(y_1,y_2,\ldots,y_k))=1$. Note that, as done before, $M$ must be well-behaved at both $\cent$ and $\dollar$. We set $\dfa_{btt}^{A} = \bigcup_{k\in\nat^{+}}\dfa_{ktt}^{A}$ and $\dfa_{btt}^{\CC} = \bigcup_{A\in\CC} \dfa_{btt}^{A}$ for any language family $\CC$.

\begin{lemma}\label{BHCFL-DFA}
$\bhcfl = \dfa_{btt}^{\cfl}$.
\end{lemma}

\begin{proof}
We split the lemma into two separate claims:
$\bhcfl\subseteq \dfa_{btt}^{\cfl}$ and
$\dfa_{btt}^{\cfl}\subseteq \bhcfl$. Let us prove the first claim.

\begin{claim}\label{CFLk-include-DFA}
Given any index $k\in\nat^{+}$, $\cfl_{k} \subseteq \dfa_{ktt}^{\cfl}$ holds.
\end{claim}

\begin{proof}
Our goal is to prove this claim by induction on $k\in\nat^{+}$.
For the base case $k=1$, it is not difficult to show that $\cfl\subseteq \dfa_{m}^{\cfl}\subseteq \dfa_{1tt}^{\cfl}$. Now, let us concentrate on  induction step $k\geq2$. Meanwhile, we intend to prove only the case where $k$ is even, because the case of odd $k$ can be proven analogously.
Since $\cfl_{k} = \cfl_{k-1}\wedge \co\cfl$, take two languages $L_1\in\cfl_{k-1}$ and $L_2\in\co\cfl$ and assume that $L = L_1\cap L_2$.
Assume also that an npda $M_2$ computes $\overline{L}_2$. Since $L_1\in \dfa_{(k-1)tt}^{\cfl}$ by our induction hypothesis, there is an oracle dfa, say, $M_1$ that  recognizes $L_1$ relative to a certain oracle $A$ in $\cfl$ with truth-table $B\in\reg$. Assume that
$M_1$ on input $x$ produces $(k-1)$-tuple $(y_1\dollar,y_2\dollar,\ldots,y_{k-1}\dollar)$ on its  query tape before entering any accepting state and that
$M_1$ satisfies $B(x,\chi^{A}_{k-1}(y_1,y_2,\ldots,y_{k-1}))=1$
if and only if $x$ is in $L_1$.

In what follows, we want to define a new machine $N$.  By simulating $M_1$, $N$ generates $k-1$ query words $(y_1,y_2,\ldots,y_{k-1})$ with no termination symbols, as well as a new query word $\natural x$, where $\natural$ is a fresh symbol.  Moreover, we define $A' = A\cup \{\natural x\mid \,\text{$M_2$ accepts $x$}\,\}$, which is  obviously in $\cfl$.
Now, we  define $B'$  as $\{(x,b_1b_2\cdots b_{k-1} b_k) \mid  (x,b_1b_2\cdots b_{k-1})\in B \wedge b_k=0\}$.  Clearly, $B'$ is regular since so is $B$.  It also follows that   $B'(x,\chi_{k}^{A'}(y_1,y_2,\ldots,y_{k-1},\natural x))=1$ if and only if  $B(x,\chi_{k-1}^{A}(y_1,y_2,\ldots,y_{k-1}))=1$ and $x\in \overline{L_2}$.  Therefore, $L$ belongs to $\cfl_{ktt}^{A'}\subseteq \cfl_{ktt}^{\cfl}$.
\end{proof}

{}From Claim \ref{CFLk-include-DFA}, it follows that $\bhcfl = \bigcup_{k\in\nat^{+}}\cfl_{k}\subseteq \bigcup_{k\in\nat^{+}}\dfa_{ktt}^{\cfl} = \dfa_{btt}^{\cfl}$.
Next, we wish to prove the following.

\begin{claim}\label{FDAktt-include-CFL}
For any index $k\in\nat^{+}$,  $\dfa_{ktt}^{\cfl}\subseteq \cfl_{k2^{k+1}}$.
\end{claim}

\begin{proof}
We prove the claim by induction on $k\in\nat^{+}$. We begin with the base case of $k=1$. Let $L$ be any language in $\dfa_{1tt}^{A}$ for a certain language $A$ in $\cfl$.  Take an oracle dfa $M_1 = (Q_1,\Sigma,\{\cent,\dollar\},\Theta,\delta_1,q_0,Q_{1,acc},Q_{1,rej})$ that recognizes $L$ relative to $A$ and take a truth-table $B$ used for $M_1$ and $A$.
Moreover, let $M_2 = (Q_2,\Theta,\{\cent,\dollar\},\Gamma_2, \delta_2,q_0,Z_0,Q_{2,acc},Q_{2,rej})$ be an npda that recognizes $A$. For the truth-table $B$, denote by $M_3$ a dfa $(Q_3,\Sigma',\{\cent,\dollar\},\delta_3,q_0,Q_{3,acc},Q_{3,rej})$ that recognizes $B$, where $\Sigma'= \{\track{\sigma}{\alpha}\mid \sigma\in\Sigma,\alpha\in\{0,1,\#\}\}$.
Without loss of generality, we assume that both $M_2$ and $M_3$ \emph{make no $\lambda$-move}.  A new machine $N_1$ works in the following way. Given any input $x$, $N_1$ first checks if $B(x,1)=1$ (\ie $\track{x}{1\#^{|x|-1}}\in B$) by simulating $M_3$. Simultaneously, $N_1$ simulates $M_1$ on $x$. This simulation is possible because $M_3$ uses no stack. When $M_1$ tries to write a symbol, say, $b$, on the $i$th query tape, $N_1$ simulates one step of $M_2$'s computation during the scanning of $b$.

A transition function $\delta_{N_1}$ of $N_1$ is formally given as follows.
Assume that $\delta_1(q_0,\cent)= (p_1,b)$, $(p',\tau w)\in\delta_2(q_0,\cent,Z_0)$, and $\delta_3(q_0,\cent) = p_3$. In the case of $b\in\Theta$, if $(p_2,w')\in\delta_2(p',b,\tau)$, then we set  $((p_1,p_2,p_3),w'w)\in \delta_{N_1}((q_0,q_0,q_0),\cent,Z_0)$. If $b=\lambda$, then we define
$((p_1,p',p_3),\tau w)\in \delta_{N_1}((q_0,q_0,q_0),\cent,Z_0)$.
Given $\sigma\in\Sigma$ and $b\in\Theta\cup\{\lambda,\dollar\}$, assume that $\delta_1(q_1,\sigma)=(p_1,b)$, $(p_2,w)\in\delta_2(q_2,b,a)$, and $\delta_3(q_3,\track{\sigma}{\alpha})=p_3$. When $\alpha\in\{0,1\}$, let $\delta_{N_1}((q_1,q_2,q_3),\sigma,a)$ contain $((p_1,p_2,p_3,\alpha),w)$ if $b\neq\lambda$; $((p_1,q_2,p_3,\alpha),a)$ if $b=\lambda$. When $\alpha=\#$, let $\delta_{N_1}((q_1,q_2,q_3,\alpha),\sigma,a)$ contain $((p_1,p_2,p_3,\alpha),w)$ if $b\neq\lambda$; $((p_1,q_2,p_3,\alpha)$ if $b=\lambda$.
If $\delta_1(q_1,\dollar)=(p_1,\lambda)$, $(p_2,w)\in \delta_2(q_2,\dollar,a)$, and $\delta_3(q_3,\dollar)=p_3$, then we set $((p_1,p_2,p_3,\alpha),w)\in \delta_{N_1}((q_1,q_2,q_3),\dollar,a)$.
Let $Q_{N_1,rej}=\{(p_1,p_2,p_3,\alpha)\mid \alpha=0 \text{ or } \exists i\in[3][p_i\in Q_{i,rej}]\}$ and $Q_{N_1,acc}=\{(p_1,p_2,p_3,1)\mid \forall i\in[3][p_i\in Q_{i,acc}]\}$.

In a similar fashion, we define $N_0$, except that (i) it checks if $B(x,0)=1$ and (ii) if $M_2$ enters an accepting state (resp., a rejecting state), then $N_0$ enters a rejecting state (resp., an accepting state). It is important to note that this machine $N_0$ is {\em co-nondeterministic}, and thus   $L(N_0)$ is in $\co\cfl$, whereas $L(N_1)$ belongs to $\cfl$.
The definitions of $N_0$ and $N_1$ yield $L= L(N_0)\cup L(N_1)$.  Hence, $L$ belongs to $\cfl\vee \co\cfl$, which is included in $\cfl_{3}\subseteq \cfl_{4}$, as requested.

For the case $k\geq2$, we need to generalize the above argument. Assume that  $L\in\dfa_{ktt}^{A}$ for a certain language $A$ in $\cfl$. Let $M_1$ be a $ktt$-reduction dfa that reduces $L$ to $A$ and let $M_2$ be an npda recognizing $A$.  As before, let us assume that \emph{$M_2$ makes no $\lambda$-move}. In the following argument, we fix a string $b=b_1b_2\cdots b_k\in\{0,1\}^k$. Letting $\cfl^{(0)}_{2}=\co\cfl_{2}$ and $\cfl^{(1)}_{2}=\cfl_{2}$, we define $\cfl^{(b_1b_2\cdots b_k)}_{2}$ as an abbreviation of $\cfl^{(b_1)}_{2}\vee \cfl^{(b_2)}_{2}\vee \cdots \vee\cfl^{(b_k)}_{2}$.

Next, we will introduce two types of machines $N_{b,0}$ and $N_{b,1}$. The machine $N_{b,1}$  takes input $x$ and checks if $B(x,b)=1$. At the same time, $N_{b,1}$  guesses a number $i\in[k]$ and simulates $M_1$ on $x$ if $b_i=1$ (and, otherwise, it enters an accepting state instantly).  Whenever $M_1$ tries to write a symbol, say, $\sigma$ on its own query tape, $N_{b,1}$ simulates one step of $M_2$'s computation corresponding to the scanning of $\sigma$.  As for the other machine $N_{b,0}$, it simulates $M_1$ on $x$ if $b_i=0$ (and accepts instantly otherwise).
Note that $N_{b,0}$ is a co-nondeterministic machine.
For  each index $e\in\{0,1\}$, let $A_{b,e}$ be composed of strings accepted by $N_{b,e}$ and
define $A_{b} = A_{b,0}\cup A_{b,1}$, which obviously belongs to $\cfl_{2}$. It is not difficult to show that $L = \bigcup_{b\in\{0,1\}^k} A_{b}$; thus, $L$ is in $\bigvee_{b\in\{0,1\}^k}\cfl^{(b)}_{2}$, which equals $\bigvee_{b\in\{0,1\}^k}[ (\bigvee_{i:b_i=1}\cfl_{2}) \vee (\bigvee_{i:b_i=0}\co\cfl_{2})]$.
By a simple calculation, the last term coincides with $(\bigvee_{k=1}^{k2^{k-1}}\cfl_{2})\vee (\bigvee_{k=1}^{k2^{k-1}}\co\cfl_{2})$. Since $\co\cfl_{2}=\cfl\vee\co\cfl\subseteq \cfl\vee \cfl_{2}$, it follows that $\bigvee_{i=1}^{k2^{k-1}}\co\cfl_{2}\subseteq \bigvee_{i=1}^{k2^{k-1}}(\cfl\vee \cfl_{2}) = (\bigvee_{i=1}^{k2^{k-1}}\cfl) \vee (\bigvee_{i=1}^{k2^{k-1}}\cfl_{2}) = \bigvee_{i=1}^{k2^{k-1}}\cfl_{2}$. Therefore, $L$ belongs to $\bigvee_{i=1}^{k2^{k}}\cfl_{2}$. Note that $\cfl_{k2^{k+1}} = \bigvee_{k=1}^{k2^{k}}\cfl_{2}$ by Claim \ref{CFL-decomposition} in Section \ref{sec:Turing-reduction}. We thus conclude that  $L$ is in $\cfl_{k2^{k+1}}$.
\end{proof}

Claim \ref{FDAktt-include-CFL} implies that $\dfa_{btt}^{\cfl} = \bigcup_{k\geq1}\dfa_{ktt}^{\cfl} \subseteq \bigcup_{k\geq1}\cfl_{k2^{k+1}} \subseteq \bhcfl$, as requested.
\end{proof}

Finally, we will present the proof of Theorem \ref{manyone-btt-equiv}.
For this proof, we need to introduce the third simulation technique of
encoding a computation path of an npda into a string. Notice that a series of nondeterministic choices made by an npda $M$ uniquely specifies which computation path the npda $M$ has followed. We encode such a series into a unique string. Let $\delta$ be a transition function of $M$.  First, we rewrite $\delta$ in the following manner.
If $\delta$ has an entry of the form $\delta(q,\sigma,\tau) =\{(p_1,\xi_1,\zeta_1),(p_2,\xi_2,\zeta_2)\}$, then we split it into two distinguished transitions: $\delta(q,\sigma,\tau)=(p_1,\xi_1,\zeta_1)$ and $\delta(q,\sigma,\tau)=(p_2,\xi_2,\zeta_2)$. Let $D$ indicate a  collection of all such new transitions. Next, we index all such new transitions using numbers in the integer interval $[\parallel\! D \!\parallel] = \{1,2,\ldots,\parallel\! D \!\parallel\}$.
For simplicity, the notation $\ceilings{\delta(q,\sigma,\tau) = (p,\xi,\zeta)}$ denotes the number assigned to the corresponding transition. A series of transitions can be expressed as a series of those indices, which is regarded as a string over the alphabet $\Sigma= [ \parallel\! D \!\parallel ]$. We call such a string an {\em encoding of a computation path} of $M$.

\vs{-2}
\begin{proofof}{Theorem \ref{manyone-btt-equiv}}
In this proof, we will prove three inclusions: $\cfl_{btt}^{\cfl} \subseteq \nfa_{btt}^{\cfl}  \subseteq  \cfl_{m}(\dfa_{btt}^{\cfl}) \subseteq \cfl_{btt}^{\cfl}$.  Obviously, these inclusions together ensure the desired equations in the theorem. We begin with the first inclusion relation.

\begin{claim}\label{first-assertion}
$\cfl_{btt}^{\cfl} \subseteq \nfa_{btt}^{\cfl}$.
\end{claim}

\begin{proof}
Fix $k\in\nat^{+}$ and let $L$ be any language in $\cfl_{ktt}^{A}$ for a certain oracle $A\in\cfl$.  For this $L$, there is a $ktt$-reduction npda, say, $M_1$ that recognizes $L$ relative to $A$.
We want to define a new oracle nfa $N$ having $k+1$ query tapes
by modifying $M_1$ in a way similar to the proof of Claim \ref{NFA-to-DYCK}. On input $x$, $N$ simulates $M_1$ on $x$ using $k+1$ query tapes as follows.  
We use the first $k$ query tapes of $N$ to produce a tuple of query words made by $M$. When $M_1$ tries to push down $w$ in place of a certain symbol on the top of a stack, $N$ guesses this symbol, say, $\sigma$ and then writes down $\sigma' w^R$ on the $k+1$st query tape. If $M_1$ pops up a certain symbol,   then $N$ first guesses this symbol, say, $\sigma$ and writes down $\sigma'$ on the $k+1$st query tape. Finally, we define $B'$ as the set $\{[x,y_1,y_2,\cdots, y_k,1]^{T}\mid B(x,y_1y_2\cdots y_k)=1\}$.   Associated with $N$'s $k+1$st query, we choose an appropriate language $C$ in $DYCK$.  Define $A'=A\cup C$, assuming that $C$ is based on a different alphabet.
Note that $C\in\cfl$. As in the proof of Claim \ref{NFA-to-DYCK}, it is possible to prove that, for any $x$,
$x$ is in $L$ if and only if $B'(x,\chi^{A'}_{k+1}(y_1,y_2,\ldots,y_k,b))=1$ for a certain valid outcome $(y_1,y_2,\ldots,y_k,b)$ of $N$ on $x$. Hence, $L$ belongs to $\nfa_{(k+1)tt}^{A'}$, which is a subclass of $\nfa_{btt}^{\cfl}$.
\end{proof}

\begin{claim}
$\nfa_{btt}^{\cfl} \subseteq \cfl_{m}(\dfa_{btt}^{\cfl})$.
\end{claim}

\begin{proof}
Let $k\geq1$. Assume that  $L\in\nfa_{ktt}^{A}$ with a certain oracle $A\in\cfl$. Let $M = (Q,\Sigma,\{\cent,\dollar\},\Theta,\Gamma,\delta,q_0,Z_0,Q_{acc},Q_{rej})$ denote a $k$tt-reduction npda that reduces $L$ to $A$.  To show that $L\in \cfl_{m}(\dfa_{btt}^{A})$, we first define an oracle npda $N_1 =(Q_1,\Sigma,\{\cent,\dollar\},\Theta,\Gamma_1, \delta_1,q_0,Z_0,Q_{acc},Q_{rej})$ as follows. Given any input $x$, $N_1$ simulates $M$ by guessing a string $y$ that may encodes a computation path of $M$.
At any time when $M$ tries to write any symbol on its own query tapes, $N_1$ simply ignores this symbol and continues the simulation because $y$ already embodies the information on this symbol.  The machine $N_1$ eventually produces $\track{\tilde{x}}{y}$ on its query tape and then enters the same halting state as $M$'s, where $\tilde{x}$ is an appropriate  $\natural$-extension of $x$.

To be more precise, $\delta_{1}$ satisfies the following. Let $(p,Z_0,\track{\natural}{d})\in\delta_1(q_0,\cent,Z_0)$ if $d = \ceilings{\delta_M(q_0,\cent) = (p,\tau_1,\ldots,\tau_k)}$.
For any $q\in Q-Q_{halt}$ and $\sigma\in \Sigma$, let $(p,a,\track{\sigma'}{d})\in \delta_{1}(q,\sigma,a)$  if $d$ equals $\ceilings{\delta_{M}(q,\sigma)=(p,\tau_1,\tau_2,\ldots,\tau_k)}$ for certain $\tau_1,\tau_2,\ldots,\tau_k\in \Theta\cup\{\lambda\}$, where $\sigma'=\natural$ if $\sigma=\lambda$, and $\sigma'=\sigma$ otherwise.
Finally, let $(p_{\dollar},a,\track{\natural}{d})\in \delta_1(q,\lambda,a)$ and $(p,a,\lambda)\in \delta_1(p_{\dollar},\dollar,a)$  if $d=\ceilings{\delta_M(q,\dollar)=(p,\lambda,\lambda,\ldots,\lambda)}$. All other transitions enter rejecting states.

Next, we define an oracle dfa $N_2$ having $k$ query tapes. On input of the form $\track{\tilde{x}}{y}$, $N_2$ deterministically simulates $M$ on $x$ by following a series of nondeterministic choices specified by $y$, and $N_2$ produces $k$ query words as $M$ does. In the case where $y$ is not any valid accepting computation path of $M$, $N_2$ rejects $\track{\tilde{x}}{y}$ in order to invalidate the produced query words.

Formally, let $N_2= (Q_2,\Theta,\{\cent,\dollar\},\Theta_2,
\delta_2,\overline{q}_0,Q_{acc},Q_{2,rej})$ with $Q_{2,rej}\supseteq Q_{rej}$, whose transition function is defined as follows.
Let $\delta_2(\overline{q}_0,\cent) = (\overline{q}_0,\lambda,\ldots,\lambda)$ and $\delta_2(\overline{q}_0,\track{\natural}{d}) = (p,\tau_1,\ldots,\tau_k)$ if $d$ equals $\ceilings{\delta_M(q_0,\cent) = (p,\tau_1,\ldots,\tau_k)}$.
For $q\in Q_2-Q_{2,halt}\cup\{\overline{q}_0\}$, let   $\delta_{2}(q,\track{\natural}{d}) = (p,\tau_1,\ldots,\tau_k)$ if $d$ equals $\ceilings{\delta_{M}(q,\lambda)=(p,\tau_1,\ldots,\tau_k)}$.
For any $\sigma\neq\natural$, let $\delta_{2}(q,\track{\sigma}{d}) = (p,\tau_1,\ldots,\tau_k)$ if $d$ is $\ceilings{\delta_{M}(q,\sigma)=(p,\tau_1,\ldots,\tau_k)}$ and $p\notin Q_{acc}$; by contrast, when $p\in Q_{acc}$, let $\delta_2(q,\track{\sigma}{d}) = (p_{\dollar},\dollar)$ and $\delta_2(p_{\dollar},\dollar) = (p,\lambda)$.
For other pairs $(q,\track{\sigma}{d})$, $\delta_{2}$ maps them to  appropriate rejecting states.

It is not difficult to show
that $N_1$ $m$-reduces $L$ to $L(N_2,A)$ and that $L(N_2,A)$ is in $\dfa_{btt}^{A}$.  Therefore, $L$ belongs to $\cfl_{m}^{L(N_2,A)} \subseteq \cfl_{m}(\dfa_{btt}^{\cfl})$.
\end{proof}

\begin{claim}\label{third-assertion}
$\cfl_{m}(\dfa_{btt}^{\cfl}) \subseteq \cfl_{btt}^{\cfl}$.
\end{claim}

\begin{proof}
To prove this claim, take any  oracle $A\in\dfa_{btt}^{\cfl}$ and assume that $L\in\cfl_{m}^{A}$ via an $m$-reduction npda $M_1$. Fixing $k\in\nat^{+}$, let  $M_2$ be a $k$tt-reduction dfa that reduces $A$ to a certain language $B$ in $\cfl$. In what follows, we want to define an oracle npda $N$ having $k$ query tapes.
Our construction of $N$ is quite similar to the one in the proof of Claim \ref{CFL(NFA)-equal-CFL}. On input $x$, $N$ simulates $M_1$ on $x$. When $M_1$ writes a symbol, say, $b$, $N$ simulates one or more steps (including a certain number of $\lambda$-moves) of $M_2$'s computation after its tape head scans $b$. Finally, $N$ outputs $M_2$'s $k$ query strings.


The above definition shows that $N$ $k$tt-reduces $L$ to $B$, and thus $L$ is in $\cfl_{ktt}^{B}\subseteq \cfl_{btt}^{\cfl}$.
\end{proof}

Combining Claims \ref{first-assertion}--\ref{third-assertion} proves that $\cfl_{btt}^{\cfl} = \nfa_{btt}^{\cfl} = \cfl_{m}(\dfa_{btt}^{\cfl})$. This  last term is further equal to $\cfl_{m}^{\bhcfl}$ by Lemma \ref{BHCFL-DFA}.  This completes the proof.
\end{proofof}

\subsection{Languages That are Low for CFL}\label{sec:low-set}

We will briefly discuss oracles that contain little information to help underlying oracle machines improve their recognition power. In Section \ref{sec:many-one-reduction}, we have introduced a many-one relativized family $\cfl_{m}^{A}$ relative to oracle $A$.
For a more general treatment, we  consider any language family $\CC$ whose  many-one relativization $\CC_{m}^{A}$ is properly defined. To specify such a family $\CC$, we succinctly call $\CC$ {\em many-one relativizable}. Analogously, we define the notion of \emph{Turing relativizability} and \emph{truth-table relativizability}. For example, $\cfl$ is many-one, truth-table, and Turing  relativizable.

We first assert that regular languages, when playing as oracles, have no power to increase the computational complexity of the relativizable family  $\cfl$.

\begin{lemma}\label{oracle-REG}
$\cfl = \cfl_{m}^{\reg} = \cfl_{btt}^{\reg} = \cfl_{T}^{\reg}$.
\end{lemma}

\begin{proof}
Since $\cfl\subseteq \cfl_{m}^{\Sigma^*}$ and $\Sigma^*\in\reg$ for $\Sigma=\{0,1\}$, it follows that $\cfl\subseteq \cfl_{m}^{\reg}$.  Moreover, by Lemmas \ref{basic-Turing} and \ref{m-vs-ktt-compare},  it holds that $\cfl_{m}^{\reg}\subseteq \cfl_{btt}^{\reg}$ and $\cfl_{m}^{\reg}\subseteq \cfl_{T}^{\reg}$.
To show that $\cfl_{btt}^{\reg} \subseteq\cfl$,  take any language $L$ in $\cfl_{ktt}^{A}$ for a certain index $k\in\nat^{+}$ and a certain language $A\in\reg$. Let $M$ be an oracle npda equipped with two write-only query tapes that $ktt$-reduces $L$ to $A$.  In addition, let $N$ denote a dfa recognizing $A$.
We aim at proving that $L\in\cfl$.   Let us consider the following algorithm.  We  start simulating  $M$ on each input without using any write-only tapes.  When $M$  tries to write down a $k$-tuple of symbols $(s_1,s_2,\ldots,s_k)$, we instead simulate $N$ using only inner states.  Note that we do not need to keep on the query tapes any information on $(s_1,s_2,\ldots,s_k)$. Along each computation path, we accept the input if both $M$ and $N$ enter accepting states.  Since the above algorithm requires no query tapes, it can be implemented by a certain npda. Moreover,  since the algorithm correctly recognizes $L$, we conclude that $L$ is in $\cfl$.

Based on a similar idea, it is possible to prove that $\cfl_{T}^{\reg}\subseteq \cfl$.
\end{proof}

Assuming that $\CC$ is many-one relativizable, a language $A$ is called {\em many-one low} for  $\CC$ if $\CC_{m}^{A}\subseteq \CC$ holds. We define $\mlow\CC$ to be the set of all languages that are low for $\CC$; that is, $\mlow\CC = \{A\mid \CC_{m}^{A}\subseteq \CC\}$.  Similarly, we define $\bttlow\CC$  and $\Tlow\CC$ as the collections of all languages that are ``btt low for $\CC$'' and ``Turing low for $\CC$,'' respectively, provided that $\CC$ is Turing and truth-table relativizable.

\begin{lemma}\label{lowness-result}
\begin{enumerate}
  \setlength{\topsep}{-2mm}%
  \setlength{\itemsep}{1mm}
  \setlength{\parskip}{0cm}

\item $\reg\subseteq \Tlow\cfl \cap \bttlow\cfl \subseteq \Tlow\cfl \cup \bttlow\cfl  \subseteq \mlow\cfl \subsetneq \cfl$.
\item $\bttlow\cfl \subsetneq \cfl\cap\co\cfl$ and $\Tlow\cfl \subsetneq \cfl\cap\co\cfl$.
\end{enumerate}
\end{lemma}

\begin{proof}
(1)  {}From Lemma \ref{oracle-REG}, it holds that  $\cfl_{T}^{\reg} = \cfl_{btt}^{\reg} =\cfl$.  It thus follows that $\reg\subseteq \Tlow\cfl \cap \bttlow\cfl$. The third inclusion comes from the fact that $\cfl_{m}^{A}\subseteq \cfl_{btt}^{A} \cap\cfl_{T}^{A}$ for any oracle $A$. The last inclusion is shown as follows. Take any language $A$ in $\mlow\cfl$. This means that $\cfl_{m}^{A}\subseteq \cfl$. Since $A$ belongs to $\cfl_{m}^{A}$,we conclude that  $A\in\cfl_{m}^{A}\subseteq \cfl$.
Next, let us show that $\cfl\neq \mlow\cfl$.
If $\cfl=\mlow\cfl$ holds, then we derive
$\cfl_{m}^{\cfl} \subseteq \cfl$. Since $\cfl(2)\subseteq\cfl_{m}^{\cfl}$ by Claim \ref{CFL(2)-bound}, we immediately conclude that $\cfl(2)=\cfl$. This is indeed a contradiction against the well-known result that $\cfl(2)\neq \cfl$. Therefore, $\mlow\cfl\neq \cfl$ must hold.

(2) Consider the case of $\bttlow\cfl$. For the containment between $\bttlow\cfl$ and $\cfl\cap\co\cfl$, let us consider any language $A$ in $\bttlow\cfl$; namely, $\cfl_{btt}^{A} \subseteq \cfl$. Since $A,\overline{A}\in\cfl_{btt}^{A}$, we obtain $A,\overline{A}\in\cfl$. Thus, $A$ must belong to $\cfl\cap\co\cfl$.

For the separation between $\bttlow{\cfl}$ and $\cfl\cap\co\cfl$,
we contrarily assume that
$\bttlow\cfl =\cfl\cap\co\cfl$. Thus, $\cfl_{btt}^{\cfl\cap\co\cfl}=\cfl$. This implies $\cfl_{m}^{\dcfl}\subseteq \cfl$, because $\dcfl\subseteq \cfl\cap\co\cfl$ and $\cfl_{m}^B\subseteq \cfl_{btt}^B$ for any oracle $B$. However, this is a contradiction because $\cfl_{m}^{\dcfl}\nsubseteq \cfl/n$ by Proposition \ref{CFL_m^CFL-separation} and Lemma \ref{non-closure-many-one}.

The case of $\Tlow\cfl$ can be similarly handled.
\end{proof}

At this moment, it is not clear whether all inclusion relations in Lemma \ref{lowness-result}(1), except for the last one, are \emph{proper inclusions} although they are expected to be proper.

\section{The CFL Hierarchy}\label{sec:CFL-hierarchy}

Nondeterministic polynomial-time Turing reductions have been used to build the polynomial  hierarchy, each level of which is generated from its lower level by applying such reductions. With use of  our Turing $\cfl$-reducibility defined in Section \ref{sec:tt-reduction} instead, a similar construction can be applied to $\cfl$, introducing a unique hierarchy, which we fondly call the \emph{CFL hierarchy}.
Throughout this section, we intend to explore fundamental properties of this new intriguing hierarchy.

\subsection{Turing CFL-Reducibility and a Hierarchy over CFL}\label{sec:Turing-reduction}

In Section \ref{sec:tt-reduction}, we have seen the usefulness of Turing CFL-reducibility. We apply Turing $\cfl$-reductions to $\cfl$, level by level, and we build a meaningful hierarchy, called succinctly the {\em CFL hierarchy}, whose $k$th level consists of three language families denoted by $\deltacflt{k}$, $\sigmacflt{k}$, and $\picflt{k}$. To be more precise,
for each level $k\geq1$, we set $\deltacflt{1}=\dcfl$, $\sigmacfl{1} = \cfl$, $\picflt{k}=\co\sigmacflt{k}$, $\deltacflt{k+1}=\dcfl_{T}(\sigmacflt{k})$, and $\sigmacflt{k+1}=\cfl_{T}(\sigmacflt{k})$. Collectively, we set $\cflh = \bigcup_{k\in\nat^{+}}\sigmacflt{k}$. As done for the polynomial hierarchy, the term ``CFL hierarchy'' in this paper refers to not only the collection $\{\deltacfl{k},\sigmacfl{k},\picfl{k}\mid k\in\nat^{+}\}$ but also the language family $\cflh$.

The $\cfl$ hierarchy can be used to categorize the complexity of typical non-context-free languages discussed in most introductory textbooks, \eg \cite{HMU01,Lin06}.  We will review such languages that naturally fall into the $\cfl$ hierarchy.

\begin{example}\label{ex:Sq-Prim}
We have seen in Example \ref{ex:DUP_2} the languages $Dup_2=\{xx\mid x\in\{0,1\}^*\}$ and $Dup_3=\{xxx\mid x\in\{0,1\}\}$, which are both in $\cfl_{m}^{\cfl}$.
Note that, since $\cfl_{m}^{A}\subseteq \cfl_{T}^{A}$ for any oracle $A$ by Lemma \ref{basic-Turing}, every language in $\cfl_{m}^{\cfl}$ belongs to $\cfl_{T}^{\cfl} = \sigmacflt{2}$. Therefore, $Dup_2$ and $Dup_3$ are members of $\sigmacflt{2}$. In addition, as shown in Example \ref{ex:CFL(k)},
the language $Sq=\{0^n1^{n^2}\mid n\geq1\}$ is in $\cfl_{m}^{\cfl}$ while $Prim=\{0^n\mid \text{ $n$ is a prime number }\}$ is in $\co(\cfl_{m}^{\cfl})$.
Since $\cfl_{m}^{\cfl}\subseteq\cfl_{T}^{\cfl}\subseteq \sigmacfl{2}$, we conclude that $Sq$ is in $\sigmacflt{2}$ and $Prim$ is in $\picflt{2}$.
A similar but more involved example is the language $MulPrim=\{0^{mn}\mid \,\text{$m$ and $n$ are prime numbers}\,\}$.
Consider the following three npda's.
The first machine $M_1$ guesses $n\in\nat^{+}$ and  nondeterministically partitions a given input $0^k$  into $(y_1,y_2,\ldots,y_n)$ and produces $w=y_1\natural y_2\natural \cdots \natural y_n$ on a query tape by inserting a new symbol $\natural$. During this process, $M_1$ pushes $u=y_1\#1^n$ to a stack, where $\#$ is a fresh symbol. In the end, $M_1$ appends $\# u$ to $w$ on the query tape. In receiving $w\# u$ as an input, the second machine $M_2$ checks whether $y_{2i-1}=y_{2i}$ for each $i\in[1,\floors{n/2}]_{\integer}$.  At any moment when the checking process fails, $M_2$ enters an appropriate  rejecting state and halts. Next, $M_2$ nondeterministically partitions $y_1$ into $(z_1,z_2,\ldots,z_e)$ and $1^n$ into $(x_1,x_2,\ldots,x_d)$ and then it produces $w\# u'\# v'$ on its query tape, where $u' = z_1\natural z_2\natural \cdots \natural z_e$ and $v'=x_1\natural x_2\natural \cdots \natural x_d$.  In receiving  $w\#u'\# v'$, the third machine $M_3$ checks if $y_{2i}=y_{2i+1}$ for all $i$ with $1\leq i\leq \floors{(n-1)/2}$.
Whenever  this process fails, $M_3$ instantly halts in a ceratin rejecting state. Next, $M_3$ nondeterministically chooses a bit $b$. If $b=0$, then $M_3$ checks if $z_{2i-1}=z_{2i}$ and also $x_{2i-1}=x_{2i}$ for $i\in[1,\floors{n/2}]_{\integer}$; on the contrary, if $b=1$, then $M_3$ checks that both $z_{2i}=z_{2i+1}$ and $x_{2i}=x_{2i+1}$ for any $i$ satisfying $1\leq i\leq \floors{(n-1)/2}$. If this checking process is successful, then $M_3$ enters an appropriate rejecting state; otherwise, it enters an accepting state. By combining those three machines,  $MulPrim$ can be shown to belong to $\cfl_{m}(\co(\cfl_{m}^{\co\cfl}))$, which is contained  in $\sigmacflt{3}$.
\end{example}


Several basic relationships among the components of the $\cfl$ hierarchy are exhibited in the next lemma. More structural properties will be discussed later in Section \ref{sec:structure-hierarchy}.

\begin{lemma}
Let $k$ be any integer satisfying $k\geq1$.
\begin{enumerate}\vs{-1}
  \setlength{\topsep}{-2mm}%
  \setlength{\itemsep}{1mm}
  \setlength{\parskip}{0cm}

\item $\cfl_{T}(\sigmacflt{k}) = \cfl_{T}(\picflt{k})$ and $\dcfl_{T}(\sigmacflt{k}) = \dcfl_{T}(\picflt{k})$.
\item $\sigmacflt{k}\cup \picflt{k}\subseteq \deltacflt{k+1}\subseteq \sigmacflt{k+1}\cap \picflt{k+1}$.
\end{enumerate}
\end{lemma}

\begin{proof}
(1) The first equality is a direct consequence of Lemma \ref{basic-Turing} since $\picflt{k}=\co\sigmacflt{k}$. The case of Turing $\dcfl$-reduction is similar in essence.

(2)  Let $k\geq1$.  Since $A\in \dcfl_{T}^{A}$ holds for any oracle $A$ by making a simple query on $x$,  it thus follows that $\sigmacflt{k}\subseteq \dcfl_{T}(\sigmacflt{k}) = \deltacflt{k+1}$. Similarly, we obtain $\picflt{k}\subseteq \dcfl_{T}(\picflt{k}) = \dcfl_{T}(\sigmacflt{k})=\deltacflt{k+1}$, where the  first  equality comes from (1).  Moreover, since $\dcfl_{T}^A\subseteq \cfl_{T}^A$ for all oracles $A$, we conclude that $\deltacflt{k+1} = \dcfl_{T}(\sigmacflt{k}) \subseteq \cfl_{T}(\sigmacflt{k}) = \sigmacflt{k+1}$. Finally, using the fact that $\dcfl_{T}^A=\co\dcfl_{T}^A$ for any oracle $A$, we easily obtain $\deltacflt{k+1} = \co\deltacflt{k+1} \subseteq \co\sigmacflt{k+1}=\picflt{k+1}$.
\end{proof}

As is shown in Example \ref{ex:Sq-Prim}, $Dup_{2}$ is in $\sigmacfl{2}$. By contrast, it is well-known that $Dup_{2}$ is not context-free; thus, $Dup_{2}\notin\sigmacfl{1}$. This fact yields the following class separation.

\begin{proposition}\label{first-second-gap}
$\sigmacfl{1}\neq\sigmacfl{2}$.
\end{proposition}


Hereafter, we will explore fundamental properties of our new hierarchy.  Our starting point is a closure property under length-nondecreasing substitution. A {\em substitution} on alphabet $\Sigma$ is actually a function $s:\Sigma\rightarrow\PP(\Theta^*)$ for a certain alphabet $\Theta$. This substitution $s$ is called \emph{length nondecreasing} if $s(\sigma)\neq\setempty$ and $\lambda\notin s(\sigma)$ for every symbol $\sigma\in\Sigma$.  We further extend this function from its finite domain $\Sigma$ to the infinite domain $\Sigma^*$. Given any string $y=\sigma_1\sigma_2\cdots \sigma_n$, where each $\sigma_i$ is a symbol in $\Sigma$, we set $s(y)$ to be the language $\{x_1x_2\cdots x_n \mid i\in[n], x_i\in s(\sigma_i)\}$. We conveniently set $s(\lambda)=\lambda$.
Moreover, for any language $L\subseteq \Sigma^*$, we define $s(L) = \bigcup_{y\in L} s(y)$.
Each language family $\sigmacflt{k}$ is closed under length-nondecreasing substitution in the following sense.

\begin{lemma}\label{substitution-close}
(length-nondecreasing substitution property) Let $k\in\nat^{+}$ and let $s$ be any length-nondecreasing substitution on alphabet $\Sigma$ satisfying $s(\sigma)\in\sigmacflt{k}$ for each symbol $\sigma\in\Sigma$. For any context-free language $L$ over $\Sigma$, $s(L)$ belongs to  $\sigmacflt{k}$.
\end{lemma}

\begin{proof}
Since the basis case $k=1$ is well-known to hold (see, \eg \cite{HMU01}),
it suffices to assume that $k\geq2$.

A major deviation from a standard proof for CFL's closure property under substitution is the presence of a query tape and a process of both querying a word and receiving its oracle answer.

Let $s:\Sigma\to\PP(\Theta^*)$ be a given length-nondecreasing substitution and let $L$ be any given context-free language over alphabet $\Sigma$. To simplify the proof, we consider only the case where $\lambda\notin L$.
For each symbol $\sigma$ in $\Sigma$, take an oracle npda $M_{\sigma} = (Q_{\sigma},\Theta,\{\cent,\dollar\},\Theta_{\sigma},\Gamma_{\sigma}, \delta_{\sigma},q_0,Z_0,Q_{\sigma,acc},Q_{\sigma,rej})$ that recognizes $s(\sigma)$ relative to a certain language $A_{\sigma}$ in $\picflt{k-1}$.
For convenience, we assume that all $A_{\sigma}$'s have different alphabets and that, when $M_{\sigma}$ halts in an accepting state, its query tape must be all blank.
In addition, let $M = (Q_{M},\Sigma,\{\cent,\dollar\},\Gamma_{M}, \delta_{M},q_0,Z_0,Q_{M,acc},Q_{M,rej})$ denote an npda recognizing $L$. To help simplify our argument, we  assume that (1) $M$ makes no $\lambda$-move and (2) by the time each oracle npda $M_{\sigma}$ enters any halting state, it must empty its own stack (except for its bottom marker) and also makes its query tape blank by making an extra query and disregarding its oracle answer, if necessary.

Consider a new oracle npda $N = (Q_{N},\Theta,\{\cent,\dollar\},\Theta_N,\Gamma_N, \delta_N,q_0,Z_0,Q_{query},Q_{acc},Q_{rej})$ that behaves in the following manner.
On input $x\in\Theta^*$,  $N$ nondeterministically splits $x$ into $(x_1,x_2,\ldots,x_n)$ for which $x=x_1x_2\cdots x_n$, where $1\leq n\leq|x|$ and $x_i\in\Theta^*$ for each $i\in[n]$. Initially, $N$ simulates exactly one step of $M$ while scanning $\cent$. Sequentially, at stage $i\in[n]$, $N$ guesses a symbol, say, $\sigma_i\in\Sigma$ and simulates, using a stack, one step of $M$ whose tape head scanning $\sigma_i$. In the end of this simulation stage, $N$ places a special separator $\#$ on the top of the stack in order to share the same stack with $M_{\sigma_i}$. Next, $N$ simulates $M_{\sigma_i}$ on the input $\cent x_i \dollar$ using an empty portion of the stack, by regarding $\#$ as a new bottom marker for $M_{\sigma_i}$.
To make intact the saved data in the stack during this simulation of $M_{\sigma_i}$, whenever $M_{\sigma_i}$ tries to remove $\#$, $N$ instantly aborts the simulation and halts in a rejecting state. When $M_{\sigma_i}$ tries to query a string to an oracle, $N$ also produces the same string on its query tape, remembers the last inner state as well as $\sigma_i$ using the stack, and then enters $q_{query}$.
If all npda's $M_{\sigma_i}$ enter certain accepting states, then $N$ accepts $x$; otherwise, it rejects $x$.

Here, we briefly describe the transition function $\delta_{N}$. Remember
that all machines are well-behaved at both $\cent$ and $\dollar$.
As a starter, when $(p,w)\in\delta_{M}(q_0,\cent,Z_0)$, let $((p,\track{\lambda}{\lambda}),w,\lambda) \in\delta_{N}(q_0,\cent,Z_0)$.  Let $((p,\track{\sigma}{q_0\cent}),\# w,\lambda)\in \delta_{N}((q,\track{\lambda}{\lambda}),\lambda,a)$ for all $\sigma\in\Sigma$ if $q\notin Q_{M,halt}$.
Let $((q',\track{\sigma}{p}),w\#,\xi)\in \delta_{N}((q',\track{\sigma}{q_0\cent},\lambda,\#)$ if $(p,wZ_0,\xi)\in\delta_{\sigma}(q_0,\cent,Z_0)$. When $(p,w,\xi)\in\delta_{\sigma}(q,\tau,a)$ for $\tau\in\Theta\cup\{\lambda\}$, let $\delta_{N}((q',\track{\sigma}{q}),\tau,a)$ contain  $((q',\track{\sigma}{p}),w,\xi)$ and  $((q',\track{\sigma}{p\dollar}),w,\xi)$. Let $((q',\track{\sigma}{p\lambda}),w,\xi) \in\delta_{N}((q',\track{\sigma}{q\dollar}),\lambda,a)$ if $(p,w,\xi)\in\delta_{\sigma}(q,\dollar,a)$ with $p\in Q_{\sigma,halt}$.
Let $((q',\track{\sigma}{p\lambda}),w,\xi) \in\delta_{N}((q',\track{\sigma}{q\lambda}),\lambda,a)$ if $(p,w,\xi)\in\delta_{\sigma}(q,\lambda,a)$.
Let $((q',\track{\lambda}{\lambda}),\lambda,\lambda)  \in\delta_{N}((q',\track{\sigma}{q\lambda}),\lambda,\#)$ for all $q\in Q_{\sigma,acc}$.
When $(q_{query},w,\xi)\in\delta_{\sigma}(q,\tau,a)$, let  $(q_{query},\track{q'}{\sigma}w,\xi) \in\delta_{N}((q',\track{\sigma}{q}),\tau,a)$.
When $(p,w,\xi)\in\delta_{\sigma}(q_e,\tau,a)$ for each $e\in\{yes,no\}$,
let  $((q',\track{\sigma}{q_e}),\lambda,\lambda) \in\delta_{N}(q_e,\lambda,\track{q'}{\sigma})$ and  $((q',\track{\sigma}{p}),w,\xi)
\in\delta_{N}((q',\track{\sigma}{q_e},\tau,a)$. For all other cases, $\delta_N$ maps them to appropriate rejecting states. In the end, we set $Q_{acc} = \{(q,\track{\lambda}{\lambda})\mid q\in Q_{M,acc}\}$ and $Q_{rej}\supseteq \{(q,\track{\lambda}{\lambda})\mid q\in Q_{M,rej}\}$.

Finally, an oracle $B$ is defined as the finite union $\bigcup_{\sigma\in \Sigma} A_{\sigma}$. It can be observed that $x$ is in $s(L)$ if and only if
$N$  accepts $x$ relative to $B$.
Notice that our induction hypothesis states that $\sigmacfl{k-1}$ is closed under length-nondecreasing substitution. From this assumption, we can prove that, by following the proof of Lemma \ref{closure-operators}, $\Sigma_{k-1}$ is closed under union. Since all $A_{\sigma}$'s are in $\sigmacflt{k-1}$, the set $B$ must be in $\sigmacflt{k-1}$.  Therefore, $s(L)$ belongs to $\cfl_{m}^{B}\subseteq \sigmacflt{k}$.
\end{proof}

Once the closure property under length-nondecreasing substitution is established for $\sigmacflt{k}$, other well-known closure properties (except for reversal and inverse homomorphism) follow directly. A \emph{homomorphism} is a function $h:\Sigma\to\Theta^*$ for alphabets $\Sigma$ and $\Theta$. Such a homomorphism is called \emph{$\lambda$-free} if $h(\sigma)\neq\lambda$ for every $\sigma\in\Sigma$.

\begin{lemma}\label{closure-operators}
For each index $k\in\nat^{+}$, the family $\sigmacflt{k}$ is closed under the following operations: concatenation, union, reversal, Kleene closure, $\lambda$-free homomorphism, and inverse homomorphism.
\end{lemma}

\begin{proof}
When $k=1$, $\sigmacflt{1}$ ($=\cfl$) satisfies all the listed closure properties (see, \eg \cite{HMU01,Lin06}).  Hereafter, we assume that $k\geq2$. All  the closure properties except for reversal and inverse homomorphism follow directly from Lemma  \ref{substitution-close}. For completeness, however, we will include the proofs of those closure properties. The remaining closure properties require different arguments.

\s

\n{\bf [union]} Given two languages $A_1$ and $A_2$ in $\sigmacflt{k}$, if at least one of $A_1$ and $A_2$ is empty, the union $A_1\cup A_2$ obviously belongs to $\sigmacfl{k}$. Henceforth, we assume that both $A_1$ and $A_2$ are nonempty. Let us define $L=\{1,2\}$ and take a length-nondecreasing substitution $s$ satisfying $s(i)=A_i$ for each $i\in\{1,2\}$. Since $L\in\reg$ and  $s(L)=A_1\cup A_2$, Lemma \ref{substitution-close} implies  that $s(L)$  belongs to $\sigmacflt{k}$.

\s

\n{\bf [concatenation]} Take any two languages $A_1,A_2\in\sigmacflt{k}$. It suffices to consider the case where they are both nonempty (as in the proof  for ``union''). In this case, we set $L=\{12\}$ and define $s(1)=A_1$ and $s(2)=A_2$. Since $L\in\reg$ and $s(L)=\{xy\mid x\in A_1,y\in A_2\}$, we apply Lemma \ref{substitution-close} to $s(L)$ and obtain the desired containment $s(L)\in\sigmacflt{k}$.

\s

\n{\bf [Kleene closure]}  Given any nonempty language $A$  in $\sigmacflt{k}$, we define $s(1)=A$ and $L=\{1\}^*$, which obviously imply $s(L)=A^*$. Note that $L$ is a regular language. Next, we apply Lemma \ref{substitution-close} and then obtain the desired membership $s(L)\in\sigmacflt{k}$.

\s

\n{\bf [$\lambda$-free homomorphism]} This is trivial since a $\lambda$-free homomorphism is a special case of a length-nondecreasing substitution.

\s

\n{\bf [inverse homomorphism]}
Let $\Sigma$ and $\Gamma$ be two alphabets and take any language $A$ in $\sigmacflt{k}$ over $\Gamma$ and any homomorphism $h$ from $\Sigma$ to $\Gamma^*$. Our goal is to show that $h^{-1}(A)$ is in $\sigmacflt{k}$. Let $M$ be an oracle npda that recognizes $A$ relative to an oracle, say,  $B$ in $\sigmacflt{k-1}$.
Notice that  our oracle npda must halt in linear time on all computation paths for any choice of oracles. Let us construct another oracle npda $N$ for $h^{-1}(A)$. Given any input $x=x_1x_2\cdots x_n$ of length $n$, $N$ applies $h$ symbol by symbol. On reading $x_i$, $N$ simulates several steps (including a certain number of $\lambda$-moves) of $M$'s computation conducted during the handling of $h(x_i)$ if $h(x_i)\neq\lambda$.
Since $h$ has a finite domain, we can embed its information into $N$'s inner states. If $N$ accepts $x$  using $B$ as an oracle, then the string $h(x)=h(x_1)\cdots h(x_n)$ is in $A$; otherwise, $h(x)$ is not in $A$. Thus, $h^{-1}(A)$ belongs to $\cfl_{T}^{B}\subseteq \sigmacflt{k}$.

\s

\n{\bf [reversal]}
This proof requires Corollary \ref{wedge-oracle} and proceeds by induction on $k\in\nat^{+}$.  As noted before, it suffices to show an induction step $k\geq2$. Assume that $A\in\sigmacflt{k}$. Corollary \ref{wedge-oracle}  implies that $A\in\nfa_{m}^B$ for a certain oracle $B\in\sigmacfl{k-1}\wedge \picfl{k-1}$. Let $M = (Q,\Sigma,\{\cent,\dollar\},\Theta,\delta,q_0,Q_{acc},Q_{rej})$  be an oracle nfa that $m$-reduces $A$ to $B$.
We aim at proving  that the reversal $A^R =\{x^R\mid x\in A\}$ also belongs to $\sigmacflt{k}$.   Here, we will construct the desired reversing machine $M_{R}$ with another oracle $B^R$. First, we conveniently set our new input instance is of the form $\dollar x^R\cent$.  Intuitively, we need to ``reverse'' the entire computation of $M$, starting at an accepting configuration with the head staying at $\dollar$ and ending at an initial configuration. To meet the runtime requirement for $M^R$, we assume,  without loss of generality, that the number of consecutive $\lambda$-moves made by $M$ is at most a constant, say, $c$.

To make the following description simple, we further assume that (1) $M$ has only one accepting state, say, $q_{acc}$ and one rejecting state, say, $q_{rej}$ and (2) $M$ should enter a halting state after reading $\dollar$, and (3) when $M$ enters an accepting state just after reading $\dollar$ (and possibly making $\lambda$-moves).

Here, let us consider a situation that $M$ produces a query word $y$ and receives its oracle answer $b$. Since we try to reverse the entire computation of $M$, conceptually, we need to design a reversing machine $M_R$ to produce the reversed word $y^R$ on the query tape. To make this strategy work, we also need the reversed oracle $B^R = \{y^R\mid y\in B\}$ in lieu of $B$.

The formal description of the transition function $\delta_R$ of the above $M_R$ is given as follows. For ease of description, we first modify $M$'s transition function $\delta$ to meet the following condition: after the input-tape head moves to a new tape cell, $M$ reads a tape symbol and then makes exactly $c$ $\lambda$-moves. Given $(p,\sigma)$, define $\delta^{(\lambda)}(p,\sigma)$ to be a set of all pairs $(p,w)$ such that there are $c+1$ triplets $(r_0,\sigma,\tau_0), (r_1,\lambda,\tau_1), (r_2,\lambda,\tau_2),\ldots,(r_{c},\lambda,\tau_{c})$ satisfying
$(r_{1},\tau_{1})\in \delta(r_0,\sigma)$ and
$(r_{i+1},\tau_{i+1})\in \delta(r_i,\lambda)$ for all $i\in[c-1]$, where  $r_0=p$ and $w=\tau_0\tau_1\cdots\tau_c$. When  $(q_{acc},w)\in\delta^{(\lambda)}(p,\dollar)$ with $w=\tau_0\tau_1\cdots\tau_c$ (where $\tau_i\in\Theta\cup\{\lambda\}$), let $(\track{p\dollar}{c-1},\tau_c)\in \delta_{R}(\track{q_0}{0},\cent)$. In contrast, when $(q,w)\in\delta^{(\lambda)}(q_0,\cent)$, we set  $(\track{q\cent}{c-1},\tau_1)\in \delta_{R}(\track{q}{0},\dollar)$. Moreover, in general, for each $i\in[c-1]$, we define  $(\track{p\sigma}{m-1},\tau_m)\in\delta_{R}(\track{p\sigma}{m},\lambda)$ and $(\track{p}{0},\tau_1)\in\delta_{R}(\track{p\sigma}{1},\lambda)$. Finally, let $(\track{p\sigma}{c-1},\tau_c)\in\delta_{R}(\track{q}{0},\sigma)$ if $(q,w)\in\delta^{(\lambda)}(p,\sigma)$.

Since $B$ is  in $\sigmacfl{k-1}\wedge \picfl{k-1}$, we take two languages $B_1,B_2\in \sigmacfl{k-1}$ for which $B=B_1\cap\overline{B}_2$. Our induction hypothesis then ensures that  $B_1^R$ and $B_2^R$ belong to  $\sigmacflt{k-1}$; thus, $B^R = B_1^R \cap \overline{B}_2^R$ holds, because $\overline{B}_2^R = \overline{B_2^R}$. We then conclude that $B^R$ is in $\sigmacfl{k-1}\wedge \picfl{k-1}$.  By the definition of $M_R$ and $B^R$,  $L^R$ can be recognized by $M_R$ relative to $B^R$.
\end{proof}


In Example \ref{ex:DUP_2}, we have seen that the two languages $Dup_{2}$ and $Dup_{3}$ are in $\sigmacflt{2}$. Since they are not context-free, these examples actually prove that $\sigmacflt{1}\neq\sigmacflt{2}$; that is, $\sigmacfl{2}\nsubseteq\cfl$.  Since $\co\cfl\nsubseteq\cfl/n$ \cite{Yam08} and $\co\cfl\subseteq \sigmacflt{2}$, we obtain a slightly improved separation  as shown in Proposition \ref{CFLH-space-n}.

\begin{proposition}\label{CFLH-space-n}
$\sigmacflt{2}\nsubseteq\cfl/n$.
\end{proposition}


Let us recall the language family $\bhcfl$, the Boolean hierarchy over $\cfl$. Here, we will show that the second level of the CFL hierarchy  contains $\bhcfl$.

\begin{proposition}\label{BHCFL-in-second-level}
$\bhcfl\subseteq \sigmacflt{2}\cap\picflt{2}$.
\end{proposition}

\begin{proof}
Obviously, the containment $\cfl_{1} \subseteq  \sigmacflt{2}$ holds. It is therefore enough  to show that $\cfl_{k}\subseteq \sigmacflt{2}$ for every index $k\geq2$. For this purpose, we wish to present a simple characterization of the $k$th level of the Boolean hierarchy over $\cfl$,  despite the fact that $\cfl\wedge \cfl\neq \cfl$.

\begin{claim}\label{CFL-decomposition}
For every index $k\geq1$, $\cfl_{2k} = \bigvee_{i\in[k]}\cfl_{2}$ and $\cfl_{2k+1} = (\bigvee_{i\in[k]}\cfl_{2}) \vee \cfl$.
\end{claim}

\begin{proof}
It is shown in \cite[Claim 4]{YK13} that,  for the family $\bcfl$ of {\em  bounded context-free languages}, $\bcfl_{2k} = \bcfl_{2k-2}\vee \bcfl_{2}$ (and thus $\bcfl_{2k} = \bigvee_{i\in[k]}\bcfl_{2}$ follows). The essentially same proof works to verify that $\cfl_{2k}= \bigvee_{i\in[k]}\cfl_{2}$.  Moreover, since $\cfl_{2k+1} = \cfl_{2k}\vee\cfl$ by the definition, we obtain $\cfl_{2k+1} = (\bigvee_{i\in[k]}\cfl_{2}) \vee \cfl$.
\end{proof}

Next, we want to show that $\cfl_{2k},\cfl_{2k+1}\subseteq \sigmacflt{2}$ for all indices $k\geq1$.
The proof proceeds by induction on $k\geq1$.
Our starting point is the following claim.

\begin{claim}\label{CFL_2-contained-Sigma_2}
$\cfl_{2}\subseteq \sigmacflt{2}$.
\end{claim}

\begin{proof}
Let $L$ be any language in $\cfl_{2}$ and take two context-free languages $A$ and $B$ satisfying  $L = A\cap\overline{B}$. Let $M$ be an appropriate npda recognizing $A$. Consider the following procedure: on input $x$, copy $x$ to the query tape and, at the same time, simulate $M$ on $x$. When $M$ enters an accepting state along a certain computation path, make a query $x$ on the word $x$ to $\overline{B}$ and wait for its oracle answer. This procedure demonstrates  that $L$ is in $\cfl_{m}^{\overline{B}}$, which is included in $\cfl_{m}^{\co\cfl} \subseteq \cfl_{T}^{\co\cfl} = \sigmacflt{2}$ by Lemma \ref{basic-Turing}.
\end{proof}

Assuming $k\geq2$, let us consider the language family $\cfl_{2k}$.
Claim \ref{CFL-decomposition} implies  that $\cfl_{2k}  = \bigvee_{i\in[k]}\cfl_{2}$.  Since $\cfl_{2}\subseteq\sigmacflt{2}$ by Claim \ref{CFL_2-contained-Sigma_2}, we obtain $\cfl_{2k} \subseteq  \bigvee_{i\in[k]}\sigmacflt{2}$.  As is shown in  Lemma \ref{basic-Boolean-op},  $\sigmacflt{2}$ is closed under union, and thus this fact implies that $\cfl_{2k}\subseteq \sigmacflt{2}$.
Next,  let us consider $\cfl_{2k+1}$ for $k\geq1$. Since $\cfl_{2k+1} = \cfl_{2k}\vee \cfl$ by the definition, the above argument implies that $\cfl_{2k+1}\subseteq \sigmacflt{2}\vee \cfl$. Since $\cfl\subseteq \sigmacflt{2}$ and the closure property of $\sigmacflt{2}$ under union, it follows that  $\cfl_{2k+1} \subseteq \sigmacflt{2}\vee \sigmacflt{2} =\sigmacflt{2}$.  As a consequence, we conclude that $\cfl_{2k},\cfl_{2k+1}\subseteq \sigmacflt{2}$. Therefore, $\bhcfl\subseteq \sigmacflt{2}$ holds.

Furthermore,  we will prove that $\bhcfl\subseteq \picflt{2}$. It is possible to prove by induction on $k\in\nat^{+}$ that  $\co\cfl_{k}\subseteq \cfl_{k+1}$. {}From this inclusion, we obtain $\co\bhcfl \subseteq \bhcfl$. By symmetry, $\bhcfl\subseteq \co\bhcfl$ holds. Thus, we conclude that $\bhcfl=\co\bhcfl$. Therefore, the earlier assertion $\bhcfl\subseteq \sigmacflt{2}$ implies $\bhcfl\subseteq \picflt{2}$ as well.
\end{proof}


Let us turn our attention to the complexity of $\cfl(\omega)$. Wotschke \cite{Wot78} proved that $\cfl(\omega)\subsetneqq \bhcfl$. We will show the  inclusion $\cfl(\omega) \subseteq \bhcfl$ by conducting a direct estimation of each language family $\cfl(k)$ in $\cfl(\omega)$.

\begin{proposition}\label{CFL(omega)-in-Pi2} {\em \cite{Wot78}}
$\cfl(\omega)\subseteq \bhcfl$ (thus, $\cfl(\omega)\subseteq \sigmacflt{2}\cap\picflt{2}$).
\end{proposition}

\begin{proof}
A key to the proof of  the first part of this proposition is the following claim.

\begin{claim}\label{CFL(k)-in-bool-CFL}
For every index $k\geq1$, $\cfl(k)\subseteq \cfl_{2k+1}$ holds.
\end{claim}

\begin{proof}
We will prove the claim by induction on $k\geq1$. When $k=1$, the claim is obviously true since $\cfl(1)=\cfl_{1}\subseteq \cfl_{3}$. For induction step, assume that $k\geq2$.  Our induction hypothesis states that $\cfl(k-1)\subseteq \cfl_{2k-1}$. Since $\cfl(k)=\cfl(k-1)\wedge \cfl$,  we obtain $\cfl(k)\subseteq \cfl_{2k-1}\wedge \cfl$. In contrast, it follows by the definition that $\cfl_{2k+1} =\cfl_{2k}\vee \cfl =  (\cfl_{2k-1}\wedge \co\cfl)\vee \cfl$.  The last term equals $(\cfl_{2k-1}\vee \cfl)\wedge (\cfl\vee \co\cfl)$. Clearly, this language family  includes $\cfl_{2k-1}\wedge \cfl$ as a subclass. Therefore, we conclude that $\cfl(k)\subseteq \cfl_{2k+1}$.
\end{proof}

By Claim \ref{CFL(k)-in-bool-CFL}, it follows that $\cfl(\omega) = \bigcup_{k\in\nat^{+}}\cfl(k) \subseteq \bigcup_{k\in\nat^{+}}\cfl_{2k+1} \subseteq \bhcfl$.  The second part of the proposition follows from Proposition \ref{BHCFL-in-second-level}
\end{proof}

Let us argue that the language family $\cfl_{m}^{\cfl(\omega)}$ is located within the third level of the $\cfl$ hierarchy.

\begin{proposition}\label{upper-bound-nine}
$\cfl_{m}^{\cfl(\omega)}\subseteq \sigmacflt{3}$.
\end{proposition}

\begin{proof}
Proposition \ref{CFL(omega)-in-Pi2} implies that $\cfl_{m}^{\cfl(\omega)}$ is included in $\cfl_{m}^{\bhcfl}$. By Proposition \ref{BHCFL-in-second-level}, it follows that  $\cfl_{m}^{\bhcfl}$ is included in $\cfl_{m}(\picflt{2})$,  which is obviously a subclass of $\cfl_{T}(\picflt{2})=\sigmacflt{3}$ by Lemma \ref{basic-Turing}.
\end{proof}

\subsection{Structural Properties of the CFL Hierarchy}\label{sec:structure-hierarchy}

After having established fundamental properties of languages in the $\cfl$ hierarchy in Section \ref{sec:Turing-reduction}, we wish to explore more  structural properties that characterize the $\cfl$ hierarchy.  Moreover, we intend to present three alternative characterizations (Theorem \ref{Sigma-m-vs-T} and Proposition \ref{logical-charact}) of the hierarchy.

Let us consider a situation in which Boolean operations ($\wedge$ and $\vee$) are applied to languages in the $\cfl$ hierarchy. We begin with general results.

\begin{lemma}\label{basic-Boolean-op}
\begin{enumerate}
  \setlength{\topsep}{-2mm}%
  \setlength{\itemsep}{1mm}
  \setlength{\parskip}{0cm}

\item $\sigmacflt{k}\vee \sigmacflt{k} = \sigmacflt{k}$ and $\picflt{k}\wedge\picflt{k}=\picflt{k}$ for any $k\geq1$.
\item $\sigmacflt{k}\wedge \picflt{k} \subseteq \sigmacflt{k+1}\cap \picflt{k+1}$ and $\sigmacflt{k}\vee \picflt{k} \subseteq \sigmacflt{k+1}\cap \picflt{k+1}$ for any $k\geq1$.
\end{enumerate}\vs{-1}
\end{lemma}

\begin{proof}
In what follows, we are focused only on the $\sigmacflt{k}$ case since the $\picflt{k}$ case is symmetric.

(1) When $k=1$, since $\cfl$ is closed under union, $\cfl\vee\cfl=\cfl$  follows immediately. Assume that $k\geq2$. Obviously,  $\sigmacfl{k}$ is contained within $\sigmacfl{k}\vee\sigmacfl{k}$. Next, let $L$ be any language in $\sigmacflt{k}\vee\sigmacflt{k}$. We then take two oracle npda's $M_0$ and $M_1$ and two languages $A,B\in\picflt{k-1}$ satisfying that $L = L(M_0,A)\cup L(M_1,B)$.  Our goal is to show that $L\in \sigmacflt{k}$.  Let us consider
another oracle npda $M$ that behaves as follows. On input $x$, $M$ guesses a bit $b$, writes it down on a query tape, and  simulates $M_b$ on $x$. Thus, when $M_b$ halts with a query word $y_b$ produced on its query tape, $N$ does the same with $by_b$. Let us define $C$ as the union  $\{0y\mid y\in A\}\cup \{1y\mid y\in B\}$.
We argue that $C$ is in $\picflt{k-1}$. To see this fact, consider the complement $\overline{C}$. Note that $\overline{C} = \{0y\mid y\in\overline{A}\}\cup \{1y\mid y\in\overline{B}\}$. Because $\overline{A},\overline{B}\in\sigmacflt{k-1}$,  by our induction hypothesis, $\overline{C}$ belongs to $\sigmacflt{k-1}$. Since $L = L(N,C)$ holds,  we conclude that $L\in\cfl_{T}^C\subseteq \sigmacflt{k}$.

(2) Assuming  $k\geq2$, let  $L$ be any language in $\sigmacflt{k}\wedge\picflt{k}$ and take an oracle npda $M$ and two languages $A\in\picflt{k-1}$ and $B\in\picflt{k}$ for which $L = L(M,A)\cap B$.
Here, we define a new oracle npda $N$ to simulate $M$ on input $x$ and generate  an encoding   $\track{\tilde{x}}{\tilde{y}}$ of  $x$ and query word $y$ along a computation of $M$ on $x$.
Next, let us define $C$ as the set $\{\track{\tilde{x}}{\tilde{y}}\mid x\in A, y\in B\}$, which belongs to $\picfl{k}\wedge \picfl{k}\subseteq \picflt{k}$ by (1) using  a fact that $\picflt{k-1}\subseteq \picflt{k}$. Since $L = L(N,C)$, $L$ belongs to $\cfl_{T}^{C}$, which is a subclass of $\sigmacflt{k+1}$. In a similar fashion, we can prove that $L\in\picflt{k+1}$. Therefore, we obtain $\sigmacfl{k}\wedge \sigmacfl{k}\subseteq \sigmacfl{k+1}\cap\picfl{k+1}$.
\end{proof}

What is missing in the list of Lemma a\ref{basic-Boolean-op} is two language families $\sigmacflt{k}\wedge \sigmacflt{k}$ and $\picflt{k}\vee\picflt{k}$. It is well-known that $\cfl\wedge \cfl =\cfl(2)\neq \cfl$. Therefore, the equality $\sigmacflt{k}\wedge \sigmacflt{k} = \sigmacflt{k}$  does not hold in the first level (\ie $k=1$). Surprisingly, we can prove that this equality actually holds for any level {\em more than}  $1$.

\begin{proposition}\label{wedge-closure}
$\sigmacflt{k}\wedge \sigmacflt{k} = \sigmacflt{k}$ and $\picflt{k}\vee\picflt{k} = \picflt{k}$ for all  levels  $k\geq2$.
\end{proposition}


This proposition is not quite trivial and its proof requires two new characterizations of $\sigmacflt{k}$ in terms of two many-one reducibilities.  These characterizations are, in fact, a natural extension of Claim \ref{NFA-to-DYCK} and, for our purpose, we want to introduce two many-one hierarchies. The {\em many-one CFL hierarchy} consists of language families
$\sigmacflm{k}$ and $\picflm{k}$ ($k\in\nat^{+}$) defined as follows:  $\sigmacflm{1} = \cfl$,  $\picflm{k} = \co\sigmacflm{k}$, and   $\sigmacflm{k+1} = \cfl_{m}(\picflm{k})$ for any level $k\geq1$, where the subscript ``$m$'' stands for ``many-one'' as before.
Relative to oracle $A$, a {\em relativized many-one NFA hierarchy}, which was essentially formulated in \cite{Rei90}, is defined as follows: $\relsigmanfa{1}{A} = \nfa_{m}^{A}$, $\relpinfa{k}{A}= \co\relsigmanfa{k}{A}$,  and $\relsigmanfa{k+1}{A} = \nfa_{m}(\relpinfa{k}{A})$ for every index $k\geq1$. Given a language family $\CC$,   the notation $\relsigmanfa{k}{\CC}$ (or $\sigmanfa{k}(\CC)$) denotes the union $\bigcup_{A\in\CC} \relsigmanfa{k}{A}$.

\begin{theorem}\label{Sigma-m-vs-T}
$\sigmacflt{k} = \sigmacflm{k} = \sigmanfa{k}(DYCK)$ for every index $k\geq1$.
\end{theorem}

Since the proof of Theorem \ref{Sigma-m-vs-T} is involved, prior to the proof, we intend to demonstrate how to prove  Proposition \ref{wedge-closure} using this theorem.

\vs{-2}
\begin{proofof}{Proposition \ref{wedge-closure}}
In what follows, it suffices to prove that $\sigmacflt{k}\wedge \sigmacflt{k}=\sigmacflt{k}$, since $\picflt{k}\vee\picflt{k}=\picflt{k}$ is obtained by symmetry. First, take any language $L$ in $\sigmacflt{k}\wedge \sigmacflt{k}$ and assume that $L =L_1\cap L_2$ for two languages $L_1,L_2\in\sigmacflt{k}$. Theorem \ref{Sigma-m-vs-T} implies that $L_1$ and $L_2$ are both  in $\sigmanfa{k}(DYCK)$.
Let us choose oracle nfa's $M_1$ and $M_2$ that respectively recognize $L_1$ and $L_2$ relative to oracles $A_1$ and $A_2$, where $A_1,A_2\in\pinfa{k-1}(DYCK)$. Let us consider  a new npda $N$ that works in the following fashion. In scanning  each input symbol,  say, $\sigma$, $N$ simulates in parallel one or more steps of $M_1$ and $M_2$ using two sets of inner states for $M_1$ and $M_2$.  Such a parallel simulation of two machines is possible because $M_1$ and $M_2$ use no stacks. Moreover, whenever  $M_1$ (resp., $M_2$) tries to write a symbol, $N$ writes it on the upper (resp., lower) track of its single query tape. To write two query strings $y_1$ and $y_2$ of $M_1$ and $M_2$, respectively, onto $N$'s query tape, we actually write their $\natural$-extensions. Now, let $\track{\tilde{y}}{\tilde{z}}$ denote a query string produced by $N$ so that $\track{\tilde{y}}{\tilde{z}}$ encodes query words $y$ and $z$ along two computation paths of $M_1$ and $M_2$. A new oracle $B$ is finally set to be  $\{\track{\tilde{y}}{\tilde{z}}\mid y\in A_1,z\in A_2\}$, which is in $\pinfa{k-1}(DYCK)\wedge \pinfa{k-1}(DYCK)$.
Since $\picflt{k-1}=\pinfa{k-1}(DYCK)$ by Theorem \ref{Sigma-m-vs-T},  Lemma \ref{basic-Boolean-op}(1) ensures that $B$ is also in $\pinfa{k-1}(DYCK)$.
The above definitions show that  $N$ $m$-reduces $L$ to $B$. Therefore, it immediately follows that $L\in\nfa_{m}^{B}\subseteq \sigmanfa{k}(DYCK) = \sigmacflt{k}$.
\end{proofof}

The first step toward the proof of Theorem \ref{Sigma-m-vs-T} is  to prove a key lemma given below.

\begin{lemma}\label{T-reduction-recursive}
For every index $k\geq1$, it holds that  $\sigmacflt{k+1} \subseteq  \cfl_{m}(\sigmacflt{k}\wedge \picflt{k}) \subseteq  \nfa_{m}(\sigmacflt{k}\wedge \picflt{k})$.
\end{lemma}

\begin{proof}
The proof of the lemma proceeds by induction on $k\in\nat^{+}$. Notice that the base case $k=1$ has been already proven as Proposition \ref{first-level-equal}. Therefore, in what follows, we aim at the induction step of $k\geq2$ by proving separately the following two containments: (1) $\sigmacflt{k+1}\subseteq \cfl_{m}(\sigmacflt{k}\wedge \picflt{k})$ and (2) $\cfl_{m}(\sigmacflt{k}\wedge \picflt{k}) \subseteq  \nfa_{m}(\sigmacflt{k}\wedge \picflt{k})$.

(1)  Let us recall the proof of Proposition \ref{first-level-equal}, in particular, the proof of the following inclusion: $\cfl_{T}^{\cfl} \subseteq \cfl_{m}(\cfl\wedge\co\cfl)$.   We note that this proof   is {\em relativizable} (that is, it works when we append an oracle to underlying npda's).  To be more precise, essentially the same proof proves that $\cfl_{T}(\cfl_{T}^A) \subseteq \cfl_{m}(\cfl_{T}^A\wedge \co\cfl_{T}^A)$ for any oracle $A$. If we choose an arbitrary language in $\sigmacflt{k-1}$ as $A$, then we conclude that $\sigmacflt{k+1}\subseteq \cfl_{m}(\sigmacflt{k}\wedge \picflt{k})$.

(2)  By setting $\CC = \sigmacflt{k}\wedge \picflt{k}$ in Lemma \ref{extendible-inclusion}, we obtain $\cfl_{m}(\sigmacflt{k}\wedge\picflt{k}) \subseteq \nfa_{m}(\dcfl\wedge \CC)$. Note that $\dcfl\wedge\CC = \dcfl\wedge (\sigmacflt{k}\wedge  \picflt{k}) = \sigmacflt{k}\wedge (\dcfl\wedge \picflt{k})$. Since $\dcfl\subseteq \picflt{k}$, it instantly follows that $\dcfl\wedge \picflt{k}\subseteq \picflt{k}\wedge \picflt{k} = \picflt{k}$ by Lemma \ref{basic-Boolean-op}(1).  In summary, we obtain the desired inclusion  $\nfa_{m}(\dcfl\wedge \CC) \subseteq \nfa_{m}(\sigmacflt{k}\wedge \picflt{k})$.
\end{proof}

The second step is to establish the following inclusion relationship between two language families $\nfa_{m}(\sigmacflm{k}\wedge \picflm{e})$ and $\cfl_{m}(\picflm{e})$.

\begin{lemma}\label{NFA-to-CFL-Pi}
For any two indices $k\geq1$ and $e\geq k-1$, it holds that $\nfa_{m}(\sigmacflm{k}\wedge \picflm{e}) \subseteq \cfl_{m}(\picflm{e})$.
\end{lemma}

\begin{proof}
Let $L$ be any language in $\nfa_{m}^A$, where $A$ is a certain language in $\sigmacflm{k}\wedge \picflm{e}$.  Here, we express $A$ as  $A_1\cap A_2$ by choosing properly two languages $A_1\in\cfl_{m}^{B_1}$  and $A_2\in\picflm{e}$, where $B_1\in\picflm{k-1}$. Let $M$ be an oracle nfa that $m$-reduces $L$ to $A$ and let $M_1$ be an oracle npda $m$-reducing $A_1$ to $B_1$.
Our goal is to construct an oracle npda $N$ and an oracle $C$ for $L$. The desired machine $N$ takes  input $x$ and  simulates $M$ on $x$. When $M$ tries to write a symbol, say, $\sigma$, $N$ writes $\sigma$ on the upper track of its query tape and also simulates, using a stack, one step of $M_1$ while reading $\sigma$ and more $\lambda$-moves of $M_1$.  When $M_1$ writes a symbol, $N$ uses the lower track to keep the symbol. Finally, $N$ produces a query string of the form $\track{\tilde{y}}{\tilde{z}}$, where $y$ is a query word of $M$ and $z$ is a query word of $M_1$. Next, we define $C$ to be $\{\track{\tilde{y}}{\tilde{z}}\mid y\in A_2, z\in B_1\}$ so that $N$ $m$-reduces $L$ to $C$. The language $C$  obviously  belongs to $\picflm{e}\wedge \picflm{k-1}$, which equals $\picflm{e}$ by Lemma \ref{basic-Boolean-op}(1), since $e\geq k-1$. Therefore, $L$ is in $\cfl_{m}^{C}\subseteq \cfl_{m}(\picflm{e})$.
\end{proof}

Finally, we are ready to give the proof of Theorem \ref{Sigma-m-vs-T}.

\vs{-2}
\begin{proofof}{Theorem \ref{Sigma-m-vs-T}}
The proof of the theorem proceeds by induction on $k\geq1$. Since Lemma \ref{non-closure-many-one} handles the base case $k=1$, it is sufficient to  assume that $k\geq2$.  First, we target the second equality  given in the theorem; namely, $\sigmacflm{k} = \sigmanfa{k}(DYCK)$.

Since  $k\geq2$, let us  assume that $L\in\cfl_{m}^{A}$ for a certain language $A$ in $\picflm{k-1}$. A proof similar to that of Claim \ref{NFA-to-DYCK} proves the existence of a certain Dyck language $D$
satisfying that $\cfl_{m}^{A} = \nfa_{m}^{B}$, where $B$ is of the form $\{ \track{\tilde{y}}{\tilde{z}} \mid y\in D, z\in A\}$ and $\tilde{y}$ and $\tilde{z}$ are appropriate $\natural$-extensions of $y$ and $z$, respectively.  This definition places $B$ into the language family $\dcfl\wedge \picflm{k-1}$, which equals $\picflm{k-1}$ by Lemma \ref{basic-Boolean-op}(1) because of $k\geq2$. By our induction hypothesis  $\picflm{k-1} = \pinfa{k-1}(DYCK)$, it follows that  $\nfa_{m}^{B} \subseteq \nfa_{m}(\picfl{m,k-1})\subseteq \nfa_{m}(\pinfa{k-1}(DYCK)) = \sigmanfa{k}(DYCK)$, and thus we obtain $L\in\sigmanfa{k}(DYCK)$.

Next, we aim at establishing the first equality $\sigmacfl{k}=\sigmacfl{m,k}$ in the theorem.  Clearly,  $\sigmacflm{k}\subseteq \sigmacflt{k}$ holds since $\cfl_{m}^A\subseteq\cfl_{T}^A$ for any oracle $A$.
Henceforth, we target
the other inclusion.
Lemma \ref{T-reduction-recursive} implies that $\sigmacflt{k} \subseteq \nfa_{m}(\sigmacflt{k-1}\wedge \picflt{k-1})$ by our induction hypothesis.  Since $\sigmacflt{k-1}=\sigmacflm{k-1}$,  we obtain $\sigmacflt{k}\subseteq \nfa_{m}(\sigmacflm{k-1}\wedge \picflm{k-1})$. By Lemma \ref{NFA-to-CFL-Pi}, it further follows that $\nfa_{m}(\sigmacflm{k-1}\wedge \picflm{k-1})\subseteq \cfl_{m}(\picflm{k-1})=\sigmacflm{k}$.  In conclusion, $\sigmacflt{k}\subseteq \sigmacflm{k}$ holds.
\end{proofof}

Lemmas \ref{T-reduction-recursive} and \ref{NFA-to-CFL-Pi} together with Theorem \ref{Sigma-m-vs-T} lead to the following characterization of $\sigmacfl{k+1}$.

\begin{corollary}\label{wedge-oracle}
For any $k\geq1$, $\sigmacfl{k+1} = \cfl_{m}(\sigmacfl{k}\wedge\picfl{k}) = \nfa_{m}(\sigmacfl{k}\wedge\picfl{k})$.
\end{corollary}

\begin{proof}
Lemma \ref{T-reduction-recursive} shows that $\sigmacfl{k+1}\subseteq \cfl_{m}(\sigmacfl{k}\wedge\picfl{k})\subseteq \nfa_{m}(\sigmacfl{k}\wedge\picfl{k})$. It thus suffices to show that $\nfa_{m}(\sigmacfl{k}\wedge\picfl{k})\subseteq \sigmacfl{k+1}$. By Theorem  \ref{Sigma-m-vs-T}, we obtain $\sigmacfl{k}=\sigmacfl{m,k}$. Hence, by Lemma \ref{NFA-to-CFL-Pi}, it follows that $\nfa_{m}(\sigmacfl{k}\wedge\picfl{k}) \subseteq \nfa_{m}(\sigmacfl{m,k}\wedge\picfl{m,k}) \subseteq \cfl_{m}(\picfl{m,k})= \cfl_{m}(\picfl{k})=\sigmacfl{k+1}$.
\end{proof}


An {\em upward collapse property} holds for the $\cfl$ hierarchy except for the first level.
Similar to the notation $\cfl_{e}$ expressing  the $e$th level of the Boolean hierarchy over $\cfl$,  a new notation  $\sigmacflt{k,e}$ is introduced to denote the $e$th level of the Boolean hierarchy over $\sigmacflt{k}$. Additionally, we set $\mathrm{BH}\sigmacflt{k} = \bigcup_{e\in\nat^{+}}\sigmacflt{k,e}$.  Notice that  $\mathrm{BH}\sigmacflt{1}$ coincides with $\bhcfl$.

\begin{lemma}\label{upward-collapse}
(upward collapse properties) Let $k$ be any integer at least $2$.
\begin{enumerate}\vs{-1}
  \setlength{\topsep}{-2mm}%
  \setlength{\itemsep}{1mm}
  \setlength{\parskip}{0cm}

\item $\sigmacflt{k} = \sigmacflt{k+1}$ if and only if $\cflh = \sigmacflt{k}$.
\item $\sigmacflt{k}=\picflt{k}$ if and only if $\mathrm{BH}\sigmacflt{k} = \sigmacflt{k}$.
\item $\sigmacflt{k}=\picflt{k}$ implies $\sigmacflt{k}=\sigmacflt{k+1}$.
\end{enumerate}\vs{-1}
\end{lemma}

\begin{proof}
(1)  It is obvious that $\cflh=\sigmacflt{k}$ implies $\sigmacflt{k}=\sigmacflt{k+1}$.
Next, assume that $\sigmacflt{k} = \sigmacflt{k+1}$. By applying the complementation operation,
we obtain $\picflt{k} = \picflt{k+1}$. Thus, it follows that $\sigmacflt{k+2}=\cfl_{T}(\picflt{k+1}) = \cfl_{T}(\picflt{k}) = \sigmacflt{k+1}$. Similarly, it is possible to prove by induction on $e\in\nat^{+}$ that $\sigmacflt{k+e} = \sigmacflt{k}$.  Therefore, $\cflh=\sigmacflt{k}$ holds.

(2) Since $\picflt{k} \subseteq \sigmacfl{k}\wedge \picfl{k}=\sigmacfl{k,2} \subseteq \mathrm{BH}\sigmacflt{k}$, obviously $\mathrm{BH}\sigmacflt{k}=\sigmacflt{k}$ implies $\sigmacflt{k}=\picflt{k}$. Next, assume that $\sigmacflt{k}=\picflt{k}$. By induction on $e\in\nat^{+}$, we wish to prove that $\sigmacflt{k,e}\subseteq \sigmacflt{k}$.
Firstly, let us consider  the language family  $\sigmacflt{k,2e+1}$ for $e\geq1$. Our induction hypothesis asserts that $\sigmacflt{k,2e}\subseteq \sigmacflt{k}$. It thus follows that $\sigmacflt{k,2e+1}=\sigmacflt{k,2e}\vee \sigmacflt{k} \subseteq \sigmacflt{k}\vee \sigmacflt{k}=\sigmacflt{k}$ by Lemma \ref{basic-Boolean-op}(1).  Secondly,  we consider the family  $\sigmacflt{k,2e+2}$ for $e\geq0$. We then obtain $\sigmacflt{k,2e+2} = \sigmacflt{k,2e+1}\wedge \picflt{k} \subseteq \sigmacflt{k}\wedge \picfl{k}$, where the last containment comes from our induction hypothesis.
Since $\picflt{k}=\sigmacflt{k}$, we obtain  $\sigmacflt{k,2e+2} \subseteq \sigmacflt{k}\wedge\picflt{k}=\picflt{k}\wedge\picflt{k}=\picflt{k}$ by Lemma \ref{basic-Boolean-op}(1). The last term obviously equals $\sigmacflt{k}$ from our assumption.  Overall, we conclude that $\mathrm{BH}\sigmacflt{k}=\sigmacflt{k}$.

(3) Assume that $\sigmacflt{k}=\picflt{k}$ and pay our attention to $\sigmacflt{k+1}$.  Since  Theorem \ref{Sigma-m-vs-T} yields   $\sigmacflt{k}=\sigmacflm{k}$,  our assumption is equivalent to $\sigmacflm{k}=\picflm{k}$. By Lemma \ref{T-reduction-recursive}, it follows that $\sigmacfl{k+1} \subseteq \nfa_{m}(\sigmacflm{k}\wedge \picflm{k}) = \nfa_{m}(\sigmacflm{k}\wedge \sigmacflm{k})$, which is included in $\nfa_{m}(\sigmacflm{k})$ by Lemma  \ref{basic-Boolean-op}(1). Since $\sigmacfl{k}\subseteq \sigmacflm{k}\wedge \picflm{k-1}$, Lemma \ref{NFA-to-CFL-Pi} implies that $\nfa_{m}(\sigmacflm{k})\subseteq \nfa_{m}(\sigmacflm{k}\wedge \picflm{k-1}) \subseteq \cfl_{m}(\picflm{k-1}) = \sigmacflm{k}$, which equals $\sigmacflt{k}$ by Theorem \ref{Sigma-m-vs-T} again.
\end{proof}

{}From Lemma \ref{upward-collapse}, if the Boolean hierarchy over $\sigmacflt{k}$  collapses to $\sigmacflt{k}$, then the $\cfl$ hierarchy also collapses. It is not clear, however, that a much weaker assumption like $\sigmacflt{k,e}=\sigmacflt{k,e+1}$ suffices to draw the collapse of the $\cfl$ hierarchy (for instance, $\sigmacflt{k+1}=\sigmacflt{k+2}$) for $k\geq2$. In addition, it is known that $\cfl\cup\co\cfl\subseteq \pcfl$ but we do not know whether $\sigmacfl{2}\subseteq \pcfl$.

We remark that Theorem \ref{Sigma-m-vs-T} provides us with a new logical characterization of
$\sigmacflt{k}$. For convenience, we define a function $Ext$ as  $Ext(\tilde{x}) = x$ for any $\natural$-extension $\tilde{x}$ of string $x$.

\begin{proposition}\label{logical-charact}
Let $k\geq1$. For any language $L\in\sigmacflt{k}$ over alphabet $\Sigma$, there exists another language $A\in\dcfl$ and a linear polynomial $p$ with $p(n)\geq n$ for all $n\in\nat$ that satisfy the following equivalence relation: for any number $n\in\nat$ and any string $x\in\Sigma^n$,
\begin{eqnarray*}\label{eqn:logical-charact}
x\in L &\text{ if and only if }&   \exists \tilde{x}(|\tilde{x}|\leq p(n))\, \exists y_1(|y_1|\leq p(n))\,
\forall y_2(|y_2|\leq p(n))\,   \\
&& \hs{15} \cdots  Q_k y_k(|y_k|\leq p(n))\, [\, x = Ext(\tilde{x})\wedge [\tilde{x},y_1,y_2,\ldots,y_k]^T\in A\,],
\end{eqnarray*}
where $Q_k$ is $\exists$ if $k$ is odd and is $\forall$ if $k$ is even. Moreover, $\tilde{x}$ expresses a $\natural$-extension of $x$.
\end{proposition}

\begin{proof}
We wish to prove the proposition by induction on $k\in\nat^{+}$.  First, we will target the case of $k=1$ and  we begin with the assumption that  $L\in\sigmacflt{1}$ ($=\cfl$). Moreover, we assume that $L$ is recognized by a certain npda, say, $M$ and $p$ denotes a linear polynomial that bounds the running time of $M$.  Consider the following new language $A$. This language $A$ is defined as the collection of all stings of the form  $\track{\tilde{x}}{y_1}$ that encodes an accepting computation path of $M$ on input $x$.  It is not difficult to verify that $A\in\dcfl$. The definition of $A$ indicates that, for every string $x$,  $x$ is in $L$ if and only if there exist a $\natural$-extension $\tilde{x}$ of $x$ and $\track{\tilde{x}}{y_1}$ is in $A$. The latter condition is logically equivalent to $\exists \tilde{x}\,(|\tilde{x}|\leq p(|x|))\exists y_1\,(|y_1|\leq p(|x|))\, [\,x=Ext(\tilde{x})\wedge \track{\tilde{x}}{y_1}\in A\,]$, because $M$ produces only strings of linear size.

For induction step $k\geq2$,  let us assume that $L\in\sigmacflt{k}$. Theorem \ref{Sigma-m-vs-T} implies  that $L$ is many-one $\nfa$-reducible to a certain oracle $B$ in
 $\pinfa{k-1}(DYCK)$ ($= \picflt{k-1}$)
via an oracle nfa $M$.  Note that the running time of $M$ is upper-bounded by a certain linear polynomial, say, $p$.
Since $\overline{B}$ is in $\sigmacflt{k-1}$, our induction hypothesis ensures that there are a linear polynomial $q$ and  a language $C$ in $\dcfl$ such that, for any $n'\in\nat$ and for any string $z_1$ of length $n'$, $z_1$ is in $\overline{B}$ if and only if   $\exists \tilde{z}'_1(|\tilde{z}'_1|\leq q(n'))\, \exists u_2(|u_2|\leq q(n'))\,  \cdots  Q'_k u_k(|u_k|\leq q(n'))\, [\, [\,z_1=Ext(\tilde{z}'_1)\wedge [\tilde{z}'_1,u_2,\ldots,u_k]^T\in C\,]$,  where $n'=|z_1|$ and $Q'_k$ is $\exists$ if $k$ is even, and $\forall$ if $k$ is odd.  In a similar way as in the base case of $k=1$, we can define a language $D$ in $\dcfl$ that is composed of strings $[\tilde{x},u_1,\tilde{z}_1]^T$ that encodes an accepting computation path of $M$ on input $x$ with query word $z_1$, where $u_1$ contains a piece of information on this accepting computation path.
Note that, for any given pair $(\tilde{x},u_1)$, there is at most one string $\tilde{z}_1$ such that $[\tilde{x},u_1,\tilde{z}_1]^T\in D$. {}From such a unique $\tilde{z}_1$,  a string $z_1=Ext(\tilde{z}_1)$ is also uniquely determined.  Note that $x$ is in $L$ if and only if there exist $\natural$-extensions $\tilde{x}$ of input $x$ and $\tilde{z}'_1$ of query word $z_1$ and also $u_1$ such that
$[\tilde{x},u_1,\tilde{z}_1]^T\in D \wedge z_1\in \overline{B}$.

To complete the proof, we want to combine two strings $[\tilde{x},u_1,\tilde{z}_1]^T$ and $[\tilde{z}'_1,u_2,\ldots,u_k]^T$ satisfying $Ext(\tilde{z}_1)=Ext(\tilde{z}'_1)$ into a single string by applying a technique of inserting $\natural$ so that a single tape head can read off all information from this string at once.
For convenience, we introduce three languages.
Let $D'=\{\track{\tilde{x}}{\tilde{w}_1}\mid \exists u_1,\tilde{z}_1\,[ Ext(\tilde{w}_1) =  \track{u_1}{\tilde{z}_1} \wedge  [\tilde{x},u_1,\tilde{z}_1]^T\in D \,]\}$,
$C' = \{[\tilde{w}_2,y_3,\ldots,y_k]^T\mid  \exists \tilde{z}'_1,y_2,\ldots,y_k \,[  Ext(\tilde{w}_2) = \track{\tilde{z}'_1}{y_2}  \wedge (\bigwedge_{i=3}^{k} Ext(y_i)=u_i ) \wedge [\tilde{z}'_1,u_2,\ldots,u_k]^T\in C \,]\}$, and $E=\{\track{\tilde{w}_1}{\tilde{z}'_1} \mid \exists u_1,\tilde{z}_1\,[ Ext(\tilde{w}_1) = \track{u_1}{\tilde{z}_1}\wedge Ext(\tilde{z}_1) = Ext(\tilde{z}'_1) \,]  \}$.  Finally, we define a language $G =\{ [\tilde{x},y_1,y_2,\ldots,y_k]^T \mid \track{\tilde{x}}{y_1}\in D' \wedge \track{y_1}{y_2}\in E \wedge   [y_1,y_2,\ldots,y_k]^T \not\in C'\}$.  It is not difficult to show that $G$ is in $\dcfl$.
Now, let $r(n)=q(p(n))$ for all $n\in\nat$.
With this language $G$,  it follows that, for any $n\in\nat$ and for any string $x$ of length $n$,  $x$ is in $L$ if and only if
$\exists \tilde{x}(|\tilde{x}|\leq r(n))\, \exists y_1(|y_1|\leq p(n))\, \forall y_2(|y_2|\leq r(n))\,  \ldots  Q_k y_k(|y_k|\leq r(n))\, [\, x=Ext(\tilde{x})\wedge   [\tilde{x},y_1,y_2,\ldots,y_k]^T \in G\,]$.
Therefore, we have completed the induction step.
\end{proof}

In the end of this section, we will give a simple complexity upper bound of $\cflh$.

\begin{proposition}\label{DSPACE-upper}
$\cflh\subseteq  \mathrm{DSPACE}(O(n))$.
\end{proposition}

\begin{proof}
By induction on $k\in\nat^{+}$, we intend to show that $\sigmacfl{k}\subseteq\mathrm{DSPACE}(O(n))$. When $k=1$, $\cfl$ belongs to $\mathrm{DSPACE}(O(\log^2{n}))$ \cite{LSH65}, which is obviously included in $\mathrm{DSPACE}(O(n))$. Let $k\geq2$. Here, we claim the following closure property of $\mathrm{DSPACE}(O(n))$ under many-one $\cfl$-reductions.

\begin{claim}\label{DSPACE-vs-CFL}
If $A$ is a language in $\mathrm{DSPACE}(O(n))$, then
$\cfl_{m}^{A}\subseteq \mathrm{DSPACE}(O(n))$.
\end{claim}

\begin{proof}
Let $M_A$ be a DTM that recognizes $A$ using $O(n)$ space. Take any language $L$ in $\cfl_{m}^{A}$ via an $m$-reduction npda $M$. We want to construct another $O(n)$-space DTM $N$ for $L$. Our two-way machine $N$ takes input $x$ and simulates $M$ on $x$ using a two-way read-only input tape and three read/write work tapes, two of which mimic $M$'s stack and query tape and the third one is used to simulate $M_A$.
On input $x$, $N$ simulates $M$ on $x$ using the first two work tapes.
Remark that, when $M$ makes a valid query, it must enter an accepting state. This signals $N$ to start simulating $M_A$ using the input tape and the third work tape. Since $M$'s computation paths have length at most $O(n)$, the deterministic  simulation of $M$ by $N$ requires only $O(n)$ space.
\end{proof}

By induction hypothesis, we obtain $\sigmacfl{k-1}\subseteq \mathrm{DSPACE}(O(n))$, which is equivalent to $\picfl{k-1}\subseteq \mathrm{DSPACE}(O(n))$.
It thus follows that, since $\mathrm{DSPACE}(O(n))\wedge \mathrm{DSPACE}(O(n)) \subseteq \mathrm{DSPACE}(O(n))$, $\sigmacfl{k-1}\wedge\picfl{k-1} \subseteq \mathrm{DSPACE}(O(n))$.
By Corollary \ref{wedge-oracle}, $\sigmacfl{k}=\cfl_{m}(\sigmacfl{k-1}\wedge \picfl{k-1})$ holds. Therefore,
$\sigmacfl{k} = \cfl_{m}(\sigmacfl{k-1}\wedge \picfl{k-1}) \subseteq \cfl_{m}(\mathrm{DSPACE}(O(n))) \subseteq \mathrm{DSPACE}(O(n))$, where the last containment comes from Claim \ref{DSPACE-vs-CFL}. We thus obtain the desired containment $\sigmacflt{k}\subseteq  \mathrm{DSPACE}(O(n))$.
\end{proof}

\subsection{BPCFL and  a Relativized CFL Hierarchy}\label{sec:BPCFL}

Let us consider a probabilistic version of $\cfl$. Recall the model of ppda's given in Section \ref{sec:npda-and-TM}. The {\em (two-sided) bounded-error\footnote{The notion of ``bounded error'' is in general different from the notion of ``isolated cut point'' of \cite{Rab63}.} probabilistic  language family} $\bpcfl$ consists of all languages that are recognized by ppda's whose error probability is bounded from above by an absolute constant $\varepsilon\in[0,1/2)$.  This language family is naturally contained in the {\em (two-sided) unbounded-error probabilistic language family} $\pcfl$. Notice that $\bpcfl\subseteq\pcfl$, $\bpcfl\subseteq \bpp$, and $\pcfl\subseteq \pp$. Hromkovi\v{c} and Schnitger \cite{HS10} studied properties of $\bpcfl$ and asserted that $\bpcfl$ and $\cfl$ are actually  incomparable; more accurately,  $\bpcfl\nsubseteq \cfl$ and $\cfl\nsubseteq\bpcfl$. Here, we show that the first separation can be  strengthened even in the presence of ``advice.''

\begin{proposition}\label{bpcfl-not-in-advice}
$\bpcfl\nsubseteq \cfl/n$.
\end{proposition}

\begin{proof}
Let us consider the example language $Equal_6$ in \cite{Yam08} that is composed of all strings $w$ over the  alphabet $\Sigma_6=\{a_1,a_2,\ldots,a_6,\#\}$ such that each symbol except $\#$ appears in $w$ the same number of times. It was proven that $Equal_6$ is located outside of $\cfl/n$ \cite{Yam08} by applying the \emph{swapping lemma\footnote{[Swapping Lemma for Context-Free Languages] Let $L$ be any infinite context-free language over alphabet $\Sigma$ with $|\Sigma|\geq 2$. There is a positive number $m$ that satisfies the following. Let $n$ be any positive number at least $2$ (called a {\em swapping-lemma constant}), let $S$ be any subset of $L\cap\Sigma^{n}$, and let  $j_0,k_0\in[2,n]_{\integer}$ be any two indices with $k_0 \geq 2j_0$ and $|S_{i,u}|< |S|/m(k_0-j_0+1)(n-j_0)$ for any  $i\in[0,n-j_0]_{\integer}$ and any $u\in\Sigma^{j_0}$. There exist two indices $i\in[1,n]_{\integer}$ and $j\in[j_0,k_0]_{\integer}$ with $i+j\leq n$ and two strings $x =x_1x_2x_3$ and $y=y_1y_2y_3$ in $S$ with $|x_1|=|y_1|=i$, $|x_2|=|y_2|=j$, and $|x_3|=|y_3|$ such that (i) $x_2\neq y_2$, (ii) $x_1y_2x_3\in L$, and (iii) $y_1x_2y_3\in L$. Here, the notation $S_{i,u}$ denotes the set $\{x\in S\mid \exists y,z\,[x=yuz\wedge |y|=i\,]\}$.  (See  \cite{Yam08} for the proof.)}
for CFL}. To complete our proof, it is therefore enough to show that $Equal_6$ falls into $\bpcfl$.

Firstly, we set $N=5$ and consider the following probabilistic procedure for ppda's.  Let $w$ be any input and define $\alpha_i = \#_{a_i}(w)$ for each index $i\in[6]$.
In the case where all $\alpha_i$'s are at most $N$, we deterministically decide whether $w$ is in $Equal_6$ without using any stack.
For simplicity, we consider only the case where $\alpha_i>N$ for all $i\in[6]$. We randomly pick up two numbers $x$ and $y$ from $[N]$. We scan $w$ from left to right. Whenever we scan $a_1$ (resp., $a_2$ and $a_3$),  we push down $1$ (resp., $1^{x}$ and $1^{y}$) into a stack. On the contrary, when we scan $a_4$ (resp., $a_5$ and $a_6$), we pop $1$ (resp., $1^{x}$ and $1^{y}$) from the stack. During a series of such  pop-ups, if the stack becomes empty (except for the bottom marker $Z_0$), then we try to push  a special symbol ``$-1$'' instead of popping up $1$'s in order to indicate that there is a deficit in the stack content.
To implement this idea, if we further push down $1$'s, then we actually pops up the same number of $-1$'s.  After reading the entire $w$,  there is
$\ell = | (\alpha_1-\alpha_4) + x(\alpha_2-\alpha_5) + y(\alpha_3-\alpha_6) |$ symbols in the stack. When the stack becomes empty, we then accept the input; otherwise, we reject it.

Formally, the transition function $\delta$ is described as follows. Let $\delta(q_0,\cent,Z_0)=\{(p^{(x,y)},Z_0)\mid x,y\in[N]\}$. For simplicity, write $p$ for $p^{(x,y)}$. We give necessary transitions only for the cases of $a_1$,  $a_2$, and $a_6$. The other cases are similar in essence. For simplicity, assume that
$\tau\in\{1,Z_0\}$. Let $\delta(p,a_1,\tau) =\{(p,1\tau)\}$ and $\delta(p,a_1,-1) =\{(p,\lambda)\}$.
When we read $a_2$, let $\delta(p,a_2,\tau) = \{(p,1^x\tau)\}$. Define $q_{2,1}$ to be $p$.
Let $\delta(p,a_2,-1)=\{(q_{2,x-1},\lambda)\}$,   $\delta(q_{2,m},\lambda,\tau)=\{(q_{2,m-1},1\tau)\}$,  and
$\delta(q_{2,m},\lambda,-1)=\{(q_{2,m-1},\lambda)\}$
for $m\in[2,x-1]_{\integer}$.
In the case of $a_6$, let $\xi\in\{-1,Z_0\}$. Let $\delta(p,a_6,\xi) = \{(p,(-1)^y\xi)\}$ and define $q_{6,1}$ to be $p$. Let $\delta(p,a_6,1)=\{(q_{6,y-1},\lambda)\}$,   $\delta(q_{6,m},\lambda,1)=\{(q_{6,m-1},\lambda)\}$, and $\delta(q_{6,m},\lambda,-1)=\{(q_{6,m-1},(-1)(-1))\}$ for $m\in[2,y-1]_{\integer}$.

If $w$ is in $Equal_6$, then we obtain $\ell=0$ for any choice of $x,y\in[N]$ since $\alpha_i=\alpha_j$ for any $i,j\in[N]$.  Conversely, we assume that $w\not\in Equal_6$ and we will later argue that the error probability $\varepsilon$ (\ie the probability of obtaining $\ell=0$) is at most $1/3$.  This clearly places $Equal_6$ in $\bpcfl$.

Let us assume that $w\not\in Equal_6$ and $\ell=0$. For any two pairs $(x_1,y_1),(x_2,y_2)\in[N]\times[N]$ that force $\ell$ to be zero, we derive $(\alpha_1-\alpha_4)+x_i(\alpha_2-\alpha_5)+y_i(\alpha_3-\alpha_6)=0$ for any index $i\in\{1,2\}$.  {}From these two equations, it follows that (*)  $(x_1-x_2)(\alpha_2-\alpha_5) = (y_2-y_1)(\alpha_3-\alpha_6)$.

(1) Consider the case where $\alpha_2=\alpha_5$ but $\alpha_3\neq \alpha_6$.  By (*), we conclude that $y_1=y_2$; that is, there is a unique solution $y$ for the equation $\ell=0$.
Hence, the total number of pairs $(x,y)$ that force $\ell$ to be zero is  exactly $N$, and thus $\varepsilon$ must equal $1/N$, which is clearly smaller than $1/3$.  The case where $\alpha_2\neq \alpha_5$ and $\alpha_3=\alpha_6$ is similar.

(2) Consider the case where $\alpha_2\neq \alpha_5$ and $\alpha_3\neq \alpha_6$.
There are two cases to consider separately.

(a) If $\alpha_2-\alpha_5$ and $\alpha_3-\alpha_6$ are relatively prime, then we conclude that $x_1=x_2$ and $y_1=y_2$ from (*). This indicates that there is a unique solution pair $(x,y)$ for the equation $\ell=0$.  Thus, the error probability $\varepsilon$ is $1/N^2$, implying $\varepsilon<1/3$.

(b) If $\alpha_2-\alpha_5 = \beta(\alpha_3-\alpha_6)$ holds for a certain non-zero integer $\beta$, then we obtain $x_1-x_2=\beta(y_2-y_1)$. Since $|x_1-x_2|,|y_2-y_1|\leq 4$, there are at most $12$ cases for tuples  $(x_1,y_1,x_2,y_2)$ with $x_1>x_2$ that satisfy  $x_1-x_2=\beta'(y_2-y_1)$ for a certain non-zero integer $\beta'$. Hence, $\varepsilon$ is at most $12/25$, which is obviously smaller than $1/3$.

(3) The case where $\alpha_1\neq \alpha_4$, $\alpha_2= \alpha_5$, and $\alpha_3= \alpha_6$ never occurs because of $\ell=0$.
\end{proof}

It is not clear whether $\bpcfl$ is located inside the CFL hierarchy, because a standard argument used to prove the containment $\bpp\subseteq\sigmap{2}\cap\pip{2}$ requires an \emph{amplification property} but a ppda  cannot, in general, amplify its success probability \cite{HS10}.

Since a relationship between $\bpcfl$ and $\sigmacflt{2}\cap\picflt{2}$ is not known at this moment, we resort to an oracle separation between those two language families. For this purpose, we want to introduce
a {\em relativization} of the target language families.
Firstly, we introduce a relativization of $\bpcfl$.
Similar to $\cfl_{m}^A$, $\bpcfl_{m}^A$ is  defined simply by providing underlying oracle ppda's with extra query tapes.  Secondly, we define a relativized CFL hierarchy.
Relative to oracle $A$, a {\em relativized CFL hierarchy} $\{\reldeltacflt{k}{A}, \relsigmacflt{k}{A},\relpicflt{k}{A}\mid k\in\nat^{+} \}$ consists of the following language families:
$\reldeltacflt{1}{A}= \dcfl_{T}^A$, $\relsigmacflt{1}{A} = \cfl_{T}^{A}$, $\relpicflt{k}{A} = \co\relsigmacfl{k}{A}$, $\reldeltacflt{k+1}{A} = \dcfl_{T}(\relsigmacflt{k}{A})$,  and $\relsigmacflt{k+1}{A} = \cfl_{T}(\relpicflt{k}{A})$ for all indices $k\geq1$.
As this definition hints, any oracle-dependent language $L^A$ in $\relsigmacflt{k}{A}$ can be recognized by a chain of $k$ $T$-reduction npda's whose last machine makes queries directly to $A$.   As for later reference, we refer to this chain of $k$ machines as  a {\em defining machine set} for $L^A$.

\begin{theorem}\label{BPCFL-oracle-sep}
There exists a recursive oracle $A$ such that $\bpcfl_{m}^{A}\nsubseteq\relsigmacflt{2}{A}$.
\end{theorem}

For the proof of Theorem \ref{BPCFL-oracle-sep}, we need to consider certain non-uniform families of levelable Boolean circuits, where a circuit is {\em levelable}  if (i) all nodes are partitioned into levels, (ii) edges exist only between adjacent levels, (iii) the first level of a circuit has the output node, (iv) the gates at the same level are of the same type, (v) two gates at two adjacent levels are of the same type, and (vi) all input nodes (labeled by literals) are at the same level (see, \eg \cite{DK00}).  The notation  $CIR_{k}(n,m)$ denotes the collection of all levelable Boolean circuits $C$ that satisfy the following conditions:
(1) $C$ has {\em depth $k$} (\ie $k$ levels of gates), (2)  the {\em top gate} (\ie the root node) of $C$ is $OR$, (3) $C$ has alternating levels of $OR$ and $AND$, (4) the {\em bottom fan-in} (\ie the maximum fan-in of any bottom gate) of $C$  is at most $m$, (5)  the fan-in of any gate except for the bottom gates is at most $n$, and (6) there are at most $n$ input variables.

The following lemma states how to translate any oracle-dependent language in $\relsigmacflt{k}{A}$ into a family of Boolean circuits in $CIR_{k+1}(2^{O(n)},O(k))$.

\begin{lemma}\label{depth-circuit}
Let $L^A$ be any oracle-dependent language over alphabet $\Sigma$ in $\relsigmacflt{k}{A}$, where $A$ is any oracle over alphabet $\Theta$.  Let $(M_1,M_2,\ldots,M_k)$ be a defining machine set for $L^A$.  Let $a$ and $c$ be two positive constants such that the running time of each $M_i$ is bounded from above by $c|x|$ and  any query word $y$ produced by $M_k$
has length at most $a|x|$, where $x$ is an input string to $L^A$.   Let $\hat{\Theta}=\Theta\cup\{0,1,\natural\}$. For every length $n\in\nat^{+}$ and every input $x\in\Sigma^n$, there exist a Boolean circuit $C$ in $CIR_{k+1}(\|\\hat{Theta}\|^{an},cn)$ such that (1)  all variables $x_1,x_2,\ldots,x_{\ell(n)}$ of $C$ are strings (and their  negations) included in  $\hat{\Theta}^{\leq an}$  and  (2) for any oracle $A$, it holds that  $x$ is in $L^A$ if and only if $C$ outputs $1$ on inputs $(\chi^A(x_1),\chi^A(x_2),\ldots,\chi^A(x_{\ell(n)}))$.
\end{lemma}

\begin{proof}
We will prove this lemma by induction on $k\geq1$.   We begin with the base case $k=1$. Let $L^A$ be any oracle-dependent language in $\cfl_{T}^{A}$, where $A$ is an oracle over alphabet $\Theta$. Let $M, a,c$ satisfy  the premise of the lemma. Fix $n$ and $x\in\Sigma^n$ arbitrarily.
By an argument similar to the proof of Proposition \ref{first-level-equal}, we can modify $M$ so that, before starting writing the $i$th query word $y_i\in\{0,1\}$, it must guess its oracle answer $b_i$ and produce $b_iy_i\natural$ on a query tape and, instead of making an actual query, it assumes that $A$ returns $b_i$. After this modification, a string produced on a query tape must be of the form $b_1y_1\natural b_2y_2\natural \cdots \natural b_{\ell}y_{\ell}\natural$. Let $V_x$ be composed of all such query strings produced along accepting computation paths. Note that $\parallel\! V_x \!\parallel \leq \|\hat{\Theta}\|^{an}$ since $M$ halts within time $an$, where $\hat{\Theta}=\Theta\cup\{0,1,\natural\}$.

Next, we will define a circuit $C$, which is an OR of ANDs, as follows. The top OR gate has edges labeled by strings in $V_x$. For each $y\in V_x$, an associated subcircuit $D_y$, consisting of an AND gate, has input nodes labeled by literals of the form $y_1^{(b_1)},y_2^{(b_2)},\ldots,y_{\ell}^{(b_{\ell})}$, where $y_i^{(0)} =y_i$ and $y_i^{(1)}=\overline{y_i}$. Let $x_1,x_2,\ldots,x_{m(n)}$ be all distinct variables (of the positive form) appearing in $C$.  It is not difficult to verify that, for any oracle $A$,  $x\in L^A$ if and only if $C(\chi^{A}(x_1),\chi^{A}(x_2),\ldots,\chi^{A}(x_{m(n)}))=1$.

Let us consider induction step $k\geq2$. Assume that $L^A$ is in $\cfl_{T}(B^A)$ for a certain oracle $B^A\in\relpicflt{k-1}{A}$. Let $M$ be a $T$-reduction npda reducing $L^A$ to $B^A$. By a similar argument as in the base case, we can modify $M$ so that it generates query words of the form $b_1y_1\natural\cdots \natural b_{\ell}y_{\ell}\natural$ without making actual queries. We set $V_x$ to be the collection of all such strings produced along accepting computation paths. Since $\overline{B^A}\in\relsigmacflt{k-1}{A}$, by our induction hypothesis, for each string $y\in V_x$, there exists a circuit $D_y$ satisfying the lemma. Instead of $D_y$, we consider its {\em dual circuit} $\overline{D_y}$.  Here, we define $C$ to be an OR of all $\overline{D_y}$'s for any $y\in V_x$. A similar reasoning as in the base case shows that, for any oracle $A$,  $x$ is in $L^A$ if and only if $C(\chi^{A}(x_1),\chi^{A}(x_2),\ldots,\chi^{A}(x_{m(n)}))=1$.
\end{proof}

Let us prove Theorem \ref{BPCFL-oracle-sep} using
Lemma \ref{depth-circuit}.
For the desired separation between $\bpcfl^{A}$ and $\relsigmacfl{k}{A}\cap\relpicfl{k}{A}$, we use the following example language:
$L^A =\{0^n\mid \, \parallel\! A\cap\Sigma^n \!\parallel >2^{n-1}\}$ over a binary alphabet $\Sigma=\{0,1\}$.  To guarantee that $L^A\in\bpcfl_{m}^A$, we will aim at constructing $A$ that satisfies
either  $\parallel\! A\cap\Sigma^n \!\parallel \leq 2^n/3$ or $\parallel\! \overline{A}\cap\Sigma^n \!\parallel \leq 2^n/3$ for every length $n\in\nat^{+}$. This will be done by choosing recursively a pair of $T$-reduction npda's that witnesses a language $B^A$ in $\relsigmacflt{2}{A}$ and by defining a sufficiently large number $n\in\nat$ and a set $A_n$ ($=A\cap\Sigma^n$) such that $0^n\in L^A$ $\not\!\leftrightarrow$ $0^n\in B^A$.

\vs{-2}
\begin{proofof}{Theorem \ref{BPCFL-oracle-sep}}
Let $\Sigma=\{0,1\}$ and consider the aforementioned oracle-dependent  language $L^A =\{0^n\mid \parallel\! A\cap\Sigma^n \!\parallel >2^{n-1}\}$. To guarantee that $L^A\in\bpcfl_{m}^A$, we consider only oracles $A$ that satisfy the following condition: (*) either  $\parallel\! A\cap\Sigma^n \!\parallel \leq 2^n/3$ or $\parallel\! \overline{A}\cap\Sigma^n \!\parallel \leq 2^n/3$ for every length $n\in\nat^{+}$.

Next, let us construct an appropriate oracle $A$ satisfying that $L^A\not\in\relsigmacflt{2}{A}$. For this purpose, we use Lemma \ref{depth-circuit}. First, we enumerate all  oracle-dependent languages in $\relsigmacflt{2}{A}$ and consider their corresponding depth-$3$ Boolean circuit families that satisfy all the conditions stated in Lemma \ref{depth-circuit}.

Recursively, we choose such a circuit family and define a large enough length $n$ and a set $A_n$ ($=A\cap\Sigma^n$). Initially, we set $n_0=0$ and $A_0=\setempty$. Assume that, at Stage $i-1$,  we have already defined $n_{i-1}$ and $A_{i-1}$. Let us consider Stage $i$. Take the $i$th circuit family $\{C_n\}_{n\in\nat}$ and two constants $a,c>0$ given by Lemma \ref{depth-circuit} so that $C_n$ belongs to $CIR_{3}(2^{an},cn)$. First, we set $n_i = \max\{n_{i-1}+1,2^{a'n_{i-1}}+1,c'+1\}$, where $a'$ and $c'$ are constants taken at Stage $i-1$. The choice of $n_i$ guarantees that $A_{n_i}$ is not affected by the behaviors of  the circuits considered at Stage $i-1$.

In the rest of the proof,  we will examine two cases.

(1) Consider the base case where the bottom fan-in is exactly $1$.
For each label $y\in\{0,1\}^{an}$, let $Q(y)$ be the set of all input variables that appear in subcircuits connected to the top OR gate by a wire labeled $y$. In particular, $Q^{+}(y)$ (resp., $Q^{-}(y)$) consists of variables in $Q(y)$ that appear in positive form (resp., negative form).
 Let us consider two cases.

(a) Assume that there exists a string $y_0$ such that $\parallel\! Q^{+}(y_0) \!\parallel \leq 2^n/3$. In this case, we set $A_n$ to be $Q^{+}(y_0)$. It is obvious that $0^n\not\in L^{A}$ and $C_n(\chi^A(x_1),\ldots,\chi^A(x_{2^{an}}))=1$.

(b) Assume that, for all $y$, $\parallel\! Q^{+}(y) \!\parallel >2^{an}/3$. Recursively, we will choose at most $an/\log(3/2)+1$ strings. At the first step, let $B_0=\{0,1\}^{an}$. Assume that $B_{i-1}$ has been defined. We will define $B_{i}$ as follows. Choose the lexicographically smallest string $w$ for which that the set $\{y\in B_{i-1}\mid w\in Q^{+}(y)\}$  has the largest cardinality. Finally, we define $w_i$ to be this string $w$ and we set $B_i = \{y\in B_{i-1}\mid w_i\not\in Q^{+}(y)\}$.
In what follows, we show that $\parallel\! B_i \!\parallel \leq (2/3) \parallel\! B_{i-1} \!\parallel$.

\begin{claim}
$\parallel\! B_i \!\parallel \leq (2/3) \parallel\! B_{i-1} \!\parallel$.
\end{claim}

\begin{proof}
Let $d$ satisfy $\parallel\! \overline{B_i}  \!\parallel = d \parallel\! B_{i-1} \!\parallel$. For each index $i\in[2^{an}]$, let $X_i =\{y\mid x_i\in Q^{+}(y)\}$. Since $\parallel\! Q^{+}(y) \!\parallel >2^{an}/3$ for all $y$'s, it holds that $\sum_{i=1}^{2^{an}} d \parallel\! X_i \!\parallel  \geq (2^{an}/3) \parallel\! B_{i-1} \!\parallel$. Note that $\sum_{i=1}^{2^{an}} \parallel\! X_i \!\parallel  = 2^{an}$. Thus, we obtain $d\geq 1/3$. Since $\parallel\! B_i \!\parallel  = \parallel\! B_{i-1} \!\parallel - \parallel\! \overline{B_i} \!\parallel$, it follows that $\parallel\! B_i \!\parallel \leq (2/3) \parallel\! B_{i-1} \!\parallel$.
\end{proof}

{}From the above claim, it follows that $\parallel\! B_i \!\parallel \leq (2/3)^i \parallel\! B_0 \!\parallel  = (2/3)^i2^{an}$.  Let $i_0$ denote the minimal number such that $\parallel\! B_i \!\parallel  =0$.  Since $i> an/\log(3/2)$ implies $\parallel\! B_i \!\parallel <1$, we conclude that $i_0\leq an/\log(3/2)+1$. Now, we write $W$ for the collection of all $w_i$'s ($1\leq i \leq i_0$) defined in the above procedure.  The desired $A_n$ is defined to be  $(\Sigma^{an}-W)\cup (\bigcup_{y}Q^{-}(y))$.

(2) Second, we will consider the case where the bottom fan-in is  more than $1$. To handle this case, we will use a special form of the so-called {\em switching lemma} to reduce this case to the base case.

A {\em restriction}   is a map $\rho$ from a set of $n$ Boolean variables  to $\{0,1,*\}$.  We define $\RR^{\ell,q}_{n}$ to be the collection of {\em restrictions} $\rho$ on a domain of $n$ variables that have exactly $\ell$ unset variables and a $q$-fraction of the variables are set to be $1$.
 For any circuit $C$, $bf(C)$ denotes the bottom fan-in of $C$.

\begin{claim}\label{switching-lemma}{\rm \cite{Bea90}}
Let $C$ be a circuit of OR of ANDs with bottom fan-in at most $r$. Let $n>0$, $s\geq0$, $\ell=pn$, and $p\leq 1/7$. It holds that $\parallel\! \{\rho\in \RR^{\ell,q}_{n}\mid  \exists D:\,\text{AND of ORs}\, [\, bf(D)\geq s\,] \} \!\parallel < (2pr/q^2)^s \parallel\! \RR^{\ell,q}_{n} \!\parallel$.
\end{claim}

Consider any subcircuit $D$, an AND of ORs, attached to the top OR-gate. By setting $q=1/3$ and $r=c$, we apply Claim \ref{switching-lemma} to $D$. The probability that $D$ is written as an OR of ANDs with bottom fan-in at most $an$ is upper-bounded by $1-(18pc)^{an}$.  Moreover, the probability that all such subcircuits $D$ are simultaneously written as circuits, each of which is an AND of ORs, is at most $[1-(18pc)^{an}]^{2^{an}}\geq 1-2^{an}(18pc)^{an}=1-(36pc)^{an}$. If we choose $p=1/72c$, then the success probability is at least $1-(36pc)^{an} \geq 1-(1/2)^{an}$, which is larger than $1/2$ for any integer $n\geq 2/a$. Since every subcircuit $D$ is written as an AND of ORs, the original circuit $C$ can be written as an OR of ANDs with bottom fan-in at most $an$. Finally, we apply the base case to this new circuit.
\end{proofof}

There also exists an obvious oracle for which $\bpcfl$ equals $\sigmacflt{2}$ since the following equalities hold.

\begin{proposition}
$\bpcfl_{T}^{\pspace} = \relsigmacflt{2}{\pspace} = \pspace$.
\end{proposition}

\begin{proof}
It is obvious that $\bpcfl_{T}^{B}\subseteq\pspace_{T}^{B}$ for every oracle $B$, where $\pspace_{T}^{B}$ is a Turing relativization of $\pspace$ relative to $B$.  Hence, it follows that $\bpcfl_{T}^{\pspace} \subseteq \pspace_{T}^{\pspace} = \pspace$. Conversely, since $A\subseteq \bpcfl_{T}^{A}$ for any oracle $A$, in particular, we obtain $\pspace \subseteq \bpcfl_{T}^{\pspace}$. The case of $\relsigmacflt{2}{\pspace}$ is similar.
\end{proof}

We have just seen an oracle that supports the containment $\bpcfl\subseteq\sigmacflt{2}$ and another oracle that does not. As this example showcases, some relativization results are  quite  counterintuitive. Before closing this subsection, we will present another plausible example regarding the {\em parity NFA language family} $\paritynfa$ whose elements are languages of the form $\{x\mid \parallel\! ACC_{M}(x) \!\parallel =1\;(\mathrm{mod}\;2)\}$ for arbitrary nfa's $M$.  In  the  unrelativized world,  it is known that $\paritynfa \subseteq \tc{1}\subseteq \ph$; however, there exists an oracle that defies this fact. As we have done for $\nfa_{m}^{A}$, we define a many-one relativization $\paritynfa_{M}^{A}$.

\begin{lemma}\label{parity-nfa-vs-ph}
There exists an oracle $A$ such that $\paritynfa_{m}^{A}\nsubseteq \ph^{A}$.
\end{lemma}

\begin{proof}
Let us consider a special language $L^A=\{0^n\mid \bigoplus_{x\in\Sigma^n}\chi^A(x)=1\;(\mathrm{mod}\,2)\}$ relative to oracle $A$. It is easy to show that, for any oracle $A$,  $L^A$ is in $\paritynfa_{m}^{A}$ by guessing a string $x$ in $\Sigma^n$ and asking $A$ to decide that $x\in A$. Since it is shown in \cite{Cai86} that $L^A\not\in\ph^A$ for a {\em random oracle} $A$, we immediately obtain the desired oracle separation.
\end{proof}

\section{A Close Relation to the Polynomial Hierarchy}\label{sec:close-relation}

In Section \ref{sec:CFL-hierarchy}, the $\cfl$ hierarchy has proven to be viable in classifying certain languages and it has three natural characterizations, as shown in Theorem \ref{Sigma-m-vs-T} and Proposition \ref{logical-charact}.
Moreover, we know that the first two levels of the $\cfl$ hierarchy are different (namely, $\sigmacflt{1}\neq\sigmacflt{2}$ stated in Proposition \ref{first-second-gap}); however, the separation of the rest of the hierarchy still remains unknown at this moment. In this section, we will discuss under what conditions the separation is possible.

\subsection{Logarithmic-Space Many-One Reductions}

Recall the space-bounded complexity class $\dl$. We hereafter consider its natural many-one relativization $\dl_{m}^{A}$ relative to oracle $A$.
Given a language $A$, a language $L$ is in $\dl_{m}^{A}$ if there exists a logarithmic-space (or log-space) oracle DTM $M$ with an extra write-only query tape (other than a two-way read-only input tape and a two-way read/write work tape) such that, for every string $x$, $x$ is in $L$ if and only if $M$ on the input $x$ uniquely produces a certain string in $A$.  Recall that any tape head on a write-only tape moves only in one direction.
More importantly, we explicitly demand that all computation paths of $M$ on any  input of length $n$ terminate within $n^{O(1)}$ steps.
As a consequence, any query word produced by $M$ must have length $n^{O(1)}$ as well. For any language family $\CC$, we denote by $\dl_{m}^{\CC}$ the union $\bigcup_{A\in\CC}\dl_{m}^{A}$. Occasionally, we also write $\dl_{m}(\CC)$ to mean $\dl_{m}^{\CC}$.  In particular, when $\CC=\cfl$, the language family $\dl_{m}^{\cfl}$ has been known as $\mathrm{LOGCFL}$ ($\mathrm{LogCFL}$ or $\mathrm{LOG(CFL)}$) in the literature.  Since $\sigmacflt{1}=\cfl$, it follows from \cite{Ven91} that $\dl_{m}(\sigmacflt{1}) = \sac{1}$.  As for the language family $\pcfl$,  for instance, Macarie and Ogihara \cite{MO98} demonstrated that $\dl_{m}^{\pcfl} \subseteq \tc{1}$.
Concerning $\bpcfl$,  the  containment
$\dl_{m}^{\bpcfl}\subseteq \dl_{m}(\sigmacflt{2})$ holds. This fact can be proven as follows. {}From $\bpcfl\subseteq\pcfl$  by their definitions, it follows  that $\dl_{m}^{\bpcfl}\subseteq \dl_{m}^{\pcfl} \subseteq \tc{1}$. However, by Claim \ref{CFLH-charact-PH}, $\dl_{m}(\sigmacfl{2})= \np$. Thus,  we conclude that $\dl_{m}^{\bpcfl}\subseteq \tc{1}\subseteq \np =\dl_{m}(\sigmacfl{2})$.

It is obvious that $\CC_1=\CC_2$ implies $\dl_{m}^{\CC_1} = \dl_{m}^{\CC_2}$; however, the converse does not always hold. Here is a simple example. Although $\cfl(k)\neq \cfl$  holds for  $k\geq2$, the following equalities  hold.

\begin{lemma}\label{CFL(k)-equal-CFL}
$\dl_{m}^{\cfl(\omega)} = \dl_{m}^{\cfl} = \sac{1}$.
\end{lemma}

\begin{proof}
Recall that $\dl_{m}^{\cfl}=\sac{1}$ \cite{Ven91}. Therefore, our goal is now set to prove that  $\dl_{m}^{\cfl(k)} = \dl_{m}^{\cfl}$  for every index $k\geq2$.
Note that, for all indices  $k\in\nat^{+}$, $\dl_{m}^{\cfl} = \dl_{m}^{\cfl(1)} \subseteq \dl_{m}^{\cfl(k)}$ because $\cfl(1)\subseteq\cfl(k)$.  The remaining task is to show that $\dl_{m}^{\cfl(k)}\subseteq \dl_{m}^{\cfl}$. Let $k\geq2$ and assume that $A\in \dl_{m}^{B}$ for a certain language $B\in\cfl(k)$. There are $k$ languages $B_1,B_2,\ldots,B_k\in\cfl$ satisfying that $B = \bigcap_{i\in[k]}B_i$. Take any log-space oracle DTM $M$ that $m$-reduces  $A$ to $B$.

We define a new oracle DTM $N_1$ so that, on input $x$, it produces  $y\natural y\natural \cdots \natural y$ ($k$ $y$'s) on its query tape if $M$ taking $x$ produces $y$. Since $k$ is a constant, $N_1$ needs only $O(\log{n})$ space to execute.  Our new oracle $C$ is the set $\{y_1\natural y_2\natural\cdots \natural y_k\mid \forall i\in[k]\,[ y_i\in B_i]\}$. It is clear that, for any $x$, $N_1(x)\in C$ if and only if $M(x)\in B$. From this equivalence, $N_1$ $m$-reduces $A$ to $C$; thus, $A\in\dl_{m}^{C}$.
It remains to show that $C$ is in $\cfl$.
For each $i\in[k]$, let $M_i$ denote an npda that recognizes $B_i$. Here, let us consider a new npda $N_2$ that behaves as follows. On input $w$, $N_2$ checks if $w$ is of the form $y_1\natural y_2\natural\cdots \natural y_k$.
At the same time, $N_2$ sequentially simulates $M_i$ on input $y_i$, starting with $i=1$. After each simulation of $M_i$ on input $y_i$, $N_2$ always empties its own stack so that each simulation does not affect the next one. Moreover, as soon as $M_i$ rejects $y_i$, $N_2$ enters a rejecting state and halts. It is obvious that $C$ is recognized by the npda $N_2$. In conclusion, $C$ is indeed a context-free language.
\end{proof}


The $\cfl$ hierarchy turns out to be a quite useful tool because it is closely related to the polynomial hierarchy $\{\deltap{k},\sigmap{k},\pip{k}\mid k\in\nat\}$. Reinhardt \cite{Rei90} first established a close  connection between his alternating hierarchy over $\cfl$ and the polynomial hierarchy.
Similar to $\sigmacflt{k,e}$, the notation $\sigmap{k,e}$ stands for the $e$th level of the \emph{Boolean hierarchies over $\sigmap{k}$}.
We want to demonstrate the following intimate relationship between $\sigmacflt{k+1,e}$ and $\sigmap{k,e}$.

\begin{theorem}\label{Sigmap-k-Boolean-eq}
For every index $e,k\in\nat^{+}$, $\dl_{m}(\sigmacflt{k+1,e}) = \sigmap{k,e}$ holds. In particular, $\dl_{m}(\sigmacfl{k+1}) = \sigmap{k}$ holds.
\end{theorem}

\begin{proof}
Fixing $k$ arbitrarily, we will show the theorem by induction on $e\in\nat^{+}$.  Our starting point is  the base case of $e=1$.  Notice that $\sigmacfl{k+1,1}=\sigmacfl{k}$ and $\sigmap{k,1}=\sigmap{k}$.

\begin{claim}\label{CFLH-charact-PH}
$\dl_{m}(\sigmacflt{k+1}) = \sigmap{k}$ holds for every index $k\in\nat^{+}$.
\end{claim}

\begin{proof}
In what follows, we will demonstrate separately that, for every  index $k\in\nat^{+}$, (1) $\dl_{m}(\sigmacflt{k+1})\subseteq \sigmap{k}$  and (2) $\sigmap{k} \subseteq \dl_{m}(\sigmacflt{k+1})$.

(1)  To prove that $\dl_{m}(\sigmacflt{k+1})\subseteq \sigmap{k}$, we start with the following useful  relationship between $\sigmacflt{k+1}$ and $\sigmap{k}$.

\begin{claim}\label{sigma_CFL-to-sigma_P}
$\sigmacflt{k+1}\subseteq \sigmap{k}$ holds for every index $k\in\nat^{+}$.
\end{claim}

\begin{proof}
This claim is proven by induction on $k\geq1$. A key to the following proof is the fact that  $\cfl_{T}^{A}\subseteq\np^{A}$ holds for every oracle $A$. When $k=1$, it holds that $\sigmacflt{2} = \cfl_{T}^{\cfl} \subseteq \np^{\cfl}$. Since $\cfl\subseteq \p$, we obtain $\np^{\cfl}\subseteq \np^{\p} =\np$, yielding the desired containment $\sigmacflt{2}\subseteq\np$.
When  $k\geq2$,  we assume by induction hypothesis that $\sigmacflt{k}\subseteq \sigmap{k-1}$. It therefore  follows that
$\sigmacflt{k+1}=\cfl_{T}(\picflt{k}) \subseteq \np(\picflt{k}) \subseteq \np(\pip{k-1}) = \sigmap{k}$.
\end{proof}

The containment  $\dl_{m}(\sigmacflt{k+1})\subseteq \dl_{m}(\sigmap{k})$ follows immediately from Claim \ref{sigma_CFL-to-sigma_P}. Hence, using the fact that $\dl_{m}(\sigmap{k})\subseteq \sigmap{k}$, we  conclude that $\dl_{m}(\sigmacflt{k+1})\subseteq \sigmap{k}$.

(2) Next, we plan to show that $\sigmap{k} \subseteq \dl_{m}(\sigmacflt{k+1})$.
An underlying idea of the following argument comes from \cite{Rei90}.
Our plan is to define a set of $k$ quantified Boolean formulas, denoted  $QBF_k$, which is slightly different from a standard one, and to prove that (a) $QBF_k$ is log-space complete for $\sigmap{k}$ and (b)  $QBF_k$ indeed belongs to $\sigmacflt{k+1}$.  Combining (a) and (b) implies that $\sigmap{k}\subseteq \dl_{m}^{QBF_k}\subseteq \dl_{m}(\sigmacflt{k+1})$.

Here, we will discuss the case where $k$ is odd.
The  language $QBF_k$ must be of a specific form so that an input-tape head of an oracle npda can read through a given instance of $QBF_k$ from left to right without back-tracking. First, we prepare the following alphabet  of distinct input symbols: $\Sigma_k= \{\exists,\forall,\wedge,\natural,+,-,0,a_1,a_2,\ldots,a_k\}$.
A \emph{string} $\phi$ over $\Sigma_{k}$ belongs to  $QBF_k$ exactly when $\phi$ is of the form
$\exists a_1^{m_1}\forall a_2^{m_2}\cdots Q_k  a_k^{m_k}\natural c_1\wedge c_2\wedge \cdots \wedge c_m$, which satisfies the following conditions:   each $m_i$ and $m$ are in $\nat^{+}$, $Q_k$ is $\exists$, each $c_i$ is a string $c_{i,\ell_{i}}c_{i,\ell_{i}-1}\cdots c_{i,2} c_{i,1}$ in $\{+,-,0\}^{\ell_{i}}$ for a certain number $\ell_i$ satisfying $\ell_i\geq \overline{m}=\sum_{j=1}^{k}m_j$, and, moreover,  the corresponding quantified Boolean formula
\[
\tilde{\phi}\equiv \exists x_{1},\ldots,x_{m_1}\forall x_{m_1+1},\ldots,x_{m_1+m_2}\cdots Q_k x_{m'+1},\ldots,x_{m'+m_k}[C_1\wedge C_2\wedge \cdots \wedge C_m]
\]
is \emph{satisfiable}, where $m'=\overline{m}-m_k$, each $C_i$ is a propositional formula of the form $(\bigvee_{j\in S_{+}(i)}x_j) \vee (\bigvee_{j\in S_{-}(i)}\overline{x_j})$ for $S_{+}(i) = \{j\in[\ell]\mid j\leq \overline{m},c_{i,\ell_{i}-j+1}=+\}$, and $S_{-}(i) = \{j\in[\ell]\mid j\leq \overline{m},c_{i,\ell_{i}-j+1}=-\}$.

When $k$ is even, we define $BQF_k$ by exchanging the roles of $\wedge$ and $\vee$ and by setting $Q_k=\forall$ in the above definition.
As is shown in \cite{Sto77}, it is possible to demonstrate that $QBF_{k}$ is log-space many-one complete for $\sigmap{k}$; that is, every language in $\sigmap{k}$ belongs to $\dl_{m}^{QBF_k}$.  If $QBF_k$ is in $\sigmacflt{k+1}$, then we obtain $\sigmap{k}\subseteq \dl_{m}^{QBF_k}\subseteq \dl_{m}(\sigmacflt{k+1})$, as requested.

Therefore, what remains undone is to prove that $QBF_k$ is indeed in $\sigmacflt{k+1}$ by constructing a nice chain of $m$-reduction npda's for $QBF_k$. As have done in Section \ref{sec:BPCFL}, we also call such a chain of $m$-reduction npda's computing $QBF_k$ by a {\em defining machine set} for $QBF_k$.
Given an index $i\in[k]$, we define  a string $\phi_i$ to be  $Q_{i} a_{i}^{m_{i}}\cdots Q_k  a_k^{m_k}\natural c_1\wedge c_2\wedge \cdots \wedge c_m$, seen as a series of symbols.
Initially, let $\psi_1$ denote $\phi_1$. For any index $i\in[k]$, the $i+1$st $m$-reduction npda $M_{i+1}$ works as follows. Assume that an input string  $\psi_i$  to the machine has  the form $\track{a_1^{m_1}}{s_1}\cdots\track{a_{i-1}^{m_{i-1}}}{s_{i-1}}\phi_{i}$. While reading the first $m_i+1$ symbols, $Q_i a_i^{m_i}$, until the next symbol $Q_{i+1}$,  $M_{i+1}$ generates all strings  $s_{i}=s_{i1}s_{i2}\cdots s_{im_i}$ in $\{0,1\}^{m_i}$ and then  produces corresponding strings  $\track{a_1^{m_1}}{s_1}\cdots\track{a_{i}^{m_{i}}}{s_{i}}\phi_{i+1}$ on its  query tape.
Note that, at any moment, if  $M_{i+1}$ discovers that  the input does not have a valid form,  it immediately halts by entering an appropriate  rejecting state.
The last machine $M_k$ works in the following manner, provided that  $\psi_{k}$ is given as its input. Firstly, $M_k$ stores the string  $\track{a_1^{m_1}}{s_1}\cdots\track{a_{k-1}^{m_{k-1}}}{s_{k-1}}$ in its stack, guesses a binary string $s_k$ of length $m_k$, and stores $\track{a_k^{m_k}}{s_k}$ also in the stack.

Secondly, $M_k$ guesses $j$, locates the block of $c_j$,  and checks whether its corresponding Boolean formula $C_j$ is satisfied by the assignment specified by $s_1s_2\cdots s_k$. This checking process can be easily done by comparing the two strings $s_1s_2\cdots s_k$ and $c_j$, symbol by symbol,  as follows: for every $i\in[k]$, if $s_{ij}$ corresponds to $c_{i,\ell_i-j+1}$, then $M_k$ accepts the input exactly when $(s_{ij}=1\wedge c_{i,\ell_i-j+1}=+)$ or $(s_{ij}=0\wedge c_{i,\ell_i-j+1}=-)$.  The existence of a defining machine set for $QBF_k$ proves that $QBF_k$ indeed belongs to $\sigmacflt{k+1}$.
\end{proof}

We have just proven the case of $e=1$. Next, let us consider the case where $e\geq2$.  A key  to the proof of this case is the following simple fact, Claim \ref{input-redundancy}.  For convenience, we say that a language family $\CC$ {\em admits input redundancy} if, for every language $L$ in $\CC$, two languages $L'=\{x\natural y\mid x\in L\}$ and $L''=\{x\natural y\mid y\in L\}$ are both in $\CC$, provided that $\natural$ is a fresh symbol that never appears in $x$ as well as $y$.

\begin{claim}\label{input-redundancy}
If two language families $\CC_1$ and $\CC_2$ admit input redundancy, then
$\dl_{m}^{\CC_1}\wedge \dl_{m}^{\CC_2} \subseteq \dl_{m}^{\CC_1\wedge \CC_2}$ and $\dl_{m}^{\CC_1}\vee \dl_{m}^{\CC_2} \subseteq \dl_{m}^{\CC_1\vee \CC_2}$.
\end{claim}

\begin{proof}
We will show only the first assertion of the lemma, because the second assertion follows similarly.
Take any language $L$ and assume that $L\in \dl_{m}^{A_1}\wedge \dl_{m}^{A_2}$ for certain  two languages $A_1\in\CC_1$ and $A_2\in\CC_2$. There are two log-space $m$-reduction DTMs $M_1$ and $M_2$ such that, for  any index $i\in\{1,2\}$ and for every string $x$, $x$ is in $L$ if and only if $M_i^{A}(x)$ outputs $y_i$ and $y_i\in A_i$. Now, we want to define another machine $M$ as follows. On input $x$,  $M$ first simulates $M_1$ on $x$ and produces $y_1\natural$ on its query tape. Subsequently, $M$ simulates $M_2$ on $x$ and produces $y_2$ also on the query tape following the string $y_1\natural$.   The language $C=\{y_1\natural y_2\mid y_1\in A_1,y_2\in A_2\}$ clearly  belongs to $\CC_1\wedge \CC_2$. It is obvious that $L$ belongs to $\dl_{m}^{C}$, which is included in $\dl_{m}^{\CC_1\wedge\CC_2}$.
\end{proof}


From Claim \ref{input-redundancy}, it follows that $\dl_{m}(\sigmacflt{k+1,2e+1}) = \dl_{m}(\sigmacflt{k+1,2e}\vee \sigmacflt{k+1}) = \dl_{m}(\sigmacflt{k+1,2e}) \vee \dl_{m}(\sigmacflt{k+1})$. Since $\dl_{m}(\sigmacflt{k+1})=\sigmap{k}$ and $\dl_{m}(\sigmacflt{k+1,2e})=\sigmap{k,2e}$ by induction hypothesis, we obtain $\dl_{m}(\sigmacflt{k+1,2e+1}) = \sigmap{k,2e}\vee\sigmap{k} =\sigmap{k,2e+1}$.
Similarly, we conclude that   $\dl_{m}(\sigmacflt{k+1,2e+2})=\sigmap{k,2e+2}$.
\end{proof}

Unfortunately, the proof presented above does not apply to derive, for instance, $\dl_{m}(\sigmacflt{k+1}\cap\picflt{k+1}) = \sigmap{k}\cap\pip{k}$, simply because no many-one complete languages are known for $\sigmap{k}\cap\pip{k}$. This remains as a challenging question.

Theorem \ref{Sigmap-k-Boolean-eq} yields the following immediate consequence.

\begin{corollary}
If the polynomial hierarchy is infinite, then so is the Boolean hierarchy over $\sigmacflt{k}$ at every level $k\geq2$.
\end{corollary}

\begin{proof}
Fix $k\geq2$ arbitrarily.  Theorem \ref{Sigmap-k-Boolean-eq} implies that, if the Boolean hierarchy over $\sigmap{k}$ is infinite,
then the Boolean hierarchy over $\sigmacflt{k+1}$ is also infinite.
Earlier, Kadin \cite{Kad88} showed that, if the polynomial hierarchy is infinite, then the Boolean hierarchy over $\sigmap{i}$ is infinite for every level $i\geq1$. By combining those two statements, we instantly obtain the desired consequence.
\end{proof}

Theorem \ref{Sigmap-k-Boolean-eq}  also yields the following consequence.

\begin{corollary}\label{PH-CFLH-collapse}
If the polynomial hierarchy is infinite, then  so is the $\cfl$ hierarchy. More specifically, for every $k\geq2$, $\ph\neq \sigmap{k}$ implies  $\sigmacflt{k}\neq\sigmacflt{k+1}$.
\end{corollary}

This corollary does not exclude a chance that both  $\sigmacfl{k}\neq\sigmacfl{k+1}$ and $\ph=\sigmap{k}$ may occur.

In terms of the inclusion relationship, as shown in Claim \ref{sigma_CFL-to-sigma_P}, $\sigmacflt{k+1}$ is upper-bounded by $\sigmap{k}$. One may wonder if $\sigmacfl{k+2}$ is also included in $\sigmap{k}$. Nonetheless, it is possible to assert that this complexity bound is tight, under the assumption that $\deltap{k+1}$ is different from $\sigmap{k+1}$.

\begin{corollary}
For any index $k\geq2$, if $\deltap{k+1}\neq \sigmap{k+1}$, then $\sigmacflt{k+2}\nsubseteq \sigmap{k}$.
\end{corollary}

\begin{proof}
We want to show the contrapositive of this corollary.  We start with the assumption  that $\sigmacflt{k+2}\subseteq \sigmap{k}$.  {}From this inclusion,   $\dl_{m}(\sigmacflt{k+2})\subseteq \dl_{m}(\sigmap{k})$ follows.   Since  $\dl_{m}(\sigmacflt{k+2})=\sigmap{k+1}$ holds by Theorem \ref{Sigmap-k-Boolean-eq},  we derive $\sigmap{k+1}\subseteq \dl_{m}(\sigmap{k})$.  On the contrary, it holds that $\dl_{m}(\sigmap{k})\subseteq \p^{\sigmap{k}} = \deltap{k+1}$ since $\dl_{m}^{A}\subseteq \p^{A}$ for any $A$. As a result, $\sigmap{k+1}$ is included in $\deltap{k+1}$. Since $\deltap{k+1}\subseteq\sigmap{k+1}$ is obvious,   $\deltap{k+1}=\sigmap{k+1}$ follows immediately.
 \end{proof}

Let us recall that $\cfl(\omega)\subseteq \sigmacflt{2}\cap \picflt{2}$ by Proposition \ref{CFL(omega)-in-Pi2}. Lemma \ref{CFL(k)-equal-CFL} implies that $\cfl(\omega)\subseteq\dl_{m}^{\cfl(\omega)} \subseteq \sac{1} \subseteq \nc{2}$. Is it also true that $\sigmacfl{2}\subseteq\nc{2}$?  If $\sigmacflt{2}\subseteq \nc{2}$ holds, then Claim \ref{sigma_CFL-to-sigma_P} implies that $\np=\dl_{m}(\sigmacflt{2}) \subseteq \dl_{m}(\nc{2})=\nc{2}$, and we obtain $\np=\nc{2}$ since $\nc{2}\subseteq \np$ is obvious. In the end, we draw the following conclusion.

\begin{corollary}\label{NP-vs-NC2-second}
If $\np\neq \nc{2}$, then $\sigmacflt{2}\nsubseteq \nc{2}$.
\end{corollary}

\subsection{Logarithmic-Space Truth-Table Reductions}

We have already paid our attention to the computational complexity of languages that are log-space many-one reducible to certain languages in $\sigmacflt{k}$.  In the rest of this section, we will turn our attention to languages that are \emph{log-space truth-table reducible} to $\sigmacflt{k}$. Henceforth,  we will use the notation $\dl_{tt}^A$ (or $\dl_{tt}(A)$) to mean a family of all languages that are log-space truth-table reducible to oracle $A$.
It is not difficult to prove that a truth-table reduction can simulate Boolean operations that define each level of the Boolean hierarchy over $\cfl$. Hence,
the Boolean hierarchy $\bhcfl$ is ``equivalent'' to $\cfl$ under the log-space truth-table reducibility.

\begin{lemma}\label{DL_tt-CFL_k-vs-CFL}
$\dl_{tt}^{\bhcfl} = \dl_{tt}^{\cfl}$.
\end{lemma}

\begin{proof}
For this lemma, we need to show  the equality  $\dl_{tt}^{\cfl_{k}} = \dl_{tt}^{\cfl}$ for every index $k\geq1$.
Since $\cfl\subseteq\cfl_{k}$, it follows that $\dl_{tt}^{\cfl}\subseteq \dl_{tt}^{\cfl_{k}}$.  Conversely,  we will show that $\dl_{tt}^{\cfl_{k}}\subseteq \dl_{tt}^{\cfl}$ by induction on $k\geq1$.  Since the base case $k=1$ is trivial, we hereafter assume that $k\geq2$. Let $L$ be any language in $\dl_{tt}^{\cfl_{k}}$. Moreover, let $M$ be a log-space oracle DTM $M$ and let $A$ be an oracle in $\cfl_{k}$ for which  that $M$ $tt$-reduces $L$ to $A$. Here, we consider the case where  $k$ is even.
Let us construct a new oracle npda $N$ as follows. On input $x$, when $M$ produces $m$ query words $y_1,y_2,\ldots,y_m$, $N$ produces $2m$ query words $0y_1,1y_1,0y_2,1y_2,\ldots,0y_m,1y_m$.
Since $A\in\cfl_{k}$, take two appropriate languages $B\in\cfl_{k-1}$ and $C\in\cfl$ satisfying $A = B\cap \overline{C}$.
We define $B'=\{0y\mid y\in B\}$ and $C'=\{1y\mid y\in C\}$ and we set $A'$ to be $B'\cup C'$, which is in $\cfl_{k-1}\vee \cfl=\cfl_{k-1}$ since $\cfl_{k-1}$ is closed under union with $\cfl$. Note that $y_i$ is in $A$ iff  $0y_i\in B'$ and $1y_i\notin C'$ iff $0y_i\in A'$ and $1y_i\notin A'$ and $1y_i\notin A'$.
We use these equivalence relations as a truth table to decide the membership of $x$ to $L$. Therefore, $L$ must be in $\dl_{tt}^{\cfl_{k-1}}$. Finally, we  apply our induction hypothesis $\dl_{tt}^{\cfl_{k-1}} = \dl_{tt}^{\cfl}$ to obtain $L\in\dl_{tt}^{\cfl}$.

When $k$ is odd, since $A=B\cup C$ for certain languages $B\in\cfl_{k-1}$ and $C\in\cfl$,  it suffices to transform $B$ and $C$ to $B'=\{0y\mid y\in B\}$ and $C'=\{1y\mid y\in C\}$. The rest of the argument is similar to the previous case.
\end{proof}


Wagner \cite{Wag90} introduced a convenient notation $\thetap{k+1}$ as an abbreviation of $\p_{T}({\sigmap{k}[O(\log{n})]})$ for each level $k\geq1$, where the script ``$[O(\log{n})]$'' means that the total number of queries made in an entire {\em computation tree}
on each fixed input of size $n$ using an appropriate oracle in $\sigmap{k}$ is bounded from above by $c\log{n}+d$ for two absolute constants $c,d\geq0$.

\begin{theorem}\label{Theta-level-charact}
For all levels $k\geq1$, $\dl_{tt}(\sigmacflt{k+1}) = \thetap{k+1}$ holds.
\end{theorem}

\begin{proof}
Let $k\geq1$. First, we will give a  useful characterization of $\thetap{k+1}$ in terms of $\sigmap{k}$ using two different truth-table reductions.  In  a way similar to $\dl_{tt}^{A}$, the notation $\p_{tt}^{A}$ (or $\p_{tt}(A)$) is introduced  using polynomial-time DTMs instead of log-space DTMs as underlying reduction machines.

\begin{claim}\label{theta-level-vs-log-space}
For every index $k\in\nat^{+}$, it holds that $\thetap{k+1} = \p_{tt}(\sigmap{k}) = \dl_{tt}(\sigmap{k})$.
\end{claim}

\begin{proof}
It suffices to show that $\p_{tt}(\sigmap{k}) \subseteq \dl_{tt}(\sigmap{k})$ and $\thetap{k+1}=\p_{tt}(\sigmap{k})$, since $\dl_{tt}(\sigmap{k})\subseteq\p_{tt}(\sigmap{k})$ is obvious. Note that the proof of $\p_{tt}^{\np} \subseteq \dl_{tt}^{\np}$ by Buss and Hay \cite{BH91} \emph{relativizes};  namely, $\p_{tt}(\np^A)\subseteq \dl_{tt}(\np^A)$ for any oracle $A$. Recall the language $QBF_k$  defined in the proof of Claim \ref{sigma_CFL-to-sigma_P}.
By choosing $QBF_{k-1}$ for $A$, we obtain $\p_{tt}(\sigmap{k}) = \p_{tt}(\np^{QBF_{k-1}})\subseteq \dl_{tt}(\np^{QBF_{k-1}}) = \dl_{tt}(\sigmap{k})$.   Moreover, the proof of $\p_{T}^{\np[O(\log{n})]}=\p_{tt}^{\np}$ given in, \eg \cite{BH91}   also \emph{relativizes}.
By a similar argument as above, we can conclude that $\p_{T}(\sigmap{k}[O(\log{n})]) = \p_{tt}(\sigmap{k})$.
\end{proof}

Since we have earlier shown that $\sigmacflt{k+1}\subseteq \sigmap{k}$, it follows that $\dl_{tt}(\sigmacflt{k+1}) \subseteq \dl_{tt}(\sigmap{k}) = \thetap{k+1}$, where the last equality comes from Claim \ref{theta-level-vs-log-space}.
In what follows, we intend to argue that $\thetap{k+1} \subseteq \dl_{tt}(\sigmacflt{k+1})$.
Assume that $L\in\thetap{k+1}$; thus, $L$ is in $\dl_{tt}(\sigmap{k})$ by Claim \ref{theta-level-vs-log-space}.
Since $QBF_k$ is log-space many-one complete for $\sigmap{k}$, we obtain $\sigmap{k}\subseteq \dl_{tt}^{QBF_k}$. It thus follows that  $L\in\dl_{tt}(\sigmap{k}) \subseteq \dl_{tt}(\dl_{tt}^{QBF_k})=\dl_{tt}^{QBF_k}$. Since $QBF_k$ belongs to $\sigmacflt{k+1}$ by the proof of Claim \ref{sigma_CFL-to-sigma_P}, we conclude that $L\in \dl_{tt}(\sigmacflt{k+1})$.
\end{proof}

\section{Challenging Open Problems}\label{sec:open-problem}

As a lengthy, challenging task, we have tried throughout this paper to develop a coherent theory to discuss structural properties of formal languages associated with context-free languages by way of introducing various notions of CFL-reducibility and then conducting an initial comprehensive study on their roles and characteristics.
Our study has made it clear that, despite the use of stacks---restricted memory devices---hampering the behaviors of npda's, an oracle mechanism used for  CFL-reducibilities in fact endows enormous power in language recognition to the npda's.
In particular, we have employed Turing CFL-reductions to build the CFL hierarchy over $\cfl$ in analogy with the Meyer-Stockmeyer polynomial (time) hierarchy.

Unfortunately, numerous fundamental questions have left unsolved and a thorough study is required to answer those questions. remember that this is merely the first step toward the full understandings of the nature of the CFL hierarchy. To promote further research on this topic, this section will provide a short list of open questions for the avid reader.

\begin{enumerate}
  \setlength{\topsep}{-2mm}%
  \setlength{\itemsep}{0mm}
  \setlength{\parskip}{0cm}%

\item \emph{Comparison among different CFL-reducibilities.}
We have discussed only main stream reducibilities. Many more reducibilities were already discussed in, \eg \cite{BLS84,LLS75}. Explore those reducibilities founded on dpda's and npda's and discuss relationships among the reducibilities.

\item \emph{Relativization.}
In Section \ref{sec:BPCFL}, we have constructed oracles, which present us interesting relativized worlds. It is of great importance to develop a full theory of relativization founded  on various types of $\cfl$-reducibility. Prove that the $\cfl$ hierarchy is an infinite hierarchy relative to  random oracles and generic oracles.

\item \emph{Separation of the CFL hierarchy.}
Corollary \ref{PH-CFLH-collapse} hints that proving $\sigmacfl{k}=\sigmacfl{k+1}$, where $k\geq2$, is quite difficult, because such a collapse leads to $\ph=\sigmap{k}$. However, it may be much more approachable to tackle the separation $\sigmacfl{k}\neq\sigmacfl{k+1}$. Is it true that $\sigmacfl{k}\neq\sigmacfl{k+1}$ implies $\sigmap{k}\neq\sigmap{k+1}$ for each $k\geq2$? Is it true that  $\sigmacfl{k+2}\nsubseteq \sigmap{k}$ for any $k\geq2$?

\item \emph{Upper and lower bounds of $\sigmacfl{k}$.}
Proposition \ref{DSPACE-upper} gives a simple upper bound on the computational complexity of $\cflh$. By Claim \ref{sigma_CFL-to-sigma_P}, $\sigmap{k+1}$ is an upper bound of $\sigmacfl{k}$ for each
level $k\geq1$.
Find much tighter upper and lower bounds of $\cflh$ and each $\sigmacfl{k}$ in terms of well-known complexity classes.

\item \emph{Separation of the Boolean CFL hierarchy.}
Concerning the Boolean hierarchy over $\sigmacfl{k}$, for $k\geq2$, is it true that $\sigmacfl{k,e}=\sigmacfl{k,e+1}$ implies $\sigmacfl{k+1}=\sigmacfl{k+2}$? Is it true that $\bhcfl \neq \sigmacfl{2}\cap\picfl{2}$, $\cfl(\omega)\neq \bhcfl$, and  $\bhcfl\neq\cfl_{m}^{\bhcfl}$?

\item \emph{Lowness property.}
In Section \ref{sec:low-set}, we have briefly discussed various lowness properties of $\cfl$. Prove or disprove that the inclusions in Lemma \ref{lowness-result} are proper. Discuss the lowness property of $\sigmacfl{k}$ for any $k\geq2$.

\item \emph{Properties of $\bpcfl$ and $\pcfl$.}
Two probabilistic language families $\bpcfl$ and $\pcfl$ can be regarded as natural analogues of $\bpp$ and $\pp$ in the polynomial-time setting. Numerous properties have been known for $\bpp$ and $\pp$. Prove similar properties for $\bpcfl$ and $\pcfl$. For example, prove or disprove that $\bpcfl$ is included in $\cflh$. Is it true that $\sigmacfl{2}\subseteq \dl_{m}^{\pcfl}$?

\item \emph{Relations to circuit complexity classes.}
In Corollary \ref{NP-vs-NC2-second}, we have shown a relationship between the CFL hierarchy and the circuit complexity class $\nc{2}$. This result suggests that $\sigmacfl{2}\subseteq\nc{2}$ is unlikely, because such a containment  implies $\np=\nc{2}$.  Prove or disprove, for instance, that $\deltacfl{2}\subseteq \ac{1}$, as well as $\sigmacfl{2}\cap\picfl{2}\subseteq \ac{1}$. What is the computational complexity of, \eg $\ac{0}(\sigmacfl{k})$?

\item \emph{Logarithmic-space reductions.}
In Section \ref{sec:close-relation}, we have used logarithmic-space (or log-space) reductions to link between the $\cfl$ hierarchy and the polynomial hierarchy. Explore more relationships between the $\cfl$ hierarchy and other well-known complexity classes. Determine the complexity of, \eg $\dl_{m}^{\pcfl}$.

\item \emph{More structural properties.}
Recently, structural properties, such as simple and immune sets \cite{Yam11}, pseudorandom sets \cite{Yam09,Yam11}, and dissectible sets \cite{YK13}, have been studied for low-complexity languages. Discuss those properties for languages in the CFL hierarchy. Find other structural properties for the CFL  hierarchy.
\end{enumerate}


\let\oldbibliography\thebibliography
\renewcommand{\thebibliography}[1]{%
  \oldbibliography{#1}%
  \setlength{\itemsep}{0pt}%
}
\bibliographystyle{plain}

\begin{thebibliography}{Gur91}
{\small

\bibitem{BGS75}
T. Baker, J. Gill, and R. Solovay. Relativizations of the P=?NP question.
{\em SIAM J. Comput.} 4 (1975) 431--442.



\bibitem{Bea90}
P. W. Beame. Lower bounds for recognizing small cliques on CRCW PRAM's. {\em Discrete Applied Mathematics}, 29 (1990) 3--20.

\bibitem{Ber79}
J. Berstel. {\em Transductions and Context-Free Languages}. B. G. Teubner, Stuttgart (1979).

\bibitem{BK88}
R. Book and K. Ko. On sets truth-table reducible to sparse sets.
{\em SIAM J. Comput.} 17 (1988) 903--919.

\bibitem{BLS84}
R. Book, T. Long, and A. Selman. Quantitative relativizations of complexity classes. {\em SIAM J. Comput.} 13 (1984) 461--487.

\bibitem{BH91}
S. R. Buss and L. Hay. On the truth-table reducibility to SAT.
{\em Inf. Comput.} 91 (1991) 86--102.

\bibitem{Cai86}
J.-Y. Cai. With probability one, a random oracle separates PSPASE from the polynomial-time hierarchy. In {\em Proc. of the 18th ACM Symposium on Theory of Computing}, pp.21--29, 1986.

\bibitem{Coo71}
S. Cook. The complexity of theorem-proving procedures. In {\em Proc. of
the 3rd ACM Symposium on Thoery of Computing}, pp.151--158, 1971.

\bibitem{DK00}
D. Du and K. Ko. {\em Theory of Computational Complexity}, John Willey \& Sons, 2000.

\bibitem{GGH67}
S. Ginsburg, S. A. Greibach, and M. A. Harrison. One-way stock languages. {\em J. ACM} 14 (1967) 389--418.

\bibitem{Gre73}
S. A. Greibach. The hardest context-free language. {\em SIAM J. Comput.}
2 (1973) 304--310.


\bibitem{HMU01}
J. E. Hopcroft, R. Motwani, and J. D. Ullman. {\em Introduction to Automata Theory, Languages, and Computation}. Second Edition. Addison-Wesley,
2001.

\bibitem{Has86}
J. H{\aa}stad. {\em Computational Limitations for Small-Depth Circuits}, Ph.D. dissertation, MIT Press, Cambridge, MA, 1986.

\bibitem{HS10}
J. Hromkovi\v{c} and G. Schnitger.  On probabilistic pushdown automata. {\em Inf. Comput.} 208 (2010) 982--995.

\bibitem{Kad88}
J. Kadin. The polynomial time hierarchy collapses if the Boolean hierarchy collapses. {\em SIAM J. Comput.} 17 (1988) 1263--1282. Erratum appears in {\em SIAM J. Comput.} 20 (1991) p.404.

\bibitem{Kar72}
R. M. Karp. Reducibility among combinatorial problems. In {\em Complexity
of Computer Computations}, R. Miller and J. Thatcher (eds), Plenum Press, pp.85--103, 1972.

\bibitem{KSY07}
S. Konstantinidis, N. Santean, and S. Yu. Representation and uniformization of algebraic transductions. {\em Acta Inf.} 43 (2007) 395--417.


\bibitem{LLS75}
R. Ladner, N. Lynch, and A. Selman. A comparison of polynomial-time
reducibilities. {\em Theor. Comput. Sci.} 1 (1975) 103--123.

\bibitem{LSH65}
P. M. Lewis II, R. E. Stearns, and J. Hartmanis.
Memory bounds for recognition of context-free and context-sensitive languages. {\em In Proceedings of IEEE Conf. Record on Switching Circuit Theory and Logic Design} (1965) pp. 191–-202.

\bibitem{Lin06}
P. Linz. {\em An Introduction to Formal Languages and Automata}. Fouth Edition.
Jones and Barlett Publishers, 2006.

\bibitem{LW73}
L. Y. Liu and P. Weiner. An infinite hierarchy of intersections of context-free languages.
{\em Math. Systems Theory} 7 (1973) 185--192.

\bibitem{MO98}
I. I. Macarie and M. Ogihara. Properties of probabilistic pushdown automata.
{\em Theor. Comput. Sci.} 207 (1998) 117--130.

\bibitem{MS72}
A. R. Meyer and L. J. Stockmeyer. The equivalence problem for regular expressions with squaring requires exponential space.
In {\em Proc. of the 13th Annual IEEE Symposium on
Switching and Automata Theory}, pp.125--129, 1972.

\bibitem{Rab63}
M. O. Rabin. Probabilistic automata. {\em Inform. Control} {6}, 230--244 (1963)

\bibitem{Rei90}
K. Reinhardt. Hierarchies over the context-free languages. In {\em Proc. of the 6th International Meeting of Young Computer Scientists on Aspects and Prospects of Theoretical Computer Science} (IMYCS), Lecture Notes in Computer Science, Springer, vol.464, pp.214--224, 1990.






\bibitem{Sel94}
A. L. Selman. A taxonomy of complexity classes of functions. {\em J. Comput. System Sci.},
48 (1994) 357--381.



\bibitem{Sto77}
L. J. Stockmeyer. The polynomial-time hierarchy. {\em Theor. Comput. Sci.}  3 (1977) 1--22.

\bibitem{SM73}
L. J. Stockmeyer and A. R. Meyer. Word problems requiring exponential time. In {\em Proc. of the 5th Annual ACM Symposium on Theory of Computing}, pp.1--9, 1973.

\bibitem{Sud78}
I. H. Sudborough. On the tape complexity of deterministic context-free languages.
{\em J. ACM},  25 (1978) 405--414.

\bibitem{TYL10}
K. Tadaki, T. Yamakami, and J. C. H. Lin. Theory of one-tape linear-time Turing machines. {\em Theor. Comput. Sci.} 411 (2010) 22--43.
An extended abstract appeared in the Proc. of the 30th SOFSEM Conference on Current Trends in Theory and Practice of Computer Science (SOFSEM 2004), Lecture Notes in Computer Science, Springer, vol.2932, pp.335--348, 2004.

\bibitem{Ven91}
H. Venkateswaran. Properties that characterize LOGCFL. {\em J. Comput. System Sci.}  42 (1991) 380--404.

\bibitem{Wag90}
K. W. Wagner. Bounded query classes. {\em SIAM J. Comput.} 19 (1990) 833--846.

\bibitem{Wot78}
D. Wotschke. Nondeterminism and Boolean operations in pda's. {\em J. Comput. System Sci.} 16 (1978) 456--461.

\bibitem{Wra77}
C. Wrathall. Complete sets and the polynomial time hierarchy. {\em Theor. Comput. Sci.} 3 (1977) 23--33.

\bibitem{Yam08}
T. Yamakami. Swapping lemmas for regular and context-free languages. Available at   arXiv:0808.4122, 2008.

\bibitem{Yam09}
T. Yamakami. Pseudorandom generators against advised context-free languages.  Available at   arXiv:0902.2774, 2009.

\bibitem{Yam10}
T. Yamakami. The roles of advice to one-tape linear-time Turing machines and finite automata. {\em Int. J. Found. Comput. Sci.} 21 (2010) 941--962.
An early version appeared in the Proc. of the 20th International Symposium on Algorithms and Computation (ISAAC 2009), Lecture Notes in Computer Science, Springer, vol.5878, pp.933--942, 2009.

\bibitem{Yam11}
T. Yamakami. Immunity and pseudorandomness of context-free languages. {\em Theor. Comput. Sci.} 412 (2011) 6432--6450.

\bibitem{YK13}
T. Yamakami and Y. Kato. The dissecting power of regular languages.
{\em Inf. Process. Lett.}  113 (2013) 116--122.

\bibitem{Yao85}
A. C. Yao. Separating the polynomial-time hierarchy by oracles. In {\em Proc. of the 26 IEEE Symposium on Foundations of Computer Science}, IEEE Computer Society Press, pp.1--10, 1985.

\bibitem{You67}
D. H. Younger. Recognition and parsing of context-free languages in time $n^3$. {\em Inf. Control} 10 (1967) 189--208.

}
\end{thebibliography}

\end{document}